\pdfoutput=1
\documentclass[acmsmall,screen]{acmart}

\setcopyright{rightsretained}
\acmDOI{10.1145/3656397}
\acmYear{2024}
\copyrightyear{2024}
\acmSubmissionID{pldi24main-p115-p}
\acmJournal{PACMPL}
\acmVolume{8}
\acmNumber{PLDI}
\acmArticle{167}
\acmMonth{6}
\received{2023-11-16}
\received[accepted]{2024-03-31}

\ccsdesc[500]{Computer systems organization~Quantum computing}

\keywords{quantum programming languages, quantum compilers}

\usepackage{braket}
\usepackage{cleveref}
\usepackage{enumitem}
\usepackage{listings}
\usepackage{mathpartir}
\usepackage{mathtools}
\usepackage{multirow}
\usepackage{nicefrac}
\usepackage[defaultlines=2,all]{nowidow}
\usepackage{pgfplots}
\usepackage{rotating}
\usepackage[separate-uncertainty=true,multi-part-units=single,detect-weight=true]{siunitx}
\usepackage{stmaryrd}
\usepackage{subcaption}
\usepackage{threeparttable}
\usepackage{tikz}
\usepackage{wrapfig}
\usetikzlibrary{positioning}
\usetikzlibrary{calc}
\usetikzlibrary{quantikz2}
\usepgfplotslibrary{colorbrewer}

\crefname{lemma}{Lemma}{Lemmas}
\Crefname{lemma}{Lemma}{Lemmas}

\lstdefinelanguage{tower}
{morekeywords={
     fun,let,bool,unit,uint,ptr,type,if,else,return,with,do,not,test,null,alloc,default,static,land,lor,lnot,lsl,lsr,true,false
   },
 sensitive=true,
 morecomment=[n]{/*}{*/},
 morecomment=[l]{//},
 morestring=[b]",
 escapechar=\%,
 columns=fullflexible,
 keepspaces=true,
 basicstyle=\ttfamily,
 mathescape=true,
}
\lstdefinestyle{color}{
    keywordstyle=\color{blue!50!black}\bfseries\ttfamily,
}
\lstdefinestyle{nonum}{
    numbers=none,
}
\lstdefinestyle{tiny}{
    basicstyle=\tiny\ttfamily,
}
\lstdefinestyle{highlightif}{
    keywordstyle = [2]{\color{orange!80!black}\bfseries\ttfamily},
    deletekeywords = {if, else},
    morekeywords = [2]{if, else},
}
\makeatletter
\lst@AddToHook{EveryPar}{%
    \edef\@temp{\noexpand\label{\lst@label-\arabic{lstnumber}}}%
    \@temp
}
\makeatother

\newcommand{\tgate}{\textit{T}}

\DeclarePairedDelimiter\abs{\lvert}{\rvert}%

\newcommand{\ctx}{\ensuremath{\Gamma}}
\newcommand{\type}{\ensuremath{\tau}}
\newcommand{\tUnit}{\ensuremath{()}}

\newcommand{\tUInt}{\ensuremath{\texttt{uint}}}
\newcommand{\tBool}{\ensuremath{\texttt{bool}}}
\newcommand{\tPair}[2]{\ensuremath{({#1}, {#2})}}
\newcommand{\tPtr}[1]{\ensuremath{\texttt{ptr}({#1})}}

\newcommand{\vValue}{\ensuremath{v}}
\newcommand{\vUnit}{\ensuremath{()}}
\newcommand{\vNum}[1]{\ensuremath{\overline{#1}}}
\newcommand{\vTrue}{\ensuremath{\texttt{true}}}
\newcommand{\vFalse}{\ensuremath{\texttt{false}}}
\newcommand{\vNull}[1]{\ensuremath{\texttt{null}_{#1}}}
\newcommand{\vPtr}[2]{\ensuremath{\texttt{ptr}_{#1}[{#2}]}}

\newcommand{\eExp}{\ensuremath{e}}
\newcommand{\ePair}[2]{\ensuremath{({#1},{#2})}}
\newcommand{\eUnop}[2]{\ensuremath{{#1}\ {#2}}}
\newcommand{\eBinop}[3]{\ensuremath{{#1}\ {#2}\ {#3}}}
\newcommand{\eProj}[2]{\ensuremath{\pi_{#1}({#2})}}

\newcommand{\oNot}{\ensuremath{\texttt{not}}}
\newcommand{\oTest}{\ensuremath{\texttt{test}}}
\newcommand{\oAdd}{\ensuremath{\texttt{+}}}
\newcommand{\oMul}{\ensuremath{\texttt{*}}}
\newcommand{\oSub}{\ensuremath{\texttt{-}}}
\newcommand{\oAnd}{\ensuremath{\texttt{\&\&}}}
\newcommand{\oOr}{\ensuremath{\texttt{||}}}

\newcommand{\sStmt}{\ensuremath{s}}
\newcommand{\sSkip}{\ensuremath{\texttt{skip}}}
\newcommand{\sSeq}[2]{\ensuremath{{#1};\,{#2}}}
\newcommand{\sBind}[2]{\ensuremath{{#1} \leftarrow {#2}}}
\newcommand{\sUnbind}[2]{\ensuremath{{#1} \rightarrow {#2}}}
\newcommand{\sSwap}[2]{\ensuremath{{#1} \Leftrightarrow {#2}}}
\newcommand{\sMemSwap}[2]{\ensuremath{*{#1} \Leftrightarrow {#2}}}
\newcommand{\sIf}[2]{\ensuremath{\texttt{if}\ {#1}\ \{\ {#2}\ \}}}
\newcommand{\sWith}[2]{\ensuremath{\texttt{with}\ \{\ {#1}\ \}\ \texttt{do}\ \{\ {#2}\ \}}}
\newcommand{\sHadamard}[1]{\ensuremath{H({#1})}}

\newcommand{\hastype}[3]{\ensuremath{{#1} \vdash {#2} : {#3}}}
\newcommand{\stmtok}[3]{\ensuremath{{#1} \vdash {#2} \dashv {#3}}}
\newcommand{\modified}[1]{\ensuremath{\textsf{mod}({#1})}}

\newcommand{\reg}{\ensuremath{R}}
\newcommand{\mem}{\ensuremath{M}}

\newcommand{\denote}[1]{\ensuremath{\left\llbracket{#1}\right\rrbracket}}
\newcommand{\reverse}[1]{\mathcal{I}[{#1}]}
\newcommand{\circuit}[1]{\mathcal{C}\!\denote{#1}}

\makeatletter
\newcommand{\oset}[3][0ex]{%
  \mathrel{\mathop{#3}\limits^{
    \vbox to#1{\kern-0.75\ex@
    \hbox{$\scriptstyle#2$}\vss}}}}
\makeatother

\newcommand{\domain}[1]{\ensuremath{\mathrm{dom}\,{#1}}}

\newcommand{\mcxComplexity}[2][]{\ensuremath{C^{\textrm{MCX}}_{#1}(#2)}}
\newcommand{\tComplexity}[2][]{\ensuremath{C^{\tgate}_{#1}(#2)}}

\begin{document}

\author{Charles Yuan}
\affiliation{
  \institution{MIT CSAIL}
  \streetaddress{32 Vassar St}
  \city{Cambridge}
  \state{MA}
  \postcode{02139}
  \country{USA}
}
\email{chenhuiy@csail.mit.edu}
\orcid{0000-0002-4918-4467}

\author{Michael Carbin}
\affiliation{
  \institution{MIT CSAIL}
  \streetaddress{32 Vassar St}
  \city{Cambridge}
  \state{MA}
  \postcode{02139}
  \country{USA}
}
\email{mcarbin@csail.mit.edu}
\orcid{0000-0002-6928-0456}

\title{The \tgate{}-Complexity Costs of Error Correction for Control Flow in Quantum Computation}

\begin{abstract}
Numerous quantum algorithms require the use of quantum error correction to overcome the intrinsic unreliability of physical qubits.
However, quantum error correction imposes a unique performance bottleneck, known as \tgate{}-complexity, that can make an implementation of an algorithm as a quantum program run more slowly than on idealized hardware.
In this work, we identify that programming abstractions for control flow, such as the quantum \texttt{if}-statement, can introduce polynomial increases in the \tgate{}-complexity of a program. If not mitigated, this slowdown can diminish the computational advantage of a quantum algorithm.

To enable reasoning about the costs of control flow, we present a cost model that a developer can use to accurately analyze the \tgate{}-complexity of a program under quantum error correction and pinpoint the sources of slowdown.
To enable the mitigation of these costs, we present a set of program-level optimizations that a developer can use to rewrite a program to reduce its \tgate{}-complexity, predict the \tgate{}-complexity of the optimized program using the cost model, and then compile it to an efficient circuit via a straightforward strategy.

We implement the program-level optimizations in Spire, an extension of the Tower quantum compiler.
Using a set of 11 benchmark programs that use control flow, we empirically show that the cost model is accurate, and that Spire's optimizations recover programs that are asymptotically efficient, meaning their runtime \tgate{}-complexity under error correction is equal to their time complexity on idealized hardware.

Our results show that optimizing a program before it is compiled to a circuit can yield better results than compiling the program to an inefficient circuit and then invoking a quantum circuit optimizer found in prior work. For our benchmarks, only 2 of 8 tested quantum circuit optimizers recover circuits with asymptotically efficient \tgate{}-complexity. Compared to these 2 optimizers, Spire uses $54\times$--$2400\times$ less compile time.
\end{abstract}

\maketitle

\section{Introduction} \label{sec:introduction}

\emph{Quantum algorithms} promise computational advantage over classical algorithms across numerous domains, including cryptography and communication~\citep{shor1997,proos2003,bennett1993,bennett2014}, search and optimization~\citep{grover1996,farhi2014}, data analysis and machine learning~\citep{rebentrost2018,biamonte2017,lloyd2015}, and physical simulation~\citep{abrams1997,childs2018,babbush2018}.

The power of quantum algorithms is rooted in their ability to manipulate quantum information, which exists in a \emph{superposition} of weighted classical states. A quantum computer may use \emph{quantum logic gates} to modify the states and weights within a superposition, and \emph{measure} quantum data to obtain a classical outcome with probability determined by the weights in the superposition.

A common representation of a quantum algorithm is as a \emph{quantum circuit}, a sequence of quantum logic gates that operate over individual \emph{qubits}, which are the quantum analogue of bits.

\paragraph{Error Correction}
Whereas an idealized quantum computer can execute any quantum algorithm, a realistic device must contend with the fact that every known physical implementation of a qubit is unreliable, meaning its state becomes irreversibly corrupted after a small number of logic gates are performed. To execute the algorithms that possess a provable asymptotic advantage in time complexity, including~\citet{shor1997,grover1996}, a quantum computer must employ \emph{quantum error correction} to encode a reliable \emph{logical} qubit within a number of unreliable physical qubits.

\paragraph{Resource Estimation}
A quantum algorithm that executes more logic gates requires logical qubits to be more reliable, which in turn demands more physical qubits, a scarce hardware resource.
Given an algorithm, it is thus essential to determine its time complexity in terms of the number of logic gates that it executes. This task, known as \emph{resource estimation}~\citep{suchara2013,hoefler2023,leymann2020}, is key to recognizing the scale of hardware needed to execute a quantum algorithm and the problem size at which it offers advantage over classical algorithms.

\paragraph{\tgate{}-Complexity Bottleneck}
In principle, conducting resource estimation for a quantum algorithm involves simply writing it out as a circuit and counting the number of logic gates used. A practical challenge is that quantum error correction affects the available logic gates and their costs.

An idealized quantum computer supports any physically realizable quantum logic gate, including analogues of classical NAND gates known as \emph{multiply-controlled NOT} (MCX) gates that are necessary for arithmetic~\citep{rines2018,gidney2021} and memory~\citep{low2018} within a quantum algorithm. By contrast, the prevailing \emph{surface code}~\citep{fowler2012} architecture for error correction that has been implemented in practice by quantum hardware from Google~\citep{google2023} and IBM~\citep{takita2016} supports only a restricted set of gates known as the \emph{Clifford+\tgate{}} gates, into which all MCX gates must be decomposed.

In turn, the decomposition of MCX uses the single-qubit \tgate{} gate, a performance bottleneck on the surface code.
Unlike \emph{Clifford} gates such as NOT and the two-qubit controlled-NOT (CNOT) that are natively supported by the code, the \tgate{} gate is realized separately via \emph{magic state distillation}~\citep{bravyi2005} at an area-latency cost\footnote{By area-latency cost, we refer to the product of the number of qubits and the number of processing cycles of the device.} of about $10^2$ times that of a CNOT gate~\citep{gidney2019} and $10^{10}$ times that of a NAND gate in classical transistors~\citep{babbush2021}.

Although deriving an efficient quantum circuit requires navigating all of its gate costs, the order of magnitude increase in cost for \tgate{} gates relative to other gates has contributed to a contemporary consensus that ``The number of \tgate{} gates \ldots{} typically dominates the cost when implementing a fault tolerant algorithm''~\citep{reiher2017}. It is therefore broad practice to quantify the runtime cost of a quantum algorithm under error correction using its \emph{\tgate{}-complexity}~\citep{babbush2018}, i.e.\ number of \tgate{} gates, which is often greater than its number of MCX gates.

\subsection{\tgate{}-Complexity Costs of Control Flow in Quantum Programs}

Resource estimation is made even more challenging by the reality that it is often impractical for a developer to explicitly write quantum circuits by hand. Instead, the developer uses \emph{quantum programming languages}~\citep{qsharp,quipper,qwire,selinger2004}, which provide programming abstractions over quantum data that are ultimately compiled to circuits.

\paragraph{Control Flow}
One abstraction provided by languages~\citep{qml,silq,javadi2014,voichick2023,tower,ying2012} is a quantum \texttt{if}-statement that conditions on the value of qubit in superposition. This concept of \emph{control flow in superposition} enables algorithms for simulation~\citep{babbush2018}, factoring~\citep{shor1997}, and search~\citep{ambainis2004} to be expressed as a program more concisely than without the abstraction.
In turn, the language compiler produces a circuit for a program that utilizes quantum \texttt{if} by translating the abstraction into individual qubit-controlled logic gates.

\paragraph{Costs of Control Flow}
The problem is that quantum error correction can make the use of control flow abstractions significantly more inefficient than on idealized hardware.
In this work, we identify that the usage of programming abstractions for control flow in superposition, such as quantum \texttt{if}, in a program can cause its asymptotic \tgate{}-complexity to be polynomially larger than the time complexity found by a standard analysis that assumes idealized hardware. This blowup arises because a quantum \texttt{if} compiles to significantly more \tgate{} gates than it does to MCX gates.

In turn, a polynomial increase in \tgate{}-complexity diminishes the theoretical advantage of quantum algorithms for tasks such as search~\citep{grover1996,brassard2002} and optimization~\citep{sanders2020} that have only polynomial advantage over classical algorithms. Moreover, emerging evidence suggests that ``at least cubic or quartic speedups are required for a practical quantum advantage''~\citep{hoefler2023}, which implies that even a linear slowdown jeopardizes the practical advantage of an algorithm that is otherwise marginally over the cubic speedup threshold.

\subsection{Cost Model for Accurately Predicting \tgate{}-Complexity Costs}

To optimize away this overhead, the developer must pinpoint the specific locations in a quantum program that incur the overhead.
A challenge is that without the ability to accurately reason about the program at syntax level, the developer must repeatedly compile it to a large circuit and count its gates, which does not efficiently or precisely identify the cause of the slowdown.

As an alternative, we present a cost model for reasoning about the \tgate{}-complexity of programming abstractions for control flow in superposition within a quantum program.
Using the cost model, a developer can pinpoint the sources of slowdown through a syntax-level analysis that accurately determines the runtime cost of each program statement under quantum error correction.

\subsection{Program-Level Optimizations for Mitigating \tgate{}-Complexity Costs}

Next, we present a set of \emph{program-level optimizations} for quantum programs.
Using them, a developer can rewrite a program to reduce its \tgate{}-complexity, predict the \tgate{}-complexity of the optimized program using the cost model, and then compile it to an efficient circuit by a straightforward strategy.

The first optimization, \emph{conditional flattening}, identifies excess \tgate{} gates caused by nested quantum \texttt{if}-statements, and removes these gates by introducing a temporary qubit and using it to flatten the structure of conditional statements. The second optimization, \emph{conditional narrowing}, identifies excess \tgate{} gates caused by statements that do not need to be placed under a quantum \texttt{if} and safely moves these statements outside the \texttt{if}, thereby narrowing the range of conditional statements.

We implement and evaluate these optimizations in Spire, an extension of the Tower~\citep{tower} quantum compiler.
For a set of 11 benchmark programs that use control flow, Spire successfully recovers programs that are asymptotically efficient, meaning their \tgate{}-complexity under error correction is equal to their time complexity on idealized hardware.

Our results show that optimizing a program before it is compiled to a circuit can yield better results than compiling the program to an inefficient circuit and then invoking a quantum circuit optimizer found in prior work. For our benchmarks, a majority of existing optimizers we tested do not recover circuits that are asymptotically efficient in \tgate{}-complexity. Moreover, Spire's optimizations followed by an existing quantum circuit optimizer achieve better results than either approach alone.

Forms of conditional narrowing and flattening appear in prior work~\citep{ittah2022,steiger2018,seidel2022}. Our novel contributions are to unify both optimizations as program rewrite rules, identify that they can mitigate the asymptotic slowdown caused by control flow, and empirically evaluate their effectiveness and speed relative to existing circuit optimizers.

\subsection{Contributions}

In this work, we present the following contributions:
\begin{itemize}[leftmargin=5.5mm]
\item \emph{Costs of Control Flow} (\Cref{sec:example}). We identify that programming abstractions for control flow in superposition can introduce polynomial overheads in the \tgate{}-complexity of a program. These costs can diminish the advantage of a quantum algorithm under error correction.
\item \emph{Cost Model} (\Cref{sec:cost-model}). We present a cost model that computes the \tgate{}-complexity of a quantum program that utilizes control flow. Using the cost model, the developer can accurately analyze the runtime cost of a program under an error-corrected quantum architecture.
\item \emph{Program-Level Optimizations} (\Cref{sec:optimizations}). We present two optimizations for quantum programs, \emph{conditional flattening} and \emph{conditional narrowing}. Using them, a developer can rewrite a program to reduce its \tgate{}-complexity, predict the \tgate{}-complexity of the optimized program using the cost model, and compile that program to an efficient circuit via a straightforward strategy.
\item \emph{Evaluation} (\Cref{sec:implementation,sec:evaluation}).
We implement the optimizations in Spire, an extension of the Tower quantum compiler.
Using a set of 11 benchmark programs that contain control flow, we empirically show that the cost model is accurate, and that Spire's optimizations can mitigate the \tgate{}-complexity costs of control flow and recover an asymptotically efficient program. By contrast, only 2 of 8 existing quantum circuit optimizers we tested recover circuits with asymptotically efficient \tgate{}-complexity. Spire uses $54\times$--$2400\times$ less compile time than these 2 optimizers.
\end{itemize}

\paragraph{Implications}
This work reveals challenges that must be overcome to fully realize the asymptotic advantage of quantum algorithms on an error-corrected quantum computer. By incorporating our cost model and optimizations, program optimizers may more precisely account for the architectural costs of error correction and the abstraction costs of control flow in a quantum program.

\section{Background on Quantum Computation} \label{sec:background}

This section overviews key concepts in quantum computation that are relevant to this work. For a comprehensive reference, please see~\citet{nielsen_chuang_2010}.

\paragraph{Superposition}
The fundamental unit of quantum information is the \emph{qubit}, a linear combination or \emph{superposition} $\gamma_0 \ket{0} + \gamma_1 \ket{1}$ of the classical \emph{basis states} 0 and 1, in which $\gamma_0, \gamma_1 \in \mathbb{C}$ are complex \emph{amplitudes} satisfying $\abs{\gamma_0}^2 + \abs{\gamma_1}^2 = 1$ describing relative weights of basis states. Examples of qubits are classical $\ket{0}$ and $\ket{1}$, and the states $\frac{1}{\sqrt{2}}{(\ket{0} + e^{i\varphi} \ket{1})}$ where $\varphi \in [0, 2\pi)$ is known as a \emph{phase}.

More generally, a \emph{quantum state} $\ket{\psi}$ is a superposition over $n$-bit strings. For example, $\ket{\psi} = \smash{\frac{1}{\sqrt{2}}}{(\ket{00}+\ket{11})}$ is a quantum state over two qubits.
Formally, multiple component states form a composite state by the \emph{tensor product} $\otimes$, e.g.\ the state $\ket{01}$ is equal to $\ket{0} \otimes \ket{1}$. As is customary in quantum computation, we also use the notation $\ket{0, 1}$ to represent $\ket{01} = \ket{0} \otimes \ket{1}$.

\paragraph{Unitary Operator}
A \emph{unitary operator} $U$ is a linear operator on quantum states that preserves inner products and whose inverse is its Hermitian adjoint $U^\dagger$.
Formally, a unitary operator may be constructed as a circuit of \emph{quantum gates}. The quantum gates over a single qubit include:
\begin{itemize}
    \item Bit flip ($X$ or NOT), which maps $\ket{x} \mapsto \ket{1 - x}$ for $x \in \{0, 1\}$;
    \item Phase flip ($Z$), which maps $\ket{x} \mapsto (-1){}^x\ket{x}$;
    \item $\pi/4$ phase rotation (\tgate{}), which maps $\ket{x} \mapsto e^{ix\pi/4}\ket{x}$;
    \item Hadamard ($H$), which maps $\ket{x} \mapsto \frac{1}{\sqrt{2}}{(\ket{0} + (-1){}^x \ket{1})}$.
\end{itemize}
A gate may be \emph{controlled} by one or more qubits, forming a larger unitary operator. For example, the two-qubit CNOT gate maps $\ket{0, x} \mapsto \ket{0, x}$ and $\ket{1, x} \mapsto \ket{1, \textrm{NOT}\ x} = \ket{1, 1-x}$. Generalized versions of CNOT are known as multi-controlled-$X$ (MCX) gates. In particular, the MCX gate with two control bits is known as the Toffoli gate, and the MCX with zero controls is the NOT gate.

The \emph{Clifford} gates~\citep{gottesman1998} are the quantum gates that can be constructed by compositions and tensor products of $H$, $S = T^2$, and CNOT\@. Examples of Clifford gates include $Z = S^2$ and $X = HZH$. By contrast, no MCX gate larger than CNOT is Clifford, and constructing e.g.\ a Toffoli gate requires the use of the non-Clifford \tgate{} gate. The Clifford gates plus the \tgate{} gate form the Clifford+\tgate{} gates, the gate set of the predominant surface code for quantum error correction.

\paragraph{Measurement}
Performing a \emph{measurement} of a quantum state probabilistically collapses its superposition into a classical outcome.
When a qubit $\gamma_0 \ket{0} + \gamma_1 \ket{1}$ is measured in the standard basis, the observed classical outcome is $0$ with probability $\abs{\gamma_0}^2$ and $1$ with probability $\abs{\gamma_1}^2$.

\paragraph{Entanglement}
A state is \emph{entangled} when it consists of two components but cannot be written as a tensor product of its components.
For example, the \emph{Bell state}~\citep{bell1964} $\smash{\frac{1}{\sqrt{2}}}(\ket{00} + \ket{11})$ is entangled, as it cannot be written as a product of two independent qubits.

Given an entangled state, measuring one of its components causes the superposition of the other component to also collapse.
For example, measuring the second qubit in the Bell state causes the state of the first qubit to also collapse, to either $\ket{0}$ or $\ket{1}$ with probability $\big\lvert\frac{1}{\sqrt{2}}\big\rvert{}^2 = \frac{1}{2}$ each.

\paragraph{Uncomputation}
Entanglement means that in general, measuring or discarding a component of a quantum state can destroy the superposition of the remainder of the state. The consequence is that a quantum algorithm may not simply discard a temporary variable that it no longer needs, which could cause a superposition collapse that negates the possibility of quantum advantage. Instead, the algorithm must \emph{uncompute} the variable~\citep{bennett1973,silq}, meaning it reverses the sequence of operations on that variable and returns it to its initial value of zero.

\section{\tgate{}-Complexity Costs of Control Flow in Quantum Programs} \label{sec:example}

In this section, we demonstrate how programming abstractions for control flow in superposition can cause a quantum program to have asymptotic time complexity worse than that found by its idealized theoretical analysis. These costs arise from the performance bottleneck of quantum error correction, and if not mitigated, can diminish the computational advantage of quantum algorithms.

\subsection{Running Example}

Quantum algorithms for search~\citep{ambainis2004,grover1996}, game tree evaluation~\citep{ambainis2010,childs2007}, combinatorial optimization~\citep{bernstein2013}, and geometry~\citep{aaronson2019} utilize abstract data structures in superposition to achieve computational advantage. For example, they rely on a set to efficiently maintain a collection of items, check an item for membership, and add and remove items. In turn, an abstract set can be concretely implemented as a linked list, whose structure and contents exist in quantum superposition.

In \Cref{fig:length-tower}, we present a program in the language Tower~\citep{tower} to compute the length of a linked list. This \texttt{length} function accepts a pointer \texttt{xs} to the head of the list and a value \texttt{acc} that stores the number of list nodes traversed so far. Line~\ref{line:length-tower-comparison} checks whether the list is empty, meaning that \texttt{xs} is null, and if so returns \texttt{acc}. If not, line~\ref{line:length-tower-dereference} dereferences \texttt{xs} to obtain the pointer to the next list node, and line~\ref{line:length-tower-addition} adds 1 to the value of \texttt{acc}.

\paragraph{Recursion}
\begin{wrapfigure}[16]{r}{.46\textwidth}
\begin{lstlisting}
type list = (uint, ptr<list>);
fun length[n](xs: ptr<list>, acc: uint) {
    with {%\label{line:length-tower-begin-with}%
        let is_empty <- xs == null;%\label{line:length-tower-comparison}%
    } do if is_empty {%\label{line:length-tower-begin-if}%
        let out <- acc;%\label{line:length-tower-out}%
    } else with {
        let temp <- default<list>;
        *xs <-> temp;%\label{line:length-tower-dereference}%
        let next <- temp.2;
        let r <- acc + 1;%\label{line:length-tower-addition}%
    } do {%\label{line:length-tower-begin-do}%
        let out <- length[n-1](next, r);%\label{line:length-tower-recursion}%
    }%\label{line:length-tower-end-if}%
    return out;
}
\end{lstlisting}
\setlength{\abovecaptionskip}{3pt}
\caption{Program computing the length of a list.} \label{fig:length-tower}
\end{wrapfigure}%
On line~\ref{line:length-tower-recursion}, the function makes a recursive call.
In Tower, all function calls are inlined by the compiler, and the values \texttt{n} and \texttt{n-1} statically instruct the compiler to unroll \texttt{length} to depth \texttt{n}.
In the example, \texttt{n} is a concrete integer known at compile time, and \texttt{length} is effectively a family of functions whose \texttt{n}th instance returns the length of the list \texttt{xs} if it is less than \texttt{n}, or 0 otherwise.

\paragraph{Quantum Data}
All data types in Tower denote data in quantum superposition. For example, when \texttt{xs} is a superposition of lists $[], [1]$, and $[1,2,3]$, the output of \texttt{length} is a superposition of the integers 0, 1, and 3. The \texttt{if}-statement on line~\ref{line:length-tower-begin-if} conditions on the value of \lstinline{is_empty} in superposition, meaning it executes the \texttt{if}-clause on the classical states in the machine state superposition where \lstinline{is_empty} is true, and the \texttt{else}-clause on all other states.

\paragraph{Uncomputation}
The structure of the program in \Cref{fig:length-tower} enables the use of uncomputation (\Cref{sec:background}) to clean up temporary values in the program.
In Tower, an operator known as \emph{un-assignment} \lstinline[style=color]{let x -> e} is defined as the reverse of the assignment \lstinline[style=color]{let x <- e}. Whereas assignment initializes \texttt{x} to zero and sets it to \texttt{e}, un-assignment resets \texttt{x} from \texttt{e} back to zero and deinitializes \texttt{x}.

In \Cref{fig:length-tower}, un-assignment does not appear explicitly but is performed implicitly by the \texttt{with}-\texttt{do} construct as follows.
First, the \texttt{with}-block on lines~\ref{line:length-tower-begin-with} to~\ref{line:length-tower-begin-if} executes, initializing a variable \lstinline{is_empty}. The \texttt{do}-block on lines~\ref{line:length-tower-begin-if} to~\ref{line:length-tower-end-if} then executes, computing the result \texttt{out}. Then, the \texttt{with}-block is executed in reverse, with all assignments flipped to un-assignments and vice versa. That is, the inverse of lines~\ref{line:length-tower-begin-with} to~\ref{line:length-tower-begin-if} un-assigns and deinitializes \lstinline{is_empty}. The function then returns \texttt{out}.

Equivalents of the \texttt{with}-\texttt{do} construct can be found in other quantum programming languages. Examples include the \texttt{within}-\texttt{apply} blocks of Q\#~\citep{qsharp} and the automatic uncomputation of scoped variables in Silq~\citep{silq} and Qunity~\citep{voichick2023}.

\subsection{Complexity Analysis and Diminished Advantage} \label{sec:complexity-analysis}

Algorithms such as~\citet{ambainis2004,grover1996,bernstein2013,aaronson2019} offer theoretical advantage over classical algorithms that is sub-exponential, meaning it could be diminished if their implementation as a program introduces additional polynomial overhead.

\paragraph{Idealized Analysis}
A standard analysis of \texttt{length} reveals that its time complexity is $O(n)$ where $n$ is the recursion depth \texttt{n} from above.
In this work, we assume that the bit width of integer and pointer registers is a small constant, with only the depth $n$ of recursion considered a variable.\footnotemark{}
At each of $n$ levels of recursion, \texttt{length} performs $O(1)$ work in primitive operations, and makes one recursive call. The recurrence $C(n) = O(1) + C(n-1)$ for time complexity yields $C(n) = O(n)$.%
\footnotetext{For detailed discussion of the effect of a variable bit width on the \tgate{}-complexity of a program, please see \Cref{sec:bitwidth}.}

In \Cref{fig:slowdown}, we plot the empirical time complexity of \texttt{length} on an idealized quantum computer, as determined by compiling \Cref{fig:length-tower} to a quantum circuit of multiply-controlled NOT (MCX) gates, and counting its \emph{MCX-complexity}, i.e.\ number of MCX gates, which is $O(n)$ as above.

\paragraph{Asymptotic Slowdown}
\begin{wrapfigure}[14]{r}{.46\textwidth}
\begin{tikzpicture}
\begin{axis}[
    xlabel=Recursion depth $n$ of \texttt{length},
    ylabel=Number of gates,
    yticklabel style={
        /pgf/number format/.cd,sci,
        /pgf/number format/precision=5
    },
    y label style={at={(axis description cs:0.025,.5)},anchor=south},
    scaled y ticks=false,
    ytick distance=5e5,
    xtick distance=1,
    width=6.25cm,
    legend pos=north west,
    legend cell align=left,
    legend style={fill=none,draw=none,nodes={scale=0.8, transform shape}},
    domain=2:10,
    samples=100,
    label style={font=\small},
    tick label style={font=\small}
]
\addplot[Dark2-A,-,thick,mark=square*] plot coordinates {
    (2,105406)
    (3,203308)
    (4,332654)
    (5,493444)
    (6,685678)
    (7,909356)
    (8,1164478)
    (9,1451044)
    (10,1769054)
};
\addlegendentry{\tgate{}-complexity}
\addplot[dashed,thick,mark=*,mark options=solid] plot coordinates {
    (2,7265)
    (3,10879)
    (4,14493)
    (5,18107)
    (6,21721)
    (7,25335)
    (8,28949)
    (9,32563)
    (10,36177)
};
\addlegendentry{MCX-complexity}
\end{axis}
\end{tikzpicture}
\setlength{\abovecaptionskip}{5pt}
\caption{Number of gates in the circuit of \Cref{fig:length-tower}.} \label{fig:slowdown}
\end{wrapfigure}
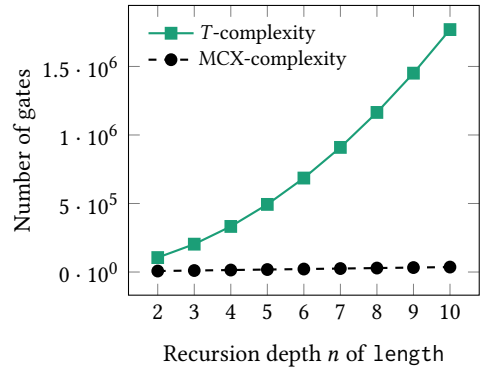
\Cref{fig:slowdown} also plots the empirical time complexity of \texttt{length} on a quantum computer with error correction, as found by compiling it to a circuit in the Clifford+\tgate{} gate set, and counting its \emph{\tgate{}-complexity}, i.e.\ number of \tgate{} gates.

The number of \tgate{} gates is an appropriate metric because \tgate{} gates act as the bottleneck of the \emph{surface code}~\citep{fowler2012}, the prevailing quantum error correction code. On the surface code, realizing the \tgate{} gate incurs an area-latency cost of about $10^2$ times that of Clifford gates such as CNOT~\citep{gidney2019} and $10^{10}$ times that of a NAND gate in classical transistors~\citep{babbush2021}.

As seen in \Cref{fig:slowdown}, the \tgate{}-complexity of the program is not $O(n)$ but rather $O(n^2)$, meaning that a quantum algorithm that invokes \texttt{length} obtains diminished advantage under error correction. Such a slowdown does not fully erase the theoretical advantage of~\citet{ambainis2004}, which is $O(N^{1/3})$ where $N$ is the size of the input. In~\citet{ambainis2004}, the depth of the data structure and hence of recursion is only poly-logarithmic in $N$, i.e. $O(\log^c N)$ for some constant $c$. However, such a slowdown jeopardizes instances of quantum search~\citep{grover1996} in which the advantage is $O(N^{1/2})$ and the depth $n$ of each query is $O(N^{1/2})$ or greater.

\subsection{\tgate{}-Complexity Costs of Control Flow} \label{sec:compilation}

The cause of the disparity is that on an error-corrected quantum computer, logic gates that are controlled by more bits are more costly to realize in terms of \tgate{}-complexity. In turn, these control bits accrue in the compiled form of a control flow abstraction such as the quantum \texttt{if}-statement.

\paragraph{Compilation of Control Flow}
To execute on a quantum computer, a program such as \Cref{fig:length-tower} is compiled to a quantum circuit, a fixed sequence of logic gates controlled by individual bits. Each statement compiles to gates controlled by all of the qubits that lead to that control flow path.

To demonstrate this translation on a smaller scale, in \Cref{fig:simple-if} we depict a simple program that uses quantum \texttt{if}-statements.
Given Booleans \texttt{x}, \texttt{y}, and \texttt{z}, the program sets the value of output variables \texttt{a} and \texttt{b} to the negation of \texttt{z} and true respectively, when \texttt{x}, \texttt{y}, and \texttt{z} are all true. Though a toy example, this program exemplifies the same overheads of control flow as in \Cref{fig:length-tower}.

In \Cref{fig:simple-if-circuit-full}, we depict the circuit to which \Cref{fig:simple-if} compiles, which has wires labeled with the name of each program variable.
Gates labeled \textit{X} denote NOT gates, while gates with black dots denote bit-controlled gates that execute only if all control bits, denoted by the dots, are true.

The nested quantum \texttt{if}-statements on lines~\ref{line:simple-if-x} and~\ref{line:simple-if-y} compile to a sequence of gates controlled by both \texttt{x} and \texttt{y}.
Line~\ref{line:simple-if-1} compiles to the first gate, a controlled-NOT (CNOT) gate that flips \texttt{t} based on the value of \texttt{z}. In turn, this CNOT is controlled by \texttt{x} and \texttt{y}.
Next, the quantum \texttt{if} on line~\ref{line:simple-if-z} compiles to gates controlled by \texttt{x}, \texttt{y}, and \texttt{z}.
Line~\ref{line:simple-if-2} compiles to three gates --- a CNOT over \texttt{t} and \texttt{a}, surrounded by NOT gates on \texttt{t}.
Line~\ref{line:simple-if-3} compiles to the next gate, a NOT over \texttt{b}. Finally, the semantics of the \texttt{with}-block states that line~\ref{line:simple-if-1} is reversed after the \texttt{do}-block, corresponding to the last gate.

\paragraph{Error Correction}
If we were targeting an ideal quantum computer not constrained by hardware, then the circuit in \Cref{fig:simple-if-circuit-full} consisting of MCX gates could serve as the final representation of the program. Indeed, the idealized analysis of \texttt{length} finds its MCX-complexity, which is linear.

By contrast, a computer that uses the surface code~\citep{fowler2012} for error correction supports the restricted Clifford+\tgate{} gate set, to which MCX gates larger than CNOT must be decomposed. In \Cref{fig:mcx-decomposition}, we depict how an MCX gate decomposes into Toffoli gates by the process of~\citet{barenco1995}. Then, in \Cref{fig:toffoli-decomposition}, we depict how Toffoli decomposes into Clifford+\tgate{} gates.

\newcommand{\notgate}{X}
\newcommand{\redctrl}[1]{\ctrl[style={orange!80!black,fill=orange!80!black},wire style=black]{#1}}
\begin{figure}[!t]
\begin{minipage}{0.33\textwidth}
\vspace*{-0.75em}%
\begin{center}
\begin{tabular}{c}
\begin{lstlisting}
if x {%\label{line:simple-if-x}%
  if y {%\label{line:simple-if-y}%
    with {
      let t <- z;%\label{line:simple-if-1}%
    } do {
        if z {%\label{line:simple-if-z}%
          let a <- not t;%\label{line:simple-if-2}%
          let b <- true;%\label{line:simple-if-3}%
} } } }%\label{line:simple-if-end}%
\end{lstlisting}
\end{tabular}
\end{center}
\end{minipage}%
\hspace*{\fill}%
\begin{minipage}{0.65\textwidth}
\vspace*{-0.5em}%
\resizebox{\textwidth}{!}{
\newcommand{\grpa}{\gategroup[2,style={dashed,inner xsep=2pt,inner ysep=2pt},background,label style={label position=above,anchor=north,yshift=1.5cm}]{\lstinline[style=color]{let t <- z;}}}
\newcommand{\grpb}{\gategroup[2,steps=3,style={dashed,inner xsep=2pt,inner ysep=2pt},background,label style={label position=above,anchor=north,yshift=2.25cm}]{\lstinline[style=color]{let a <- not t;}}}
\newcommand{\grpc}{\gategroup[1,style={dashed,inner xsep=2pt,inner ysep=2pt},background,label style={label position=above,anchor=north,yshift=3.75cm}]{\lstinline[style=color]{let b <- true;}}}
\newcommand{\grpd}{\gategroup[2,style={dashed,inner xsep=2pt,inner ysep=2pt},background,label style={label position=above,anchor=north,yshift=1.5cm}]{\lstinline[style=color]{let t -> z;}}}
\begin{quantikz}[row sep={0.75cm,between origins},column sep={0.75cm,between origins}]
\lstick{\texttt{x}} & \redctrl{2}     &[1.25cm] \redctrl{3}  & \redctrl{4}     & \redctrl{2}     &[1.5cm] \redctrl{5}   &[1.8cm] \redctrl{3} & \ghost{X} \\
\lstick{\texttt{y}} & \redctrl{2}     & \redctrl{2}          & \redctrl{3}     & \redctrl{1}     & \redctrl{4}          & \redctrl{2}        & \ghost{X} \\
\lstick{\texttt{z}} & \ctrl{1}\grpa   & \ctrl{1}             & \redctrl{2}     & \ctrl{1}        & \ctrl{3}             & \ctrl{1}\grpd      & \ghost{X} \\
\lstick{\texttt{t}} & \gate{\notgate} & \gate{\notgate}\grpb & \ctrl{1}        & \gate{\notgate} &                      & \gate{\notgate}    &           \\
\lstick{\texttt{a}} &                 &                      & \gate{\notgate} &                 &                      &                    &           \\
\lstick{\texttt{b}} &                 &                      &                 &                 & \gate{\notgate}\grpc &                    &
\end{quantikz}
}%
\end{minipage}
\setlength{\abovecaptionskip}{2pt}
\setlength{\belowcaptionskip}{-8pt}
\begin{minipage}[t]{0.31\textwidth}
\caption{Tower program that uses nested quantum \texttt{if}-statements.} \label{fig:simple-if}
\end{minipage}%
\hspace{1em}%
\begin{minipage}[t]{0.64\textwidth}
\caption{Translation of \Cref{fig:simple-if} to a circuit. On each multiply-controlled-NOT (MCX) gate, each \textcolor{orange!80!black}{orange} control bit incurs \tgate{}-complexity.} \label{fig:simple-if-circuit-full}
\end{minipage}
\end{figure}%
\begin{figure}[!t]
\begin{minipage}{0.34\textwidth}
\resizebox{\textwidth}{!}{
\begin{quantikz}[row sep={0.75cm,between origins}]
\ghost{T} \\
& \ctrl[style={orange!80!black,fill=orange!80!black},wire style=black]{3} & \\
& \ctrl[style={orange!80!black,fill=orange!80!black},wire style=black]{2} & \\
& \ctrl{1} & \\
& \gate{\notgate} &
\end{quantikz}
\hspace*{0.4em}\raisebox{-0.3em}{\huge =}\hspace*{-1.2em}
\begin{quantikz}[column sep=0.2cm,row sep={0.75cm,between origins}]
\lstick{$\ket{0}$} & \gate{\notgate} & \ctrl[style={orange!80!black,fill=orange!80!black},wire style=black]{4} & \gate{\notgate} & \rstick{$\ket{0}$} \\
& \ctrl{-1} & & \ctrl{-1} & \\
& \ctrl[style={orange!80!black,fill=orange!80!black},wire style=black]{-2} & & \ctrl[style={orange!80!black,fill=orange!80!black},wire style=black]{-2} & \\
& & \ctrl{1} & & \\
& & \gate{\notgate} & &
\end{quantikz}
}%
\end{minipage}%
\hspace*{\fill}%
\begin{minipage}{0.66\textwidth}
\resizebox{\textwidth}{!}{
\begin{quantikz}[row sep={0.75cm,between origins}]
& \ctrl[style={orange!80!black,fill=orange!80!black},wire style=black]{2} & \ghost{T} \\
& \ctrl{1} & \\
& \gate{\notgate} &
\end{quantikz}
=
\begin{quantikz}[column sep=0.2cm,row sep={0.75cm,between origins}]
& & & & \ctrl{2} & & & & \ctrl{2} & \ctrl{1} & & \ctrl{1} & \gate[style={fill=orange!20}]{T} \\
& & \ctrl{1} & & & & \ctrl{1} & \gate[style={fill=orange!20}]{T\smash{{}^\dagger}} & & \gate{\notgate} & \gate[style={fill=orange!20}]{T\smash{{}^\dagger}} & \gate{\notgate} & \gate{S} & \\
& \gate{H} & \gate{\notgate} & \gate[style={fill=orange!20}]{T\smash{{}^\dagger}} & \gate{\notgate} & \gate[style={fill=orange!20}]{T} & \gate{\notgate} & \gate[style={fill=orange!20}]{T\smash{{}^\dagger}} & \gate{\notgate} & \gate[style={fill=orange!20}]{T} & \gate{H} & & &
\end{quantikz}
}%
\end{minipage}
\setlength{\abovecaptionskip}{2pt}
\setlength{\belowcaptionskip}{-1pt}
\hspace{1em}%
\begin{minipage}[t]{0.33\textwidth}
\caption{Decomposing MCX to Toffoli.} \label{fig:mcx-decomposition}
\end{minipage}%
\begin{minipage}[t]{0.67\textwidth}
\caption{Decomposing Toffoli into Clifford+\tgate{} gates.} \label{fig:toffoli-decomposition}
\end{minipage}
\end{figure}%
\begin{figure}[!t]
\begin{minipage}{0.34\textwidth}
\vspace*{-0.75em}%
\begin{center}
\begin{tabular}{c}
\begin{lstlisting}
with {%\label{line:simple-if-optimized-with}%
  let t <- z;%\label{line:simple-if-optimized-1}%
  let s <- x && y && z;%\label{line:simple-if-optimized-s}%
} do {
  if s {%\label{line:simple-if-optimized-if}%
    let a <- not t;%\label{line:simple-if-optimized-2}%
    let b <- true;%\label{line:simple-if-optimized-3}%
} }
\end{lstlisting}
\end{tabular}
\end{center}
\end{minipage}%
\hspace{\fill}%
\begin{minipage}{0.64\textwidth}
\vspace*{-0.5em}%
\resizebox{\textwidth}{!}{
\newcommand{\grpe}{\gategroup[2,style={dashed,inner xsep=2pt,inner ysep=2pt},background,label style={label position=above,anchor=north,yshift=2.25cm}]{\lstinline[style=color]{let t <- z;}}}
\newcommand{\grpf}{\gategroup[2,steps=3,style={dashed,inner xsep=2pt,inner ysep=2pt},background,label style={label position=above,anchor=north,yshift=3cm}]{\lstinline[style=color]{let a <- not t;}}}
\newcommand{\grpg}{\gategroup[1,style={dashed,inner xsep=2pt,inner ysep=2pt},background,label style={label position=above,anchor=north,yshift=4.5cm}]{\lstinline[style=color]{let b <- true;}}}
\newcommand{\grph}{\gategroup[2,style={dashed,inner xsep=2pt,inner ysep=2pt},background,label style={label position=above,anchor=north,yshift=2.25cm}]{\lstinline[style=color]{let t -> z;}}}
\begin{quantikz}[row sep={0.75cm,between origins},column sep={0.75cm,between origins}]
\lstick{\texttt{s}} &                 &[.75cm] \gate{\notgate} &[.75cm] \ctrl{4}      & \redctrl{5}     & \ctrl{4}        &[1.5cm] \ctrl{6}      &[1cm] \gate{\notgate} &[1cm]            & \\
\lstick{\texttt{x}} & \ghost{X}       & \redctrl{-1}           &                      &                 &                 &                      & \redctrl{-1}         &                 & \\
\lstick{\texttt{y}} & \ghost{X}       & \redctrl{-2}           &                      &                 &                 &                      & \redctrl{-2}         &                 & \\
\lstick{\texttt{z}} & \ctrl{1}\grpe   & \ctrl{-3}              & \ghost{X}            &                 &                 &                      & \ctrl{-3}            & \ctrl{1}\grph   & \\
\lstick{\texttt{t}} & \gate{\notgate} &                        & \gate{\notgate}\grpf & \ctrl{1}        & \gate{\notgate} &                      &                      & \gate{\notgate} & \\
\lstick{\texttt{a}} &                 &                        &                      & \gate{\notgate} &                 &                      &                      &                 & \\
\lstick{\texttt{b}} &                 &                        &                      &                 &                 & \gate{\notgate}\grpg &                      &                 &
\end{quantikz}
}%
\end{minipage}
\setlength{\abovecaptionskip}{2pt}
\setlength{\belowcaptionskip}{-8pt}
\begin{minipage}[t]{0.34\textwidth}
\caption{Optimized version of \Cref{fig:simple-if}.} \label{fig:simple-if-optimized}
\end{minipage}%
\hspace{\fill}%
\begin{minipage}[t]{0.64\textwidth}
\caption{Quantum circuit that corresponds to \Cref{fig:simple-if-optimized}.} \label{fig:simple-if-circuit-optimized}
\end{minipage}
\end{figure}%

The decomposition of an MCX to a Clifford+\tgate{} circuit introduces \tgate{}-complexity.
For example, \Cref{fig:toffoli-decomposition} uses 7 \tgate{} gates to decompose one Toffoli gate,\footnote{As a technical note, the gate $T^\dagger = TSZ$ has a \tgate{}-complexity of 1, as it can be realized using Clifford gates plus one \tgate{} gate.} meaning that \Cref{fig:mcx-decomposition} uses $3 \times 7 = 21$ \tgate{} gates to decompose an MCX gate with 3 control bits.
In general, \citet[Proposition 4.1]{beverland2020} prove that an MCX gate with $n \ge 2$ controls requires at least $n+1$ \tgate{} gates to realize. This lower bound is not reached in practice by \Cref{fig:mcx-decomposition,fig:toffoli-decomposition}, which instead require $7 \times (2(n-2) + 1)$ \tgate{} gates.

\paragraph{Costs of Control Flow}
In other words, on error-corrected quantum hardware, instructions become more costly to execute as the program's control flow becomes more deeply nested.
To accurately predict performance under error correction, one must account for the \tgate{}-complexity of each control bit beyond the first on each MCX gate --- only the first is free in principle because CNOT is a Clifford gate.
In \Cref{fig:simple-if-circuit-full}, we highlight these additional control bits in \textcolor{orange!80!black}{orange}. In addition to the 6 MCX gates, the 13 orange controls cost at least $7 \times 2 \times 13 = 182$ \tgate{} gates using \Cref{fig:toffoli-decomposition,fig:mcx-decomposition}.

\subsection{Cost Model for Accurately Predicting \tgate{}-Complexity Costs} \label{sec:example-cost-model}

The increased cost of control flow under quantum error correction explains the discrepancy between the idealized analysis of \texttt{length} in \Cref{sec:complexity-analysis} and the empirical gate counts in \Cref{fig:slowdown}.
We now demonstrate how using our \tgate{}-complexity cost model, a developer can conduct an analysis that pinpoints the overhead of control flow in a quantum program.

\paragraph{Running Example}
\begin{figure}
\vspace*{-5pt}%
\begin{minipage}[t]{0.47\textwidth}
\begin{lstlisting}[style=highlightif]
fun length[n](xs: ptr<list>, acc: uint) {
 with {
  let is_empty <- xs == null;
 } do {
  if is_empty { let out <- acc; }%\label{line:length-inlined-is_empty}%
  else with {
   /* elided: compute next, r */%\label{line:length-inlined-elided-1}%
   let is_empty2 <- next == null;
  } do {
   if is_empty2 { let out <- r; }
   else with {
    /* elided: compute next2, r2 */%\label{line:length-inlined-elided-2}%
    let is_empty3 <- next2 == null;%\label{line:length-inlined-is_empty2}%
   } do {
    if is_empty3 { let out <- r2; }%\label{line:length-inlined-is_empty3}%
    else with {
     let temp <- default<list>;%\label{line:length-inlined-begin-with}%
     *next2 <-> temp;
     let next3 <- temp.2;
     let r3 <- r2 + 1;%\label{line:length-inlined-end-with}%
    } do {
     let out <- length[n-3](next3, r3);
 }}}}
 return out;
}
\end{lstlisting}
\end{minipage}%
\begin{minipage}[t]{0.52\textwidth}
\begin{lstlisting}[style=highlightif]
fun length[n](xs: ptr<list>, acc: uint) {
 with {
  let is_empty <- xs == null;
  let not_empty <- not is_empty;
  /* elided: compute next, r */
  let is_empty2 <- not_empty && next == null;
  let not_empty2 <- not_empty && next != null;
  /* elided: compute next2, r2 */
  let is_empty3 <- not_empty2 && next2 == null;
  let not_empty3 <- not_empty2 && next2 != null;
  let temp <- default<list>;%\label{line:length-optimized-begin-with}%
  *next2 <-> temp;
  let next3 <- temp.2;
  let r3 <- r2 + 1;%\label{line:length-optimized-end-with}%
 } do {
  if is_empty { let out <- acc; }
  if is_empty2 { let out <- r; }%\label{line:length-optimized-is_empty2}%
  if is_empty3 { let out <- r2; }%\label{line:length-optimized-is_empty3}%
  if not_empty3 {
   let out <- length[n-3](next3, r3);
 }}
 return out;
}
\end{lstlisting}
\end{minipage}
\setlength{\abovecaptionskip}{1pt}
\setlength{\belowcaptionskip}{-9pt}
\begin{minipage}[t]{0.47\textwidth}
\caption{Version of \Cref{fig:length-tower} inlined to 3 levels of recursion, depicting the nesting of conditionals.} \label{fig:length-inlined}
\end{minipage}%
\hspace{\fill}%
\begin{minipage}[t]{0.52\textwidth}
\caption{Optimized version of \Cref{fig:length-inlined}.} \label{fig:length-optimized}
\end{minipage}
\end{figure}
In \Cref{fig:length-inlined}, we illustrate a version of \texttt{length} in which the recursive call on line~\ref{line:length-tower-recursion} of \Cref{fig:length-tower} has been unfolded and inlined twice to reveal the nesting of \texttt{if}-statements.
This program features three levels of nested \texttt{if}, highlighted in \textcolor{orange!80!black}{orange}. When the program is compiled to MCX gates, each \texttt{if} becomes a sequence of control bits placed over its branches. For example, the gates corresponding to the assignment on line~\ref{line:length-inlined-is_empty} are conditioned by \lstinline{is_empty}.

The source of the asymptotic cost is that nested conditional statements compile to nested control bits. For example, the assignment on line~\ref{line:length-inlined-is_empty2} lies under two levels of \texttt{if}-statements, and compiles to a sequence of gates that are controlled by \lstinline{is_empty} and \lstinline{is_empty2}. Likewise, the assignment on line~\ref{line:length-inlined-is_empty3} is controlled by three bits, as are all of lines~\ref{line:length-inlined-begin-with} to~\ref{line:length-inlined-end-with}.

\paragraph{Analysis with Cost Model}
Returning to the recursive form of \texttt{length} in \Cref{fig:length-tower}, we now use our cost model to repair the analysis of \Cref{sec:complexity-analysis} to account for the \tgate{}-complexity of control flow.
Let $\mcxComplexity{n}$ denote the MCX-complexity and $\tComplexity{n}$ the \tgate{}-complexity of the program.

To compute $\tComplexity{n}$, we start as before with the $O(1)$ primitive operations per level and the $\tComplexity{n-1}$ term for the recursive call.
Next, we account for the \tgate{}-complexity of control flow.
In \Cref{fig:length-tower}, the \texttt{if}-\texttt{else} on lines~\ref{line:length-tower-begin-if} to~\ref{line:length-tower-end-if} incurs one control bit for each statement on lines~\ref{line:length-tower-out} to~\ref{line:length-tower-addition}, adding an $O(1)$ term. On line~\ref{line:length-tower-recursion}, the \texttt{if} incurs $O(1)$ cost for each of the $\mcxComplexity{n-1} = O(n)$ primitive operations in the recursive call.
The final recurrence is:
\begin{align*}
    \tComplexity{n}\ \ = \!\!\!\!\!\!\!\underbrace{O(1)}_{\text{operations in level}}\!\!\!\!\!\!\! +\quad \underbrace{\tComplexity{n-1}}_{\text{recursive call}}\quad +\!\!\!\!\!\! \underbrace{O(1)}_{\substack{\text{control flow over}\mathstrut\\\text{operations in level}}}\!\!\!\!\!\!\! +\ \ \ \underbrace{\mcxComplexity{n-1}}_{\substack{\text{control flow over}\mathstrut\\\text{recursive call}}}\ \ =\ \ \tComplexity{n-1} + O(n)
\end{align*}
which yields $\tComplexity{n} = O(n^2)$, agreeing with the empirical results in \Cref{fig:slowdown}.

\subsection{Program-Level Optimizations for Mitigating \tgate{}-Complexity Costs}

We showed that control flow can incur asymptotic overhead in \tgate{}-complexity when compiled using a straightforward strategy. We next present two \emph{program-level optimizations} that rewrite the syntax of the program and produce a new program that then compiles using the same straightforward strategy to a circuit with reduced \tgate{}-complexity.
The first one, \emph{conditional flattening}, can provide an asymptotic speedup while the second, \emph{conditional narrowing}, yields additional constant speedups.

\paragraph{Conditional Flattening}
In \Cref{fig:simple-if-optimized}, we present an optimized form of \Cref{fig:simple-if} that has been subject to both optimizations. First, the \emph{conditional flattening} optimization eliminates control bits that are introduced by nested \texttt{if}-statements, by flattening them via the use of temporary variables. Whereas the original program in \Cref{fig:simple-if} uses three \texttt{if}-statements on lines~\ref{line:simple-if-x},~\ref{line:simple-if-y}, and~\ref{line:simple-if-z}, the optimized program in \Cref{fig:simple-if-optimized} introduces a variable \texttt{s} on line~\ref{line:simple-if-optimized-s} and uses it in a single \texttt{if} on line~\ref{line:simple-if-optimized-if}.

The benefit can be seen in \Cref{fig:simple-if-circuit-optimized}, the circuit to which \Cref{fig:simple-if-optimized} compiles. The gates to which lines~\ref{line:simple-if-optimized-2} and~\ref{line:simple-if-optimized-3} compile are now controlled by only \texttt{s} rather than \texttt{x}, \texttt{y}, and \texttt{z} as in the original circuit in \Cref{fig:simple-if-circuit-full}, saving 8 control bits or $7 \times 2 \times 8 = 112$ \tgate{} gates. Though the computation of \texttt{s} adds 4 control bits, this cost is asymptotically constant with respect to the length of the body of the \texttt{if}.

\paragraph{Conditional Narrowing}
Second, the \emph{conditional narrowing} optimization eliminates control bits introduced by a \texttt{with}-\texttt{do} block under an \texttt{if}-statement, by moving the \texttt{if} into the \texttt{do}-block. In \Cref{fig:simple-if-optimized}, line~\ref{line:simple-if-optimized-1} is no longer under an \texttt{if} as in the original \Cref{fig:simple-if}. As a result, in \Cref{fig:simple-if-circuit-optimized}, the first and last gates are not controlled by \texttt{x}, \texttt{y}, and \texttt{z}, saving 4 more control bits over \Cref{fig:simple-if-circuit-full}.

\paragraph{Running Example}
\begin{wrapfigure}[17]{r}{.44\textwidth}
\vspace*{-1.7em}%
\begin{lstlisting}[xleftmargin=9pt]
fun length[n](xs, acc: uint, b: bool) {
 with {
  let is_empty <- b && xs == null;
  let n_e <- b && xs != null;
  let temp <- default<list>;
  *xs <-> temp;
  let next <- temp.2;
  let r <- acc + 1;
 } do {
  if is_empty { let out <- acc; }
  let rest <- length[n-1](next, r, n_e);
  if n_e { out <-> rest; }
  let rest -> 0;
 }
 return out;
}
\end{lstlisting}
\setlength{\abovecaptionskip}{3pt}
\caption{Program that reflects the optimizations in \Cref{fig:length-optimized} back to the original form in \Cref{fig:length-tower}.} \label{fig:length-optimized-recursive}
\end{wrapfigure}
In \Cref{fig:length-optimized}, we depict the result of optimizations on the unfolded \texttt{length} program from \Cref{fig:length-inlined}.
Conditional flattening turns 3 levels of nested \texttt{if} to 1, such that assuming 8-bit registers, lines~\ref{line:length-optimized-is_empty2} and~\ref{line:length-optimized-is_empty3} save $7 \times 2 \times (2-1+3-1) \times 8 = 336$ \tgate{} gates. Accounting for uncomputation, the use of temporary variables adds $7 \times 2 \times 2 \times 2 = 56$ \tgate{} gates, for a net savings of $336 - 56 = 280$ \tgate{} gates.

Next, conditional narrowing saves a further $7 \times 2 \times 4 \times 8 = 448$ \tgate{} gates by moving lines~\ref{line:length-optimized-begin-with} to~\ref{line:length-optimized-end-with} outside \texttt{if}-statements.
Notably, the program remains safe even though pointer dereferences have been moved outside null checks. All writes to the observable output \texttt{out} remain guarded by appropriate checks, meaning uninitialized data never propagates to the output.

\paragraph{Efficient \tgate{}-Complexity}
When applied to the original \texttt{length} program from \Cref{fig:length-tower}, these optimizations produce the program depicted in \Cref{fig:length-optimized-recursive}.
In this program, recursion no longer takes place under nested \texttt{if}. As a result, in the \tgate{}-complexity analysis, control flow incurs only $O(1)$ overhead, and the recurrence yields an asymptotically efficient $O(n)$ for the optimized program.

\subsection{Comparison to Quantum Circuit Optimizers}

In principle, an alternative to the approach of program-level optimizations -- rewrite the program so that it straightforwardly compiles to a more efficient circuit -- is to emit the asymptotically inefficient circuit of the original program and then attempt to recover an asymptotically efficient circuit using a general-purpose \emph{quantum circuit optimizer}~\citep{hietala2021,sivarajah2021,xu2022,xu2023,quizx,kissinger2020} that researchers have developed to remove and replace inefficient sequences of gates in circuits.

To compare these approaches, we implemented both program-level optimizations in Spire, an extension to the Tower compiler.
In \Cref{fig:optimization-program-level}, we plot the \tgate{}-complexity of \texttt{length} after Spire's optimizations only, and no circuit optimizer. In \Cref{fig:optimization-circuit-level}, we plot the \tgate{}-complexity without Spire's optimizations, and only the Qiskit~\citep{qiskit} and Feynman~\citep{amy2014,feynman} circuit optimizers. We also plot in \Cref{fig:optimization-program-level} the results of a combined approach --- running Spire on the original program, compiling the optimized program to a circuit, and then running Feynman on that circuit. Lastly, we plot the idealized MCX-complexity from \Cref{fig:slowdown} for reference.\footnote{As a technical note, the MCX-complexity is the performance on an idealized architecture and not the minimal \tgate{}-complexity to implement the function.
The MCX and \tgate{}-complexities should be compared only in terms of asymptotics, not constants.}%
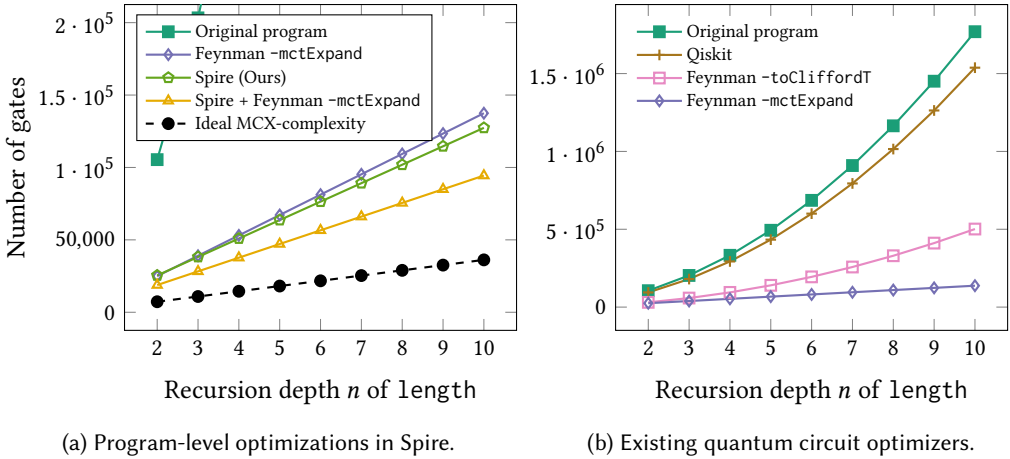
\begin{figure}
\begin{subfigure}{0.5\textwidth}
\begin{tikzpicture}
\begin{axis}[
    xlabel=Recursion depth $n$ of \texttt{length},
    ylabel=Number of gates,
    y label style={at={(axis description cs:0.02,.5)},anchor=south},
    scaled y ticks=false,
    ytick distance=50000,
    xtick distance=1,
    ymin=-12500,
    ymax=2.125e5,
    width=6.8cm,
    legend pos=north west,
    legend cell align=left,
    legend style={nodes={scale=0.7, transform shape},cells={align=left}},
    domain=2:10,
    samples=100,
    tick label style={font=\small},
    every axis plot/.append style={thick}
]
\addplot[Dark2-A,mark=square*] plot coordinates {
    (2,105406)
    (3,203308)
    (4,332654)
    (5,493444)
    (6,685678)
    (7,909356)
    (8,1164478)
    (9,1451044)
    (10,1769054)
};
\addlegendentry{Original program}
\addplot[Dark2-C,-,mark=diamond] plot coordinates {
    (2,24952)
    (3,39020)
    (4,53084)
    (5,67152)
    (6,81220)
    (7,95272)
    (8,109340)
    (9,123392)
    (10,137460)
};
\addlegendentry{Feynman \texttt{-mctExpand}}
\addplot[Dark2-E,-,mark=pentagon] plot coordinates {
    (2,25438)
    (3,38178)
    (4,50918)
    (5,63658)
    (6,76398)
    (7,89138)
    (8,101878)
    (9,114618)
    (10,127358)
};
\addlegendentry{Spire (Ours)}
\addplot[Dark2-F,-,mark=triangle] plot coordinates {
    (2,18780)
    (3,28232)
    (4,37684)
    (5,47136)
    (6,56588)
    (7,66040)
    (8,75492)
    (9,84944)
    (10,94396)
};
\addlegendentry{Spire + Feynman \texttt{-mctExpand}}
\addplot[dashed,thick,mark=*,mark options=solid] plot coordinates {
    (2,7265)
    (3,10879)
    (4,14493)
    (5,18107)
    (6,21721)
    (7,25335)
    (8,28949)
    (9,32563)
    (10,36177)
};
\addlegendentry{Ideal MCX-complexity}
\end{axis}
\end{tikzpicture}%
\caption{Program-level optimizations in Spire.} \label{fig:optimization-program-level}
\end{subfigure}%
\hspace*{\fill}%
\begin{subfigure}{0.5\textwidth}
\begin{tikzpicture}
\begin{axis}[
    xlabel=Recursion depth $n$ of \texttt{length},
    scaled y ticks=false,
    ytick distance=5e5,
    xtick distance=1,
    width=6.8cm,
    legend pos=north west,
    legend cell align=left,
    legend style={fill=none,draw=none,nodes={scale=0.7, transform shape}},
    domain=2:10,
    samples=100,
    tick label style={font=\small},
    every axis plot/.append style={thick}
]
\addplot[Dark2-A,-,mark=square*] plot coordinates {
    (2,105406)
    (3,203308)
    (4,332654)
    (5,493444)
    (6,685678)
    (7,909356)
    (8,1164478)
    (9,1451044)
    (10,1769054)
};
\addlegendentry{Original program}
\addplot[Dark2-G,-,mark=+] plot coordinates {
    (2,94058)
    (3,180090)
    (4,293204)
    (5,433270)
    (6,600288)
    (7,794258)
    (8,1015180)
    (9,1263054)
    (10,1537880)
};
\addlegendentry{Qiskit}
\addplot[Set2-D,-,mark=square] plot coordinates {
    (2,30496)
    (3,58032)
    (4,94434)
    (5,139800)
    (6,194164)
    (7,257506)
    (8,329830)
    (9,411120)
    (10,501412)
};
\addlegendentry{Feynman \texttt{-toCliffordT}}
\addplot[Dark2-C,-,mark=diamond] plot coordinates {
    (2,24952)
    (3,39020)
    (4,53084)
    (5,67152)
    (6,81220)
    (7,95272)
    (8,109340)
    (9,123392)
    (10,137460)
};
\addlegendentry{Feynman \texttt{-mctExpand}}
\end{axis}
\end{tikzpicture}
\caption{Existing quantum circuit optimizers.} \label{fig:optimization-circuit-level}
\end{subfigure}

\caption{\tgate{}-complexity of \texttt{length} after quantum circuit optimizers and program-level optimizations in Spire.}
\label{fig:optimization}
\end{figure}%

First, Qiskit and one configuration of Feynman do not produce a circuit with linear \tgate{}-complexity, while Spire and a second configuration of Feynman do.
A possible explanation for the difference is that conditional flattening is not captured by the rewrites of Clifford+\tgate{} gates that Qiskit implements.
By contrast, as we discuss in \Cref{sec:discussion}, the optimization can be captured by cancelling adjacent Clifford+Toffoli gates, enabling a configuration of Feynman using that strategy to succeed.

Next, Spire and Feynman together achieve better speedups than either alone --- Feynman leaves behind some fraction of \tgate{} gates that Spire can eliminate.
As we discuss in \Cref{sec:discussion}, one challenge that a circuit optimizer must overcome to fully capture the effect of conditional narrowing is that the circuit optimizer must perform rewrites over an unbounded number of gates.

Finally, Spire takes only \SI{0.05}{\second} to emit an efficient circuit, whereas Feynman takes 2 minutes to do so in this case. The reason is that whereas the circuit optimizer must process a large circuit to shrink it down, Spire optimizes the program so that the large circuit is not created in the first place.

\section{Tower Language Overview} \label{sec:tower}

In this section, we briefly review the syntax and semantics of Tower~\citep{tower}, a quantum programming language featuring abstractions for control flow in superposition.

\paragraph{Language Syntax}
The Tower language features the data types of integers, tuples, and pointers, along with operations on these data types.
In \Cref{fig:tower-syntax}, we depict the core syntax of the language.
\begin{figure}
\vspace*{-0.5em}%
\scalebox{0.9}{\parbox{.5\linewidth}{%
\begin{align*}
    \textsf{Type}\ \type \Coloneqq{} & \tUnit \mid \tUInt \mid \tBool \mid \tPair{\type_1}{\type_2} \mid \tPtr{\type} \\
    \textsf{Value}\ \vValue \Coloneqq{} & x \mid \vUnit \mid \ePair{x_1}{x_2} \mid \vNum{n} \mid \vTrue \mid \vFalse \mid \vNull{\type} \mid \vPtr{\type}{p} \hspace{1em} (n \in \textsf{UInt}, p \in \textsf{Addr}) \\
    \textsf{Expression}\ \eExp \Coloneqq{} & \vValue \mid \eProj{1}{x} \mid \eProj{2}{x} \mid \eUnop{uop}{x} \mid \eBinop{x_1}{bop}{x_2} \\
    \textsf{Operator}\ uop \Coloneqq{} & \oNot \mid \oTest \hspace{2em} bop \Coloneqq{} \oAnd \mid \oOr \mid \oAdd \mid \oSub \mid \oMul \\
    \textsf{Statement}\ \sStmt \Coloneqq{} & \sIf{x}{s} \mid \sSeq{\sStmt_1}{\sStmt_2} \mid \sSkip \mid \sBind{x}{\eExp} \mid \sUnbind{x}{\eExp} \mid \sHadamard{x} \mid \sSwap{x_1}{x_2} \mid \sMemSwap{x_1}{x_2}
\end{align*}%
}}
\setlength{\abovecaptionskip}{3pt}
\setlength{\belowcaptionskip}{-7pt}
\caption{Core syntax of the Tower quantum programming language.} \label{fig:tower-syntax}
\end{figure}

In Tower, all recursive function definitions and calls are inlined by the compiler, producing a program that uses only the core syntax above~\citep[Section 6]{tower}. In the example from \Cref{fig:length-tower}, the annotation \texttt{n} instructs the compiler to inline \texttt{length} into itself \texttt{n} times.

Apart from standard imperative programming features, Tower supports a number of constructs necessary for quantum programming. The \emph{un-assignment} construct $\sUnbind{x}{\eExp}$ uncomputes (\Cref{sec:background}) the value of $x$ using the value of $e$. The construct $\sSwap{x_1}{x_2}$ swaps the values of variables $x_1$ and $x_2$, and $\sMemSwap{x_1}{x_2}$ swaps a value stored in memory at pointer $x_1$ with the value of $x_2$.

We study a version of Tower extended with a statement $\sHadamard{x}$ that executes a Hadamard gate (\Cref{sec:background}) on the Boolean variable $x$.
Because the Hadamard and Toffoli gates are universal for quantum computation~\citep{shi2002}, the availability of the Hadamard and NOT gates and the $\sIf{x}{s}$ construct means that any quantum computation can be expressed as a Tower program.

\paragraph{Language Semantics}
The type system of Tower assigns a type to each value or expression and determines whether a statement is well-formed. In \Cref{sec:full-type-system}, we define typing for values and expressions, and the judgment $\stmtok{\ctx}{\sStmt}{\ctx'}$, which states that the statement $\sStmt$ is well-formed under a context $\ctx$ of variables and produces a context $\ctx'$ of the updated declarations after executing $\sStmt$.

The circuit semantics of Tower assigns to each program $s$ a corresponding quantum circuit $\circuit{s}$ that can execute on a quantum computer.
In \Cref{sec:full-circuit-semantics}, we define this semantics and specifically how $\circuit{s}$ maps an input machine state $\ket{\reg, \mem}$ to an output machine state $\ket{\reg', \mem'}$. Here, $\reg$ denotes a \emph{register file} mapping variables to values and $\mem$ denotes a \emph{memory} mapping addresses to values.

In Tower, the dereferencing of a null pointer is a no-op, not a runtime error. When a variable is re-defined, the value of its corresponding register becomes the XOR of its old and new values.

\paragraph{Derived Forms}
Each statement $s$ in Tower is \emph{reversible}, meaning that there exists a statement $\reverse{s}$ whose semantics are the reverse of $s$. Specifically, $\reverse{\sSeq{s_1}{s_2}}$ is $\sSeq{\reverse{s_2}}{\reverse{s_1}}$. Similarly, $\reverse{\sBind{x}{e}}$ is $\sUnbind{x}{e}$ and vice versa, $\reverse{\sIf{x}{s}}$ is $\sIf{x}{\reverse{s}}$, and the reverse of any other $s$ is $s$ itself.

Based on this concept, we define the derived form $\sWith{s_1}{s_2}$ as $\sSeq{s_1}{\sSeq{s_2}{\reverse{s_1}}}$, and use it to automate the insertion of uncomputation statements for variables within block scope.
Memory allocation and deallocation desugar to core constructs, following the process described in \citet[Section 5]{tower}.
Other derived forms, such as the \texttt{if}-\texttt{else} construct, are described in \citet[Appendix B]{tower} and similarly desugar to core constructs.

\section{Cost Model} \label{sec:cost-model}

In this section, we present a cost model that computes the \tgate{}-complexity of a quantum program that utilizes programming abstractions for control flow in superposition. Using the cost model, a developer can perform a syntactic analysis that determines the runtime cost of a program on an error-corrected quantum architecture and pinpoint the sources of asymptotic slowdown.

Given a program $\sStmt$, the cost model quantifies the number of gates in the circuit $\circuit{s}$ to which it compiles, following the semantics of \Cref{sec:tower}. More generally, the cost model also matches the compilation of other languages with quantum \texttt{if}, such as QML~\citep{qml}, ScaffCC~\citep{javadi2014}, Silq~\citep{hans2022}, and Qunity~\citep{voichick2023}.

\paragraph{MCX-Complexity}
We denote by $\mcxComplexity{s}$ the MCX-complexity of the program $s$, which is formally defined as the number of gates in its compiled circuit $\circuit{s}$ when expressed in an idealized gate set consisting of arbitrarily controllable Clifford gates, which includes arbitrary MCX gates:
\begin{alignat*}{5}
\mcxComplexity{\sSkip} &= 0 \quad & \mcxComplexity{\sSeq{s_1}{s_2}} &= \mcxComplexity{s_1} + \mcxComplexity{s_2} \\
\mcxComplexity{\sIf{x}{\sStmt}} &= \mcxComplexity{\sStmt} \quad & \mcxComplexity{\sStmt} &= c^\textrm{MCX}_s \text{ for any other $s$}
\end{alignat*}
where $0 \le c^\textrm{MCX}_s = O(1)$ represents the number of arbitrarily controllable Clifford gates, including MCX gates, used by the primitive operation $s$. This constant is determined by the implementation of $s$, and all primitive $s$ satisfy $c^\textrm{MCX}_s > 0$ except for only $\sSkip$ or $\sBind{x}{v}$ or $\sUnbind{x}{v}$ where $v$ has an all-zero bit representation for which no gates are emitted.
The reason for why the \texttt{if}-statement does not increase the MCX-complexity is that the number of arbitrarily controllable Clifford gates, including MCX gates, does not change when more control bits are added to gates.

\begin{theorem}[MCX-Complexity Soundness]
If $s$ is well-formed, i.e.\ $\stmtok{\ctx}{s}{\ctx'}$, then the number of arbitrarily controllable Clifford gates in $\circuit{s}$ is equal to $\mcxComplexity{s}$, up to choices for $c^\textrm{MCX}_s$.
\end{theorem}
\begin{proof}
By induction on the definition of $\circuit{s}$. The significant case is \texttt{if}, as explained above.
\end{proof}

\paragraph{T-Complexity}
We denote by $\tComplexity{s}$ the \tgate{}-complexity of the program $s$, which is formally defined as the number of \tgate{} gates in its compiled circuit $\circuit{s}$ when expressed in the Clifford+\tgate{} gate set:
\begin{alignat*}{5}
\tComplexity{\sSkip} &= 0 & \tComplexity{\sSeq{s_1}{s_2}} &= \tComplexity{s_1} + \tComplexity{s_2} \\
\tComplexity{\sIf{x}{\sSeq{s_1}{s_2}}} &= \tComplexity{\sIf{x}{s_1}} + \tComplexity{\sIf{x}{s_2}} \quad & \tComplexity{\sIf{x}{\sHadamard{y}}} &= c^{\tgate}_\mathit{CH} \\
\tComplexity{\sIf{x}{\sBind{y}{v}}} &= \tComplexity{\sIf{x}{\sUnbind{y}{v}}} = 0 \text{ for value $v$} \quad & \\
\tComplexity{\sIf{x}{\sStmt}} &= c^{\tgate}_\textrm{ctrl} * \mcxComplexity{s} + \tComplexity{s} \text{ for other $s$} \quad & \tComplexity{\sStmt} &= c^{\tgate}_s \text{ for other $s$}
\end{alignat*}
where $0 \le c^{\tgate}_s = O(1)$ represents the number of \tgate{} gates used by the primitive operation $s$, which is determined by the implementation of $s$. Simple $s$ such as $\sBind{x}{v}$ and $\sHadamard{x}$ have $c^{\tgate}_s = 0$, whereas others such as $\sBind{x}{\eBinop{y}{\oMul}{z}}$ for which an arithmetic circuit must be instantiated have $c^{\tgate}_s > 0$.
The constant $0 < c^{\tgate}_\mathit{CH} = O(1)$ represents the number of \tgate{} gates required to implement a controlled-Hadamard gate. Using the construction of \citet[Figure 17]{lee2021}, we have $c^{\tgate}_\mathit{CH} = 8$.
The constant $0 < c^{\tgate}_\textrm{ctrl} = O(1)$ represents the number of \tgate{} gates required to add an additional control bit to a multi-controlled gate. Using the decompositions in \Cref{fig:toffoli-decomposition,fig:mcx-decomposition}, we have $c^{\tgate}_\textrm{ctrl} = 2 \times 7 = 14$.

\begin{theorem}[\tgate{}-Complexity Soundness]
If $s$ is well-formed, i.e.\ $\stmtok{\ctx}{s}{\ctx'}$, then the number of \tgate{} gates in $\circuit{s}$ is equal to $\tComplexity{s}$, up to choices for the constants $c^\textrm{MCX}_s$, $c^{\tgate}_s$, $c^{\tgate}_\mathit{CH}$, and $c^{\tgate}_\textrm{ctrl}$.
\end{theorem}
\begin{proof}
By induction on $\circuit{s}$. There are three significant cases. The first is $\sIf{x}{\sBind{y}{v}}$ and $\sIf{x}{\sUnbind{y}{v}}$, which add a control bit to a circuit $\circuit{\sBind{y}{v}}$ or $\circuit{\sUnbind{y}{v}}$ respectively that does not contain any controlled or Hadamard gates, meaning that the resulting \tgate{}-complexity is zero.

The second is $\sIf{x}{\sHadamard{y}}$, a controlled-Hadamard gate with cost $c^{\tgate}_\mathit{CH}$ by definition.

The third is $\sIf{x}{\sStmt}$ where $\sStmt$ is $\sBind{y}{e}$ or $\sUnbind{y}{e}$ or $\sIf{y}{\sStmt'}$ or $\sSwap{y}{z}$ or $\sMemSwap{y}{z}$.
In these cases, the number of gates in $\circuit{\sStmt}$ that are not Clifford when one more control bit is added is proportional to $\mcxComplexity{\sStmt}$, and adding a control for $x$ incurs a \tgate{}-complexity of $c^{\tgate}_\textrm{ctrl}$ at each such gate.
\end{proof}

\section{Program-Level Optimizations} \label{sec:optimizations}

In this section, we present \emph{program-level optimizations} for quantum programs that utilize control flow in superposition. Using these optimizations, a developer can rewrite a program to reduce its \tgate{}-complexity, predict the \tgate{}-complexity of the optimized program using the cost model, and then compile the program to an efficient circuit using a straightforward strategy. Forms of these optimizations appear in prior work~\citep{ittah2022,steiger2018,seidel2022}, and we present in this section a novel unification of these optimizations as program rewrite rules.

\subsection{Conditional Flattening Optimization} \label{sec:conditional-flattening}

The conditional flattening optimization identifies instances in which control bits are introduced by nested \texttt{if}-statements and can be optimized by flattening the structure of \texttt{if}-statements. Specifically, the optimization performs these following program rewrite rules whenever possible:
\begin{align*}
\sIf{x}{\sIf{y}{s}} &\rightsquigarrow \sWith{\sBind{x'}{\eBinop{x}{\oAnd}{y}}}{\sIf{x'}{s}} \\
\sIf{x}{\sSeq{s_1}{s_2}} &\rightsquigarrow \sSeq{\sIf{x}{s_1}}{\sIf{x}{s_2}}
\end{align*}
Whereas the original program incurs many control bits over $s$, the optimized program computes a temporary value and uses it to control $s$ using only one bit, yielding an asymptotic improvement:

\begin{theorem} \label{thm:nested-conditionals-efficiency}
When $\circuit{s}$ contains $k$ MCX gates with at least one control and $s$ falls under $n$ levels of nested \texttt{if}, conditional flattening reduces the \tgate{}-complexity of the program from $O(kn)$ to $O(k + n)$.
\end{theorem}

\begin{proof}
By induction on the structure of $s$. For each of the $n-1$ layers of \texttt{if} that is removed, the \tgate{}-complexity of the program reduces by $k$, while the inserted \texttt{with}-block has $O(1)$ \tgate{}-complexity.
\end{proof}

We next show that this program-level optimization preserves the circuit semantics of a program with respect to its free variables, as formalized by the following definition:
\newsavebox{\slashwire}
\savebox{\slashwire}{
\begin{quantikz}
    & \qwbundle{}
\end{quantikz}
}
\newcommand{\varwire}{\smash{\hspace{-0.3em}{$\oset{\hspace{1.7em}\tiny k}{\usebox{\slashwire}}$}\hspace{-0.3em}}}
\newsavebox{\tripwire}
\savebox{\tripwire}{
\begin{quantikz}[classical gap=0.075cm]
    \setwiretype{b} & \qw{}
\end{quantikz}
}
\newcommand{\varswire}{\hspace{-0.3em}{\usebox{\tripwire}}\hspace{-0.3em}}
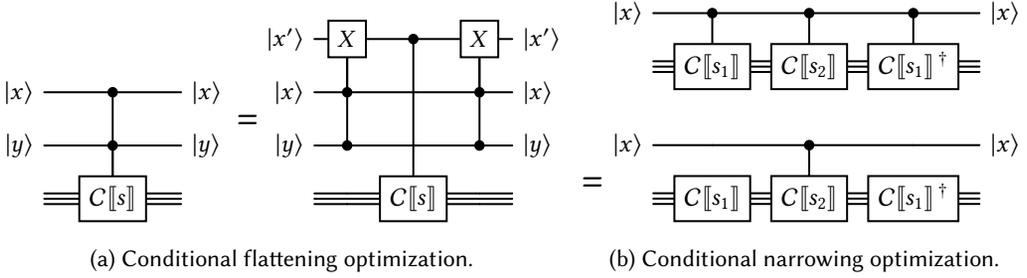
\begin{figure}
\begin{subfigure}[b]{0.56\textwidth}
\resizebox{\textwidth}{!}{
\begin{quantikz}[row sep={0.75cm,between origins},classical gap=0.075cm]
\ghost{T} \\
\lstick{$\ket{x}$} & \ctrl{2} & \rstick{$\ket{x}$} \\
\lstick{$\ket{y}$} & \ctrl{1} & \rstick{$\ket{y}$} \\
\setwiretype{b} & \gate{\circuit{s}} &
\end{quantikz}
\raisebox{-0.1em}{\huge =}\hspace*{-0.05em}%
\begin{quantikz}[column sep=0.2cm,row sep={0.75cm,between origins},classical gap=0.075cm]
\lstick{$\ket{x'}$} & \gate{\notgate} & \ctrl{3} & \gate{\notgate} & \rstick{$\ket{x'}$} \\
\lstick{$\ket{x}$} & \ctrl{-1} & & \ctrl{-1} & \rstick{$\ket{x}$} \\
\lstick{$\ket{y}$} & \ctrl{-2} & & \ctrl{-2} & \rstick{$\ket{y}$} \\
& \setwiretype{b} & \gate{\circuit{s}} & &
\end{quantikz}
}
\caption{Conditional flattening optimization.} \label{fig:nested-conditionals-circuit}
\end{subfigure}%
\hspace*{\fill}%
\begin{subfigure}[b]{0.44\textwidth}
\resizebox{\textwidth}{!}{
\raisebox{-3.2em}{\huge =}%
\begin{quantikz}[column sep=0.3cm,row sep={0.75cm,between origins},classical gap=0.075cm]
\lstick{$\ket{x}$} & \ctrl{1} & \ctrl{1} & \ctrl{1} & \ghost{\circuit{s_1}} \rstick{$\ket{x}$} \\
\setwiretype{b} & \gate{\circuit{s_1}} & \gate{\circuit{s_2}} & \gate{\circuit{s_1}\smash{{}^\dagger}} & \\[1em]
\lstick{$\ket{x}$} & & \ctrl{1} & & \ghost{\circuit{s_1}} \rstick{$\ket{x}$} \\
\setwiretype{b} & \gate{\circuit{s_1}} & \gate{\circuit{s_2}} & \gate{\circuit{s_1}\smash{{}^\dagger}} &
\end{quantikz}
}
\caption{Conditional narrowing optimization.} \label{fig:conditional-inverses-circuit}
\end{subfigure}
\caption{Circuit equivalence rules that hold by direct reasoning on quantum circuits and visually demonstrate the soundness of the program-level optimizations. The notation \varswire{} denotes a collection of many registers.} \label{fig:optimizations-circuits}
\end{figure}%
\begin{definition}[Circuit Equivalence]
Given a set $X$ of variables, we say that register files $\reg_1$ and $\reg_2$ are \emph{equivalent}, denoted $\reg_1 \equiv_X \reg_2$, when they map the variables in $X$ to equal values respectively and all other variables to zero.
Given two sets $X, X'$ of variables, we say that circuits $\mathcal{C}_1$ and $\mathcal{C}_2$ are \emph{equivalent}, denoted $X \vdash \mathcal{C}_1 \equiv \mathcal{C}_2 \dashv X'$, when given any memory $M$ and two register files $R_1$ and $R_2$ such that $R_1 \equiv_X R_2$, we have $\mathcal{C}_1\ket{R_1, M} = \ket{R'_1, M'}$ and $\mathcal{C}_2\ket{R_2, M} = \ket{R'_2, M'}$ where $R'_1 \equiv_{X'} R'_2$.
\end{definition}

\begin{theorem}[Conditional Flattening Soundness]
Assume $\stmtok{\ctx}{\sIf{x}{\sIf{y}{s}}}{\ctx'}$. Then, we have $\domain{\ctx} \vdash \circuit{\sIf{x}{\sIf{y}{s}}} \equiv \circuit{\sWith{\sBind{x'}{\eBinop{x}{\oAnd}{y}}}{\sIf{x'}{s}}} \dashv \domain{\ctx'}$.
\end{theorem}

\begin{proof}
The claim follows from a circuit equivalence that we visually depict in \Cref{fig:nested-conditionals-circuit}.
\end{proof}

\subsection{Conditional Narrowing Optimization}

The conditional narrowing optimization identifies instances in which control bits are introduced by a \texttt{with}-\texttt{do} block under an \texttt{if}-statement, which can be optimized by moving the \texttt{if}-statement under the \texttt{do}-block. Specifically, the optimization performs the following rewrite whenever possible:
\begin{align*}
\sIf{x}{\sWith{s_1}{s_2}} \rightsquigarrow \sWith{s_1}{\sIf{x}{s_2}}
\end{align*}
The optimized program unconditionally executes $s_1$ and its reverse, for a constant improvement:

\begin{theorem}
When $\circuit{s_1}$ contains $k$ MCX gates with at least one control, the conditional narrowing optimization reduces the \tgate{}-complexity of the program by an $O(k)$ additive term.
\end{theorem}

\begin{proof}
By induction on the structure of $s_1$, where $k$ controls are removed on $s_1$ and its reverse.
\end{proof}

\begin{theorem}[Conditional Narrowing Soundness]
Let $\stmtok{\ctx}{\sIf{x}{\sWith{s_1}{s_2}} }{\ctx'}$. Then, $\domain{\ctx} \vdash \circuit{\sIf{x}{\sWith{s_1}{s_2}} } \equiv \circuit{\sWith{s_1}{\sIf{x}{s_2}}} \dashv \domain{\ctx'}$.
\end{theorem}

\begin{proof}
The claim follows from a circuit equivalence that we visually depict in \Cref{fig:conditional-inverses-circuit}.
\end{proof}

\section{Implementation: Spire Quantum Compiler} \label{sec:implementation}

As the artifact of this work, we implemented Spire, an extension of the Tower compiler that performs the optimizations of \Cref{sec:optimizations}. In this section, we briefly describe the architecture of the Tower compiler, the transformations added by Spire, and the challenges that arose in implementation.

\paragraph{Compiler Overview}
The Tower compiler has four main stages. First, given a Tower program, the lexer and parser construct its abstract syntax tree. Next, the compiler lowers the surface AST to the \emph{core intermediate representation}, whose syntax is presented in \Cref{sec:tower}. This lowering involves inlining all function calls and translating memory allocation and derived forms to core syntax.

Then, the compiler lowers the core IR to an \emph{abstract circuit} that is analogous to classical assembly, with the abstractions of word-sized registers; arithmetic, logical, memory, and data movement instructions; and instructions controlled by registers. The compiler invokes a register allocator to map IR variables to registers and compiles \texttt{if}-statements to multiply-controlled instructions.

Finally, the compiler lowers the abstract circuit to a \emph{concrete circuit} by instantiating each arithmetic, logical, memory, and data movement instruction as an explicit sequence of MCX gates. The compiler then emits the concrete circuit in the quantum circuit format of~\citet{dotQC}.

\paragraph{Spire Transformations}
We implemented Spire as a compiler pass that transforms the core IR\@.
First, we modified the core IR to add \texttt{with}-\texttt{do} blocks, facilitating the conditional narrowing optimization.
Next, we implemented a compiler pass that rewrites the core IR using the conditional flattening and conditional narrowing optimizations.
As they are simple syntax rewrites, this pass constitutes only 12 lines of OCaml code, which we present in \Cref{sec:ocaml}.
Then, we added a simple compiler pass that flattens the structure of \texttt{with}-\texttt{do} blocks before continuing to the next stage.

\paragraph{Downstream Challenges}
Though the new passes are simple, they required detailed analysis and altered assumptions in the register allocation approach taken by the compiler.
In \Cref{sec:implementation-case-study}, we detail the challenge that arises and our solution, as a case study for quantum compiler developers.

\section{Evaluation} \label{sec:evaluation}

In this section, we evaluate our cost model and optimizations as measured by the \tgate{}-complexity of a benchmark suite of quantum programs. We answer the following research questions:
\begin{enumerate}
    \item [\emph{RQ1.}] How accurately does the cost model predict the asymptotic \tgate{}-complexity of programs?
    \item [\emph{RQ2.}] By how much do the program-level optimizations of conditional flattening and conditional narrowing improve the \tgate{}-complexity of a quantum program?
    \item [\emph{RQ3.}] By how much do quantum circuit optimizers from existing work improve the \tgate{}-complexity of a quantum program after it has been fully compiled to a circuit of logic gates?
    \item [\emph{RQ4.}] What is the effect on compilation time of performing the program-level optimizations, and how does it compare to the effect on compilation time of quantum circuit optimizers?
\end{enumerate}

In \Cref{tbl:benchmarks}, we list the benchmarks that we use throughout this evaluation and include in the paper artifact. They are data structure operations used by quantum algorithms for search~\citep{ambainis2004}, optimization~\citep{bernstein2013}, and geometry~\citep{aaronson2019}, and include the \texttt{length} example from \Cref{sec:example} and others such as insertion into a radix tree-based set.

In \Cref{sec:eval-program-level,sec:eval-circuit-optimizers}, we also introduce \texttt{length-simple}, a simplified version of \texttt{length} that has the same asymptotic \tgate{}-complexity but omits the primitive operations on lines~\ref{line:length-tower-dereference} and~\ref{line:length-tower-addition}. These lines perform a memory dereference and an addition operation respectively. Semantically, dropping them causes the \texttt{length} function to return an incorrect output. For compilation, dropping them results in a circuit whose size has the same asymptotic behavior but is scaled down by a fraction.

The reason we perform this simplification is to enable a comparison to existing quantum circuit optimizers.
Without this simplification, the circuit would be two orders of magnitude larger, meaning that all but one of the existing optimizers we tested would take more than 1 hour to run.

\subsection{RQ1: Accuracy of Cost Model}

\begin{enumerate}
    \item [\emph{RQ1.}] How accurately does the cost model predict the asymptotic \tgate{}-complexity of programs?
\end{enumerate}

\paragraph{Methodology}

To obtain the predicted asymptotic \tgate{}-complexity, we performed the same analysis as in \Cref{sec:example-cost-model}. We performed an asymptotic analysis because the values of constants in the cost function, in particular the costs of primitive operations such as arithmetic and memory, are difficult to determine theoretically and significantly affect the precision of non-asymptotic estimates.

As an example, the function \texttt{insert} in \Cref{tbl:benchmarks} inserts an element into a set data structure that is concretely implemented as a radix tree. This function invokes a string \texttt{compare} operation and a recursive call at each level, all under an \texttt{if}. Because the other operations in the program have equal or less \tgate{}-complexity compared to \texttt{compare}, the overall \tgate{}-complexity of \texttt{insert} is:
\begin{align*}
    \tComplexity[\texttt{insert}]{d}\, ={} & \!\!\!\underbrace{\tComplexity[\texttt{compare}]{d}}_{\text{operations in level}}+\!\underbrace{\mcxComplexity[\texttt{compare}]{d}}_{\text{control flow in level}}\!\! +\ \,\underbrace{\tComplexity[\texttt{insert}]{d-1}}_{\text{recursive call}}\ \,+\!\!\!\!\!\!\!\underbrace{\mcxComplexity[\texttt{insert}]{d-1}}_{\text{control flow over recursive call}} \\
    ={} & \,O(d^2) + O(d) + \tComplexity[\texttt{insert}]{d-1} + O(d^2)
\end{align*}
which solves to $\tComplexity[\texttt{insert}]{d} = O(d^3)$, an asymptotic increase over the MCX-complexity of $O(d^2)$.

To compute the empirical \tgate{}-complexity, we used Spire (\Cref{sec:implementation}) with optimizations off to compile each program to MCX gates. We then counted \tgate{} gates as follows: each MCX with $c \ge 2$ controls corresponds to $2(c - 2) + 1$ Toffoli gates as in \Cref{fig:mcx-decomposition}, and each Toffoli corresponds to 7 \tgate{} gates as in \Cref{fig:toffoli-decomposition}.
To determine the scaling in the recursion depth $n$ or $d$, we repeated the process for depths from 2 to 10 and found the lowest-degree polynomial that exactly fits the \tgate{}-complexities.

To obtain the predicted and empirical MCX-complexity, we performed the same procedure as above, except that we used the MCX-complexity recurrence and counted the number of MCX gates.

\paragraph{Results}
\newcommand{\subname}{\textcolor{gray}{$-$}\ }
\begin{table}
\caption{List of benchmark programs and their MCX and \tgate{}-complexities, in terms of the size $n$ or depth $d = O(\log n)$ of the data structure. We report \tgate{}-complexity both before and after program-level optimizations. ``Predicted'' reports the asymptotic MCX or \tgate{}-complexity predicted by the cost model, and ``Empirical'' reports the MCX or \tgate{}-complexity of the compiled circuit. Large empirical figures are reported in \Cref{sec:full-complexities}.}
\resizebox{\textwidth}{!}{%
\begin{tabular}{ l c @{} c @{} c c @{} c @{} c c @{} c @{} }
\toprule
                                    & \multicolumn{2}{c}{MCX-Complexity} & & \multicolumn{2}{c}{\tgate{}-Complexity Before Optimizations} & & \multicolumn{2}{c}{\tgate{}-Complexity After Optimizations} \\
\cmidrule{2-3} \cmidrule{5-6} \cmidrule{8-9}
Program                             & Predicted & Empirical & & Predicted & Empirical & & Predicted & Empirical \\
\midrule
List & & & & \\
\subname{}\texttt{length}          & $O(n)$ & $2246n + 32$ & & $O(n^2)$ & $15722n^2 + 19292n + 3934$ & & $O(n)$ & $12740n - 42$ \\
\subname{}\texttt{sum}             & $O(n)$ & $2642n + 32$ & & $O(n^2)$ & $18494n^2 + 19628n + 4298$ & & $O(n)$ & $13272n - 42$ \\
\subname{}\texttt{find\_pos}       & $O(n)$ & $2294n + 32$ & & $O(n^2)$ & $16058n^2 - 8820n + 6426$ & & $O(n)$ & $12740n - 42$ \\
\subname{}\texttt{remove}          & $O(n)$ & $4990n + 32$ & & $O(n^2)$ & $34930n^2 + 26376n + 10304$ & & $O(n)$ & $58912n - 12124$ \\
Queue & & & & \\
\subname{}\texttt{push\_back}      & $O(n)$ & $2864n + 32$ & & $O(n^2)$ & $20048n^2 + 11508n + 4634$ & & $O(n)$ & $46256n - 13006$ \\
\subname{}\texttt{pop\_front}      & $O(1)$ & $1452$ & & $O(1)$   & $8456$ & & $O(1)$ & $8456$ \\
String & & & & \\
\subname{}\texttt{is\_prefix}      & $O(n)$ & $4585n + 32$ & & $O(n^2)$ & $64190n^2 - 11529n + 6545$ & & $O(n)$ & $16758n - 42$ \\
\subname{}\texttt{num\_matching}   & $O(n)$ & $6052n + 5516$ & & $O(n^2)$ & $84728n^2 + 129360n + 59710$ & & $O(n)$ & $21826n + 18676$ \\
\subname{}\texttt{compare}         & $O(n)$ & $4633n + 32$ & & $O(n^2)$ & $97293n^2 + 10598n + 4781$ & & $O(n)$ & $17773n - 42$ \\
Set (radix tree) & & & & \\
\subname{}\texttt{insert}          & $O(d^2)$ & $O(d^2)$ (App.~\ref{sec:full-complexities}) & & $O(d^3)$ & $O(d^3)$ (\Cref{sec:full-complexities}) & & $O(d^2)$ & $256914d^2 + 1413244d - 840$ \\
\subname{}\texttt{contains}        & $O(d^2)$ & $O(d^2)$ (App.~\ref{sec:full-complexities}) & & $O(d^3)$ & $O(d^3)$ (\Cref{sec:full-complexities}) & & $O(d^2)$ & $134064d^2 + 687008d - 42$ \\
\bottomrule
\end{tabular}%
}
\label{tbl:benchmarks}
\end{table}
For each benchmark in \Cref{tbl:benchmarks}, the cost model accurately predicts the asymptotic \tgate{}-complexity, as confirmed by the matching empirical \tgate{}-complexity.
In particular, for each benchmark whose MCX-complexity is not constant, meaning the recurrence is nontrivial, it accurately predicts that the \tgate{}-complexity of the unoptimized program is one degree higher than the MCX-complexity.

\subsection{RQ2: Effect of Program-Level Optimizations on \tgate{}-Complexity} \label{sec:eval-program-level}

\begin{enumerate}
    \item [\emph{RQ2.}] By how much do the program-level optimizations of conditional flattening and conditional narrowing improve the \tgate{}-complexity of a quantum program?
\end{enumerate}

\paragraph{Methodology}
For this question, we used Spire to execute each optimization on each benchmark program and found the empirical \tgate{}-complexity by counting \tgate{} gates in the same way as in RQ1.

\paragraph{Results}
In \Cref{tbl:benchmarks}, we present the \tgate{}-complexity of each program after applying both optimizations.
For each benchmark, the optimizations recover a program whose \tgate{}-complexity is equal to the MCX-complexity, as determined both by the cost model and by circuit compilation.

For \texttt{length} and \texttt{length-simplified}, the \tgate{}-complexity improves from quadratic to linear.
In \Cref{fig:optimizer-efficacy-b}, we plot the \tgate{}-complexity of \texttt{length-simplified} after applying each of the optimizations in Spire.
When used alone, conditional narrowing achieves 19.9\% improvement over the original program at depth $n = 10$, and conditional flattening alone achieves 88.2\% improvement. When Spire applies conditional narrowing on top of conditional flattening, conditional narrowing achieves a further 63.0\% improvement, which stacks to 95.6\% improvement end-to-end.

In \Cref{sec:optimization-costs}, we analyze the \tgate{}-complexity that conditional flattening incurs due to its added uncomputation. Across all of the benchmarks in \Cref{tbl:benchmarks} at recursion depth $n = 10$, 0 to 4.81\% (average 0.49\%) of the \tgate{} gates in the final compiled circuit correspond to the uncomputation that is introduced by conditional flattening. At depth $n = 2$, this figure is 0 to 2.85\% (average 0.30\%).

\subsection{RQ3: Effect of Existing Circuit Optimizers on \tgate{}-Complexity} \label{sec:eval-circuit-optimizers}

\begin{enumerate}
    \item [\emph{RQ3.}] By how much do quantum circuit optimizers from existing work improve the \tgate{}-complexity of a quantum program after it has been fully compiled to a circuit of logic gates?
\end{enumerate}

\paragraph{Methodology}

We evaluated the following optimizers: Qiskit~\citep{qiskit}, VOQC~\citep{hietala2021}, Pytket~\citep{sivarajah2021}, Feynman~\citep{amy2014,feynman}, Quartz~\citep{xu2022}, and QUESO~\citep{xu2023}. We also evaluated QuiZX~\citep{quizx}, a fast Rust port of PyZX~\citep{kissinger2020} that produces outputs identical to PyZX\@.\footnote{As part of our evaluation, we ran PyZX for comparison with QuiZX, and observed that they produce circuits with identical \tgate{}-complexity, though PyZX takes more time to produce the output. We thus do not report PyZX results separately.}

First, we used Spire to compile the \texttt{length-simplified} program to a MCX circuit.
Notably, among the optimizers above, only Feynman directly accepts inputs containing MCX gates of arbitrary size, by means of a dedicated pass it provides to convert large MCX gates into Toffoli gates. By contrast, the other optimizers above do not accept MCX gates larger than Toffoli. For these optimizers, we used Feynman to preprocess the circuit into the Clifford+Toffoli or Clifford+CCZ gate sets accepted by each optimizer, without changing its \tgate{}-complexity. Then, we executed each optimizer to generate a Clifford+\tgate{} circuit, and then counted the \tgate{} gates in the resulting circuit.

To the extent possible, we specified configurations that are indicated by prior literature:
\begin{itemize}
\item For Qiskit, we invoked \lstinline{qiskit.compiler.transpile} with \lstinline{optimization_level=3}.
\item For VOQC, we invoked \lstinline{Voqc.Main.optimize_nam}.
\item For Pytket, we invoked two independent modes: \lstinline{pytket.passes.FullPeepholeOptimise} and \lstinline{pytket.passes.ZXGraphlikeOptimisation}, and report them separately below.
\item For Feynman, we invoked two different configurations: \lstinline{feynopt -mctExpand -O2} and \lstinline{feynopt -toCliffordT -O2}, and report them separately below.
\item For QuiZX, we invoked \lstinline{quizx::simplify::full_simp}.
\end{itemize}

\paragraph{Results}
In \Cref{fig:optimizer-efficacy-a}, we plot the \tgate{}-complexity of the \texttt{length-simplified} program at various recursion depths, before and after applying each circuit optimizer.
Of the tested optimizers, 6 of 8 do not asymptotically improve the \tgate{}-complexity of the circuit from quadratic to linear. They achieve 0\% to 71.4\% improvement over the original circuit at depth $n = 10$.\footnote{We note three results in \Cref{fig:optimizer-efficacy-a} that are close but distinct: at depth $n = 10$, VOQC obtains 17530 \tgate{} gates, Pytket ZX obtains 17176 gates, and Feynman \texttt{-toCliffordT} obtains 17166 gates, which is about 2\% fewer than VOQC\@.} Only Feynman \texttt{-mctExpand} and QuiZX obtain linear \tgate{}-complexity, achieving 88.0\% and 93.4\% improvement respectively.
\begin{figure}
\begin{subfigure}{0.5\textwidth}
\begin{tikzpicture}
\begin{axis}[
    xlabel=Recursion depth $n$,
    ylabel=Number of \tgate{} gates,
    scaled y ticks=false,
    ytick distance=2000,
    xtick distance=1,
    ymax=15000,
    width=7cm,
    legend pos=north west,
    legend cell align=left,
    legend style={nodes={scale=0.7, transform shape},cells={align=left},at={(0.5,1)},anchor=north,/tikz/every even column/.append style={column sep=0.4em}},
    legend columns=2,
    domain=2:10,
    samples=100,
    tick label style={font=\small},
    every axis plot/.append style={thick}
]
\addplot[Dark2-A,-,mark=square*] plot coordinates {
    (2,1638)
    (3,4046)
    (4,7854)
    (5,13062)
    (6,19670)
    (7,27678)
    (8,37086)
    (9,47894)
    (10,60102)
};
\addlegendentry{Original}
\addplot[Dark2-G,-,mark=+] plot coordinates {
    (2,1453)
    (3,3583)
    (4,6947)
    (5,11511)
    (6,17275)
    (7,24239)
    (8,32403)
    (9,41767)
    (10,52331)
};
\addlegendentry{Qiskit}
\addplot[Dark2-H,-,mark=*] plot coordinates {
    (2,630)
    (3,1876)
    (4,4284)
    (5,8092)
    (6,13300)
    (7,19908)
    (8,27916)
    (9,37324)
    (10,48132)
};
\addlegendentry{CN alone}
\addplot[Dark2-C,-,mark=diamond] plot coordinates {
    (2,764)
    (3,1570)
    (4,2376)
    (5,3182)
    (6,3986)
    (7,4792)
    (8,5594)
    (9,6400)
    (10,7206)
};
\addlegendentry{F. \texttt{-mctExpand}}
\addplot[Dark2-B,-,mark=Mercedes star flipped] plot coordinates {
    (2,1134)
    (3,1876)
    (4,2618)
    (5,3360)
    (6,4102)
    (7,4844)
    (8,5586)
    (9,6328)
    (10,7070)
};
\addlegendentry{CF alone}
\addplot[Dark2-D,-,mark=asterisk] plot coordinates {
    (2,430)
    (3,848)
    (4,1266)
    (5,1692)
    (6,2112)
    (7,2532)
    (8,2950)
    (9,3372)
    (10,3792)
};
\addlegendentry{QuiZX}
\addplot[Dark2-E,-,mark=pentagon] plot coordinates {
    (2,490)
    (3,756)
    (4,1022)
    (5,1288)
    (6,1554)
    (7,1820)
    (8,2086)
    (9,2352)
    (10,2618)
};
\addlegendentry{Spire (Ours)}
\addplot[Dark2-F,-,mark=triangle] plot coordinates {
    (2,348)
    (3,534)
    (4,720)
    (5,906)
    (6,1092)
    (7,1278)
    (8,1464)
    (9,1650)
    (10,1836)
};
\addlegendentry{Spire + F. \texttt{-mctExpand}}
\end{axis}
\end{tikzpicture}%
\caption{Program-level optimizations in Spire.} \label{fig:optimizer-efficacy-b}
\end{subfigure}%
\hspace*{\fill}
\begin{subfigure}{0.5\textwidth}
\begin{tikzpicture}
\begin{axis}[
    xlabel=Recursion depth $n$,
    scaled y ticks=false,
    ytick distance=10000,
    xtick distance=1,
    width=7cm,
    legend pos=north west,
    legend cell align=left,
    legend style={fill=none,draw=none,nodes={scale=0.7, transform shape}},
    domain=2:10,
    samples=100,
    tick label style={font=\small},
    every axis plot/.append style={thick}
]
\addplot[Dark2-A,-,mark=square*] plot coordinates {
    (2,1638)
    (3,4046)
    (4,7854)
    (5,13062)
    (6,19670)
    (7,27678)
    (8,37086)
    (9,47894)
    (10,60102)
};
\addlegendentry{Original program}
\addplot[Dark2-G,-,mark=+] plot coordinates {
    (2,1453)
    (3,3583)
    (4,6947)
    (5,11511)
    (6,17275)
    (7,24239)
    (8,32403)
    (9,41767)
    (10,52331)
};
\addlegendentry{Qiskit}
\addplot[Set2-B,-,mark=star] plot coordinates {
    (2,972)
    (3,2506)
    (4,4594)
    (5,7908)
    (6,12422)
    (7,18136)
    (8,25050)
    (9,33164)
    (10,42478)
};
\addlegendentry{Pytket peephole}
\addplot[Set2-A,-,mark=o] plot coordinates {
    (2,552)
    (3,1262)
    (4,2384)
    (5,3912)
    (6,5838)
    (7,8158)
    (8,10884)
    (9,14010)
    (10,17530)
};
\addlegendentry{VOQC}
\addplot[Set2-C,-,mark=x] plot coordinates {
    (2,550)
    (3,1258)
    (4,2328)
    (5,3812)
    (6,5682)
    (7,7954)
    (8,10620)
    (9,13676)
    (10,17176)
};
\addlegendentry{Pytket ZX}
\addplot[Set2-D,-,mark=square] plot coordinates {
    (2,548)
    (3,1258)
    (4,2328)
    (5,3806)
    (6,5678)
    (7,7950)
    (8,10620)
    (9,13694)
    (10,17166)
};
\addlegendentry{F. \texttt{-toCliffordT}}
\addplot[Dark2-C,-,mark=diamond] plot coordinates {
    (2,764)
    (3,1570)
    (4,2376)
    (5,3182)
    (6,3986)
    (7,4792)
    (8,5594)
    (9,6400)
    (10,7206)
};
\addlegendentry{F. \texttt{-mctExpand}}
\addplot[Dark2-D,-,mark=asterisk] plot coordinates {
    (2,430)
    (3,848)
    (4,1266)
    (5,1692)
    (6,2112)
    (7,2532)
    (8,2950)
    (9,3372)
    (10,3792)
};
\addlegendentry{QuiZX}
\end{axis}
\end{tikzpicture}
\caption{Existing quantum circuit optimizers.} \label{fig:optimizer-efficacy-a}
\end{subfigure}

\caption{\tgate{}-complexity of \texttt{length-simplified} after program-level optimizations and quantum circuit optimizers. CF, CN, and F.\ abbreviate conditional flattening, conditional narrowing, and Feynman respectively.} \label{fig:optimizer-efficacy}
\end{figure}
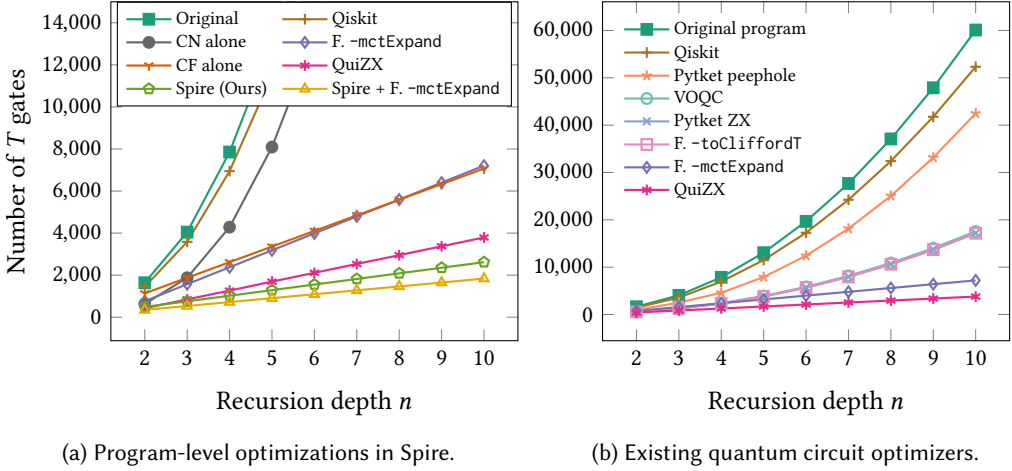

We do not plot Quartz and QUESO because the versions of these two optimizers available at the start of our experimentation require several hours to terminate for most of our benchmarks, even when the user specifies a 1-hour timeout. The partial results we obtained indicate that the \tgate{}-complexity of their output circuits is quadratic rather than linear. At depth $n = 5$, Quartz achieves 37\% improvement in \tgate{}-complexity, and at $n = 2$, QUESO achieves 13\% improvement. For more details on the methodology and results for these optimizers, please see \Cref{sec:quartz-queso}.

Notably, only select configurations of Feynman obtain asymptotic improvement in \tgate{}-complexity. In \Cref{fig:optimizer-efficacy-a}, we plot the \tgate{}-complexity Feynman obtains using two different flags: \texttt{-toCliffordT}, which is quadratic, and \texttt{-mctExpand}, which is linear. The difference is that the first configuration translates the circuit to the Clifford+\tgate{} gate set before applying gate simplifications, whereas the second simplifies the original circuit in terms of Toffoli gates before translating to Clifford+\tgate{}.

Spire's program-level optimizations also synergize with existing quantum circuit optimizers to achieve better results than either alone. In \Cref{fig:optimizer-efficacy-b}, we also plot the \tgate{}-complexity of applying Spire's optimizations followed by Feynman \texttt{-mctExpand}, which for \texttt{length-simplified} achieves 96.9\% improvement over the original program compared to 88.0\% for Feynman alone. In \Cref{tbl:performance}, we summarize the \tgate{}-complexity improvement of running either Feynman \texttt{-mctExpand} or QuiZX after Spire's optimizations. The latter achieves 98.1\% improvement compared to 93.4\% for QuiZX alone.

In \Cref{sec:more-evaluation}, we present more results showing that when the conditional narrowing optimization is used before Feynman or QuiZX, the output circuits are better than Feynman or QuiZX alone. These results indicate that even when a circuit optimizer achieves asymptotically efficient circuits, it can still benefit from the constant improvements provided by conditional narrowing.

\subsection{RQ4: Effect of Optimizations on Compilation Time}

\begin{enumerate}
    \item [\emph{RQ4.}] What is the effect on compilation time of performing the program-level optimizations, and how does it compare to the effect on compilation time of quantum circuit optimizers?
\end{enumerate}

\paragraph{Methodology}
To answer this question, we measured the time taken by Spire to emit a circuit for both the \texttt{length} and \texttt{length-simplified} programs, with program-level optimizations enabled or disabled. Then, we measured the time taken by Feynman \texttt{-mctExpand} and QuiZX to optimize the circuit emitted by Spire with optimizations enabled or disabled.
All timings are reported as the mean and standard error of 5 runs on one core of an AMD Threadripper 1920X and 32 GB of RAM\@.

\paragraph{Results}
Given the original \texttt{length} program at depth $n = 10$, Spire takes \SI{0.08}{\second} to emit a circuit without performing program-level optimizations, and \SI{0.05}{\second} with the optimizations.
The reason that compilation time decreases is that while the optimizations take tens of microseconds to perform, they enable the compiler to save significant time generating controls in the output circuit.

In \Cref{tbl:performance}, we summarize the performance of Feynman \texttt{-mctExpand} and QuiZX on each circuit.
When executed on the original \texttt{length} circuit, Feynman takes \SI{121.96 \pm 0.08}{\second}; QuiZX exceeds available memory and does not terminate after 72 hours.
By comparison, Spire alone yields comparable circuits in \SI{0.05}{\second}, which is $2400\times$ faster than Feynman.
When Spire's optimizations are run before Feynman, the smaller input circuit that is produced enables Feynman to take only \SI{17.05 \pm 0.01}{\second}, a 7$\times$ improvement.
These circuits remain large enough for QuiZX to be memory constrained.

\begin{table}
\caption{Summary of comparison and synergy between Spire and existing circuit optimizers, in terms of \tgate{}-complexity reduction and compilation time. Figures are given for both \texttt{length} and \texttt{length-simplified} programs at depth $n = 10$. We show only optimizers that achieve linear \tgate{}-complexity.}
\resizebox{\textwidth}{!}{
\begin{tabular}{ l c c c c c }
\toprule
& \multicolumn{2}{c}{\texttt{length-simplified}} & & \multicolumn{2}{c}{\texttt{length}} \\
\cmidrule{2-3} \cmidrule{5-6}
& \tgate{} Reduction & Compile Time & & \tgate{} Reduction & Compile Time \\
\midrule
Feynman \texttt{-mctExpand} & 88.0\% & \SI{0.54 \pm 0.00}{\second} & & 92.2\% & \SI{121.96 \pm 0.08}{\second} \\
QuiZX & 93.4\% & \SI{3510.80 \pm 1.97}{\second} & & \multicolumn{2}{c}{\textcolor{gray}{(consumes >32 GB RAM)}} \\
Spire (Ours) & 95.6\% & \textbf{\SI{0.01 \pm 0.00}{\second}} & & 92.8\% & \textbf{\SI{0.05}{\second}} \\
Spire + Feynman \texttt{-mctExpand} & 96.9\% & \SI{0.08 \pm 0.00}{\second} & & \textbf{94.7\%} & \SI{17.05 \pm 0.01}{\second} \\
Spire + QuiZX & \textbf{98.1\%} & \SI{1.18 \pm 0.00}{\second} & & \multicolumn{2}{c}{\textcolor{gray}{(consumes >32 GB RAM)}} \\
\bottomrule
\end{tabular}%
}%
\label{tbl:performance}%
\vspace*{-4pt}%
\end{table}

\subsection{Discussion} \label{sec:discussion}

First, our results indicate that VOQC, Quartz, Pytket ZX, and Feynman \texttt{-toCliffordT} obtain an intermediate result in \tgate{}-complexity that is higher than Feynman \texttt{-mctExpand} and QuiZX and lower than Qiskit and Pytket peephole. An explanation consistent with these results is that the first four optimizers implement the optimization of rotation merging~\citep{nam2018} that merges phase rotations across an arbitrary number of gates, whereas the last two do not.

Next, one explanation for why Feynman \texttt{-mctExpand} and QuiZX reduce the asymptotic \tgate{}-complexity of the program in our results is that they successfully identify and exploit the structure of Toffoli gates. Specifically, Feynman \texttt{-mctExpand} first cancels Toffoli gates in the circuit before translating them to Clifford+\tgate{} gates. Meanwhile, QuiZX uses an internal representation known as ZX-calculus~\citep{kissinger2020} that discovers long-range circuit structure at the expense of compile time, which in \Cref{tbl:performance} is 14$\times$--6500$\times$ longer than Feynman.\footnote{We note that Pytket ZX does not reduce the asymptotic \tgate{}-complexity in our results even though it also uses the ZX-calculus. This discrepancy may be due to different optimization choices taken by Pytket and QuiZX.}

By contrast, Qiskit, Pytket, VOQC, Quartz, and QUESO do not perform rewrites at the level of Toffoli gates. They instead either require the input to consist only of Clifford+\tgate{} gates, or decompose all Toffoli gates in the input to them. As we show next, the value of the structure of Toffoli gates is that cancelling Toffoli gates can capture the effect of conditional flattening. By contrast, cancelling adjacent gates no longer captures this effect after Toffoli gates are lowered to Clifford+\tgate{} gates.

\paragraph{Conditional Flattening}
In \Cref{fig:nested-conditionals-naive}, we present a sub-program of \Cref{fig:simple-if} with only the assignment to \texttt{a} that is controlled by \texttt{x}, \texttt{y}, and \texttt{z}, and its corresponding sub-circuit from \Cref{fig:simple-if-circuit-full}. We also depict the result of decomposing its MCX gates to Toffoli gates via the rule in \Cref{fig:mcx-decomposition}.
Compared to \Cref{fig:simple-if-circuit-optimized}, the final circuit in \Cref{fig:nested-conditionals-naive} incurs additional \tgate{}-complexity from Toffoli gates.

Now suppose that the final circuit in \Cref{fig:nested-conditionals-naive} is given to a quantum circuit optimizer. In general, to recover an asymptotically efficient circuit, the optimizer must eliminate all but a small number of Toffoli gates. In \Cref{fig:nested-conditionals-naive}, it must eliminate the redundant self-inverse gates in \textcolor{gray}{gray}.

The problem is that adjacent Toffoli gates become difficult to identify when Toffoli gates have been decomposed into Clifford+\tgate{} gates.
In \Cref{fig:toffoli-non-cancellability}, we depict the decomposition of a pair of Toffoli gates into a sequence of 32 Clifford+\tgate{} gates by the standard rule in \Cref{fig:toffoli-decomposition}. Because the decomposition of each Toffoli is asymmetric, the circuit optimizer cannot reduce this sequence to an empty circuit by merely cancelling adjacent Clifford+\tgate{} gates.\footnote{In particular, a GitHub issue open since 2021 (\url{https://github.com/Qiskit/qiskit/issues/6740}) describes the inability of Qiskit to optimize away the Clifford+\tgate{} sequence corresponding to adjacent Toffoli gates as depicted in \Cref{fig:toffoli-non-cancellability}.} To sidestep this problem, Feynman and \citet{maslov2005,nam2018} perform rewrites on Toffoli gates before they are decomposed to Clifford+\tgate{} gates. For other alternative approaches, see \Cref{sec:optimization-alternatives}.

\paragraph{Conditional Narrowing}
Even worse, conditional narrowing cannot be captured by a quantum circuit optimizer that acts on gate windows of any finite size. The rule in \Cref{fig:conditional-inverses-circuit} removes control bits on $\circuit{s_1}$ and $\smash{\circuit{s_1}^\dagger}$ when these sequences have been identified as inverses. The problem is that $\circuit{s_2}$ lies between them and can be of arbitrary length, meaning that without program structure, discovering the relationship between $\circuit{s_1}$ and $\smash{\circuit{s_1}^\dagger}$ requires a window of unbounded size.

This fact contributes to an explanation for why in \Cref{sec:eval-circuit-optimizers}, some of the tested circuit optimizers leave behind some fraction of \tgate{} gates that are otherwise captured by the conditional narrowing optimization. In principle, a circuit optimizer can capture conditional narrowing using rewrites over an unbounded number of gates, such as by an appropriate implementation of rotation merging.

\newcommand{\graynot}{\gate[style={draw=gray},label style=gray]{\notgate}}%
\newcommand{\grayctrl}[1]{\ctrl[style={gray,fill=gray},wire style=gray]{#1}}%
\newcommand{\grayslice}{\slice[style={gray,draw=gray}]{}}%
\newcommand{\grpi}{\gategroup[2,steps=3,style={dashed,inner xsep=2pt,inner ysep=2pt},background,label style={label position=above,anchor=north,yshift=2.25cm}]{\lstinline[style=color]{let a <- not t;}}}%
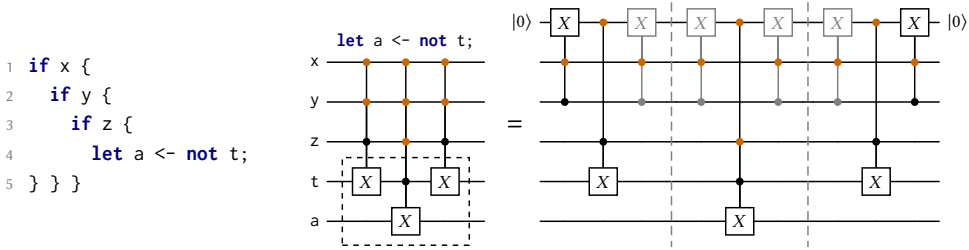
\begin{figure}
\begin{minipage}[c]{0.3\textwidth}
\begin{lstlisting}
if x {
  if y {
    if z {
      let a <- not t;
} } }
\end{lstlisting}
\end{minipage}%
\begin{minipage}[c]{0.65\textwidth}
\vspace*{-0.51em}%
\resizebox{\textwidth}{!}{
\begin{quantikz}[row sep={0.75cm,between origins}, classical gap=0.075cm, column sep={0.75cm,between origins}]
\ghost{T} \\
\lstick{\texttt{x}} & \redctrl{2}          & \redctrl{3}     & \redctrl{2}     & \\
\lstick{\texttt{y}} & \redctrl{2}          & \redctrl{3}     & \redctrl{2}     & \\
\lstick{\texttt{z}} & \ctrl{1}             & \redctrl{2}     & \ctrl{1}        & \\
\lstick{\texttt{t}} & \gate{\notgate}\grpi & \ctrl{1}        & \gate{\notgate} & \\
\lstick{\texttt{a}} &                      & \gate{\notgate} &                 &
\end{quantikz}
\hspace*{0.6em}\raisebox{-0.3em}{\huge =}\hspace*{-1em}%
\begin{quantikz}[column sep=0.2cm,row sep={0.75cm,between origins},classical gap=0.075cm]
\lstick{$\ket{0}$} & \gate{\notgate} & \redctrl{4}     & \graynot      & \grayslice & & \graynot      & \redctrl{5}     & \graynot      & \grayslice & & \graynot      & \redctrl{4}     & \gate{\notgate} & \rstick{$\ket{0}$} \\
                   & \redctrl{-1}    &                 & \redctrl{-1}  &            & & \redctrl{-1}  &                 & \redctrl{-1}  &            & & \redctrl{-1}  &                 & \redctrl{-1}    & \\
                   & \ctrl{-2}       &                 & \grayctrl{-2} &            & & \grayctrl{-2} &                 & \grayctrl{-2} &            & & \grayctrl{-2} &                 & \ctrl{-2}       & \\
                   &                 & \ctrl{1}        &               &            & &               & \redctrl{2}     &               &            & &               & \ctrl{1}        &                 & \\
                   &                 & \gate{\notgate} &               &            & &               & \ctrl{1}        &               &            & &               & \gate{\notgate} &                 & \\
                   &                 &                 &               &            & &               & \gate{\notgate} &               &            & &               &                 &                 &
\end{quantikz}
}%
\end{minipage}
\setlength{\abovecaptionskip}{5pt}
\setlength{\belowcaptionskip}{-8pt}
\caption{Direct compilation of nested conditionals to a Clifford+Toffoli circuit using the MCX decomposition in \Cref{fig:mcx-decomposition}. The redundant Toffoli gates (\textcolor{gray}{gray}) must be eliminated to obtain an efficient circuit.} \label{fig:nested-conditionals-naive}
\end{figure}

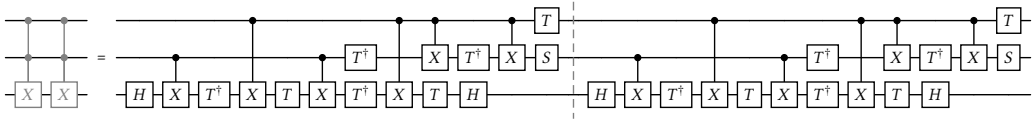
\begin{figure}
\resizebox{\textwidth}{!}{
\begin{quantikz}[column sep=0.2cm, row sep={0.75cm,between origins}]
& \grayctrl{2} & \grayctrl{2} & \ghost{T} \\
& \grayctrl{1} & \grayctrl{1} & \\
& \graynot & \graynot &
\end{quantikz}
=
\begin{quantikz}[column sep=0.2cm,row sep={0.75cm,between origins}]
& & & & \ctrl{2} & & & & \ctrl{2} & \ctrl{1} & & \ctrl{1} & \gate{T} & \grayslice & & & & & \ctrl{2} & & & & \ctrl{2} & \ctrl{1} & & \ctrl{1} & \gate{T} & \\
& & \ctrl{1} & & & & \ctrl{1} & \gate{T\smash{{}^\dagger}} & & \gate{\notgate} & \gate{T\smash{{}^\dagger}} & \gate{\notgate} & \gate{S} & & & & \ctrl{1} & & & & \ctrl{1} & \gate{T\smash{{}^\dagger}} & & \gate{\notgate} & \gate{T\smash{{}^\dagger}} & \gate{\notgate} & \gate{S} & \\
& \gate{H} & \gate{\notgate} & \gate{T\smash{{}^\dagger}} & \gate{\notgate} & \gate{T} & \gate{\notgate} & \gate{T\smash{{}^\dagger}} & \gate{\notgate} & \gate{T} & \gate{H} & & & & & \gate{H} & \gate{\notgate} & \gate{T\smash{{}^\dagger}} & \gate{\notgate} & \gate{T} & \gate{\notgate} & \gate{T\smash{{}^\dagger}} & \gate{\notgate} & \gate{T} & \gate{H} & & &
\end{quantikz}
}
\setlength{\abovecaptionskip}{3pt}
\setlength{\belowcaptionskip}{-7pt}
\caption{Two adjacent Toffoli gates after the standard Clifford+\tgate{} decomposition in \Cref{fig:toffoli-decomposition}. Though equal to the empty identity circuit, this gate sequence cannot be reduced to such by adjacent gate cancellations alone.} \label{fig:toffoli-non-cancellability}
\end{figure}

\section{Future Directions}

\paragraph{Benchmarking Optimizers}
One explanation for why certain existing circuit optimizers do not asymptotically reduce the \tgate{}-complexity of programs with control flow is that such programs are not a focus of the benchmarks on which optimizers have been primarily evaluated. The typical benchmarks, as in \citet[Appendix F]{xu2023}, are circuits with up to $10^3$ \tgate{} gates that are built directly from logic gates, not compiled from quantum programs with control flow. This work does not evaluate on these benchmarks as they do not exhibit the asymptotic behavior of interest.

Instead, this work studies the asymptotic behavior of families of circuits that are compiled from programming abstractions and large enough to be relevant to the regime of practical quantum advantage. For example, \citet{gidney2021} project that $4 \cdot 10^8$ Toffoli gates are necessary to break 1024-bit RSA, and $3 \cdot 10^{10}$ Toffoli gates to break elliptic curve cryptography. At such scales, optimization techniques that are profitable and tractable for small circuits, such as small peepholes and ZX-calculus, become less effective and would benefit from higher-level program structure.

Consequently, it is important future work to develop more explicit implementations of quantum algorithms to serve as large-scale benchmarks for quantum compilation that may reveal other quantum programming abstractions whose costs must be considered and mitigated.

\paragraph{Architectural Bottlenecks}
Aside from \tgate{}-complexity, error-corrected architectures are also constrained by the number of qubits used by a computation. Conditional narrowing does not affect qubit usage, as it only removes control bits from statements. In \Cref{sec:optimization-costs}, we show that given a compiler that uses the MCX decomposition in \Cref{fig:mcx-decomposition,fig:toffoli-decomposition}, conditional flattening introduces no more than $O(1)$ extra qubits in the circuit for the optimized program as compared to the unoptimized program.
The reason is that the new temporary variable $x'$ from the rule in \Cref{sec:conditional-flattening} reuses a qubit that would exist in the compiled circuit for the program even without conditional flattening. This extra qubit -- marked with $\ket{0}$ in \Cref{fig:mcx-decomposition} -- is introduced when the compiler decomposes all MCX to Clifford+\tgate{} gates as needed for a program regardless of conditional flattening.

For sake of thoroughness, we note that alternatives to \Cref{fig:mcx-decomposition} exist that use no extra qubits but use more \tgate{} gates~\citep[Section 7]{barenco1995}.
An important future direction is to explore the trade-offs of different MCX decompositions, and simultaneously optimize \tgate{}-complexity alongside qubit complexity and other metrics such as \emph{\tgate{}-depth} and \emph{quantum volume}~\citep{cross2019}.

Though this work focuses on the widely recognized bottleneck of \tgate{}-complexity on the surface code architecture, the asymptotic costs it presents arise on any error-corrected quantum computer. Fundamentally, the Eastin-Knill theorem~\citep{eastin2009} states that no quantum error-correcting code can \emph{transversely}, i.e.\ natively and efficiently, implement a gate set that is universal for quantum computation. Some gate -- in the surface code, the \tgate{} gate -- is always a bottleneck.

For example, while Reed-Muller codes support an efficient \tgate{} gate, they give up the Hadamard gate in exchange and are thus not universal for quantum computation~\citep{zeng2007}.
In general, strong evidence~\citep{newman2018,oconnor2018} indicates that a Toffoli or MCX gate will act as a performance bottleneck under any quantum error-correcting code.

\paragraph{Other Quantum Architectures}
Apart from the surface code, the abstraction cost of control flow also occurs broadly on hardware architectures in which MCX gates must be decomposed to native gates. For example, on an architecture with only single and two-qubit gates such as CNOT, an MCX gate with many control bits compiles to a proportional number of CNOT gates, making it important to study further how to reduce the performance impact of two-qubit gates~\citep{maslov2016b}.

\section{Related Work} \label{sec:related-work}

\paragraph{\tgate{}-Complexity Optimization}
Optimizations for \tgate{}-complexity have long been investigated in the literature of quantum algorithms. For example, instances of conditional narrowing and conditional flattening are used by physical simulation algorithms~\citep[Figures 1, 6, and 7]{babbush2018} to save control bits during state preparation and Hamiltonian selection respectively.

Researchers have proposed quantum compilers featuring variants of conditional narrowing~\citep{steiger2018,ittah2022} and separately of conditional flattening~\citep{seidel2022}. Novel to this work is our unification of both optimizations as syntax rewrite rules, which produce high-level programs that can be analyzed by the cost model.
Other novel contributions in this work are that we identify that these optimizations can mitigate the asymptotic slowdown caused by control flow, and empirically evaluate their effectiveness and speed relative to existing circuit optimizers.

\paragraph{Quantum Resource Analysis}
Researchers have proposed frameworks~\citep{liu2022,olmedo2019,avanzini2022} to analyze the expected runtime of a quantum program.
Unlike our cost model, prior frameworks do not support reasoning for abstractions for control flow in superposition such as the quantum \texttt{if}-statement.
In order to analyze a program featuring control flow, they require the developer to first lower all abstractions to explicit quantum logic gates.

However, as identified in this work, it is precisely this compilation process itself that introduces asymptotic overhead in \tgate{}-complexity. Our cost model and optimizations enable the developer to identify and mitigate the costs without compiling the program to an asymptotically large circuit.

\section{Conclusion}

The practical realization of quantum algorithms requires designers of programming languages and compilers to reconcile the expressive power of programming abstractions with the performance bottlenecks of error correction. As this work shows, control flow incurs \tgate{}-complexity costs that are significant yet can be mitigated by simple optimizations.
Our work holds out the promise of enabling both expressive and efficient control flow abstractions in quantum programming.

Our work additionally demonstrates the value of a deep study of the interface between quantum programs and error-corrected hardware. This study and our results illuminate a path to a future that combines powerful techniques from classical compilers with search-based optimization of circuits to increase the efficiency of both current and future quantum software.

\section*{Data Availability Statement}

The artifact for this paper, including source code, benchmark programs, and evaluation package, is available on Zenodo~\citep{artifact}.

\section*{Acknowledgements}
We thank Ellie Cheng, Tian Jin, Jesse Michel, Patrick Rall, Alex Renda, Agnes Villanyi, Logan Weber, and Cambridge Yang for helpful feedback on this work. This work was supported in part by the National Science Foundation (CCF-1751011) and the Sloan Foundation.

\nocite{bravyi2021,maslov2016,selinger2013}
\bibliography{biblio.bib}


\begin{thebibliography}{79}


\ifx \showCODEN    \undefined \def \showCODEN     #1{\unskip}     \fi
\ifx \showDOI      \undefined \def \showDOI       #1{#1}\fi
\ifx \showISBNx    \undefined \def \showISBNx     #1{\unskip}     \fi
\ifx \showISBNxiii \undefined \def \showISBNxiii  #1{\unskip}     \fi
\ifx \showISSN     \undefined \def \showISSN      #1{\unskip}     \fi
\ifx \showLCCN     \undefined \def \showLCCN      #1{\unskip}     \fi
\ifx \shownote     \undefined \def \shownote      #1{#1}          \fi
\ifx \showarticletitle \undefined \def \showarticletitle #1{#1}   \fi
\ifx \showURL      \undefined \def \showURL       {\relax}        \fi
\providecommand\bibfield[2]{#2}
\providecommand\bibinfo[2]{#2}
\providecommand\natexlab[1]{#1}
\providecommand\showeprint[2][]{arXiv:#2}

\bibitem[Aaronson et~al\mbox{.}(2020)]%
        {aaronson2019}
\bibfield{author}{\bibinfo{person}{Scott Aaronson}, \bibinfo{person}{Nai-Hui
  Chia}, \bibinfo{person}{Han-Hsuan Lin}, \bibinfo{person}{Chunhao Wang}, {and}
  \bibinfo{person}{Ruizhe Zhang}.} \bibinfo{year}{2020}\natexlab{}.
\newblock \showarticletitle{On the Quantum Complexity of Closest Pair and
  Related Problems}. In \bibinfo{booktitle}{\emph{Computational Complexity}}.
\newblock
\urldef\tempurl%
\url{https://doi.org/10.4230/LIPIcs.CCC.2020.16}
\showDOI{\tempurl}


\bibitem[Abrams and Lloyd(1997)]%
        {abrams1997}
\bibfield{author}{\bibinfo{person}{Daniel~S. Abrams} {and}
  \bibinfo{person}{Seth Lloyd}.} \bibinfo{year}{1997}\natexlab{}.
\newblock \showarticletitle{Simulation of Many-Body Fermi Systems on a
  Universal Quantum Computer}.
\newblock \bibinfo{journal}{\emph{Phys. Rev. Letters}}  \bibinfo{volume}{79}
  (\bibinfo{date}{Sep} \bibinfo{year}{1997}).
\newblock
Issue 13.
\urldef\tempurl%
\url{https://doi.org/10.1103/PhysRevLett.79.2586}
\showDOI{\tempurl}


\bibitem[Altenkirch and Grattage(2005)]%
        {qml}
\bibfield{author}{\bibinfo{person}{Thorsten Altenkirch} {and}
  \bibinfo{person}{J. Grattage}.} \bibinfo{year}{2005}\natexlab{}.
\newblock \showarticletitle{A Functional Quantum Programming Language}. In
  \bibinfo{booktitle}{\emph{ACM/IEEE Symposium on Logic in Computer Science}}.
\newblock
\urldef\tempurl%
\url{https://doi.org/10.1109/LICS.2005.1}
\showDOI{\tempurl}


\bibitem[Ambainis(2004)]%
        {ambainis2004}
\bibfield{author}{\bibinfo{person}{Andris Ambainis}.}
  \bibinfo{year}{2004}\natexlab{}.
\newblock \showarticletitle{Quantum walk algorithm for element distinctness}.
  In \bibinfo{booktitle}{\emph{IEEE Symposium on Foundations of Computer
  Science}}.
\newblock
\urldef\tempurl%
\url{https://doi.org/10.1109/FOCS.2004.54}
\showDOI{\tempurl}


\bibitem[Ambainis et~al\mbox{.}(2010)]%
        {ambainis2010}
\bibfield{author}{\bibinfo{person}{Andris Ambainis}, \bibinfo{person}{A.~M.
  Childs}, \bibinfo{person}{B.~W. Reichardt}, \bibinfo{person}{R. \v{S}palek},
  {and} \bibinfo{person}{S. Zhang}.} \bibinfo{year}{2010}\natexlab{}.
\newblock \showarticletitle{Any AND-OR Formula of Size N Can Be Evaluated in
  Time $N^{1/2+o(1)}$ on a Quantum Computer}.
\newblock \bibinfo{journal}{\emph{SIAM J. Comput.}} \bibinfo{volume}{39},
  \bibinfo{number}{6} (\bibinfo{year}{2010}).
\newblock
\urldef\tempurl%
\url{https://doi.org/10.1137/080712167}
\showDOI{\tempurl}


\bibitem[Amy(2024)]%
        {feynman}
\bibfield{author}{\bibinfo{person}{Matthew Amy}.}
  \bibinfo{year}{2024}\natexlab{}.
\newblock \bibinfo{title}{Feynman: Quantum circuit analysis toolkit}.
\newblock \bibinfo{howpublished}{\url{https://github.com/meamy/feynman}}.
\newblock


\bibitem[Amy et~al\mbox{.}(2014)]%
        {amy2014}
\bibfield{author}{\bibinfo{person}{Matthew Amy}, \bibinfo{person}{Dmitri
  Maslov}, {and} \bibinfo{person}{Michele Mosca}.}
  \bibinfo{year}{2014}\natexlab{}.
\newblock \showarticletitle{Polynomial-Time T-Depth Optimization of Clifford+T
  Circuits Via Matroid Partitioning}.
\newblock \bibinfo{journal}{\emph{IEEE Transactions on Computer-Aided Design of
  Integrated Circuits and Systems}} \bibinfo{volume}{33}, \bibinfo{number}{10}
  (\bibinfo{year}{2014}).
\newblock
\urldef\tempurl%
\url{https://doi.org/10.1109/TCAD.2014.2341953}
\showDOI{\tempurl}


\bibitem[Avanzini et~al\mbox{.}(2022)]%
        {avanzini2022}
\bibfield{author}{\bibinfo{person}{Martin Avanzini}, \bibinfo{person}{Georg
  Moser}, \bibinfo{person}{Romain Pechoux}, \bibinfo{person}{Simon Perdrix},
  {and} \bibinfo{person}{Vladimir Zamdzhiev}.} \bibinfo{year}{2022}\natexlab{}.
\newblock \showarticletitle{Quantum Expectation Transformers for Cost
  Analysis}. In \bibinfo{booktitle}{\emph{ACM/IEEE Symposium on Logic in
  Computer Science}}.
\newblock
\urldef\tempurl%
\url{https://doi.org/10.1145/3531130.3533332}
\showDOI{\tempurl}


\bibitem[Babbush et~al\mbox{.}(2018)]%
        {babbush2018}
\bibfield{author}{\bibinfo{person}{Ryan Babbush}, \bibinfo{person}{Craig
  Gidney}, \bibinfo{person}{Dominic~W. Berry}, \bibinfo{person}{Nathan Wiebe},
  \bibinfo{person}{Jarrod McClean}, \bibinfo{person}{Alexandru Paler},
  \bibinfo{person}{Austin Fowler}, {and} \bibinfo{person}{Hartmut Neven}.}
  \bibinfo{year}{2018}\natexlab{}.
\newblock \showarticletitle{Encoding Electronic Spectra in Quantum Circuits
  with Linear T Complexity}.
\newblock \bibinfo{journal}{\emph{Phys. Rev. X}} \bibinfo{volume}{8},
  \bibinfo{number}{4} (\bibinfo{date}{Oct} \bibinfo{year}{2018}).
\newblock
\urldef\tempurl%
\url{https://doi.org/10.1103/PhysRevX.8.041015}
\showDOI{\tempurl}


\bibitem[Babbush et~al\mbox{.}(2021)]%
        {babbush2021}
\bibfield{author}{\bibinfo{person}{Ryan Babbush}, \bibinfo{person}{Jarrod~R.
  McClean}, \bibinfo{person}{Michael Newman}, \bibinfo{person}{Craig Gidney},
  \bibinfo{person}{Sergio Boixo}, {and} \bibinfo{person}{Hartmut Neven}.}
  \bibinfo{year}{2021}\natexlab{}.
\newblock \showarticletitle{Focus beyond Quadratic Speedups for Error-Corrected
  Quantum Advantage}.
\newblock \bibinfo{journal}{\emph{PRX Quantum}}  \bibinfo{volume}{2}
  (\bibinfo{date}{Mar} \bibinfo{year}{2021}).
\newblock
Issue 1.
\urldef\tempurl%
\url{https://doi.org/10.1103/PRXQuantum.2.010103}
\showDOI{\tempurl}


\bibitem[Barenco et~al\mbox{.}(1995)]%
        {barenco1995}
\bibfield{author}{\bibinfo{person}{Adriano Barenco},
  \bibinfo{person}{Charles~H. Bennett}, \bibinfo{person}{Richard Cleve},
  \bibinfo{person}{David~P. DiVincenzo}, \bibinfo{person}{Norman Margolus},
  \bibinfo{person}{Peter Shor}, \bibinfo{person}{Tycho Sleator},
  \bibinfo{person}{John~A. Smolin}, {and} \bibinfo{person}{Harald Weinfurter}.}
  \bibinfo{year}{1995}\natexlab{}.
\newblock \showarticletitle{Elementary gates for quantum computation}.
\newblock \bibinfo{journal}{\emph{Phys. Rev. A}} \bibinfo{volume}{52},
  \bibinfo{number}{5} (\bibinfo{date}{Nov} \bibinfo{year}{1995}).
\newblock
\urldef\tempurl%
\url{https://doi.org/10.1103/physreva.52.3457}
\showDOI{\tempurl}


\bibitem[Bell(1964)]%
        {bell1964}
\bibfield{author}{\bibinfo{person}{J.~S. Bell}.}
  \bibinfo{year}{1964}\natexlab{}.
\newblock \showarticletitle{On the Einstein Podolsky Rosen paradox}.
\newblock \bibinfo{journal}{\emph{Physics}}  \bibinfo{volume}{1}
  (\bibinfo{date}{Nov} \bibinfo{year}{1964}).
\newblock
Issue 3.
\urldef\tempurl%
\url{https://doi.org/10.1103/PhysicsPhysiqueFizika.1.195}
\showDOI{\tempurl}


\bibitem[Bennett(1973)]%
        {bennett1973}
\bibfield{author}{\bibinfo{person}{Charles~H. Bennett}.}
  \bibinfo{year}{1973}\natexlab{}.
\newblock \showarticletitle{Logical Reversibility of Computation}.
\newblock \bibinfo{journal}{\emph{IBM Journal of Research and Development}}
  \bibinfo{volume}{17}, \bibinfo{number}{6} (\bibinfo{year}{1973}).
\newblock
\urldef\tempurl%
\url{https://doi.org/10.1147/rd.176.0525}
\showDOI{\tempurl}


\bibitem[Bennett and Brassard(2014)]%
        {bennett2014}
\bibfield{author}{\bibinfo{person}{Charles~H. Bennett} {and}
  \bibinfo{person}{Gilles Brassard}.} \bibinfo{year}{2014}\natexlab{}.
\newblock \showarticletitle{Quantum cryptography: Public key distribution and
  coin tossing}.
\newblock \bibinfo{journal}{\emph{Theoretical Computer Science}}
  \bibinfo{volume}{560} (\bibinfo{year}{2014}).
\newblock
\urldef\tempurl%
\url{https://doi.org/10.1016/j.tcs.2014.05.025}
\showDOI{\tempurl}


\bibitem[Bennett et~al\mbox{.}(1993)]%
        {bennett1993}
\bibfield{author}{\bibinfo{person}{Charles~H. Bennett}, \bibinfo{person}{Gilles
  Brassard}, \bibinfo{person}{Claude Cr\'epeau}, \bibinfo{person}{Richard
  Jozsa}, \bibinfo{person}{Asher Peres}, {and} \bibinfo{person}{William~K.
  Wootters}.} \bibinfo{year}{1993}\natexlab{}.
\newblock \showarticletitle{Teleporting an unknown quantum state via dual
  classical and Einstein-Podolsky-Rosen channels}.
\newblock \bibinfo{journal}{\emph{Phys. Rev. Letters}}  \bibinfo{volume}{70}
  (\bibinfo{date}{Mar} \bibinfo{year}{1993}).
\newblock
Issue 13.
\urldef\tempurl%
\url{https://doi.org/10.1103/PhysRevLett.70.1895}
\showDOI{\tempurl}


\bibitem[Bernstein et~al\mbox{.}(2013)]%
        {bernstein2013}
\bibfield{author}{\bibinfo{person}{Daniel~J. Bernstein},
  \bibinfo{person}{Stacey Jeffery}, \bibinfo{person}{Tanja Lange}, {and}
  \bibinfo{person}{Alexander Meurer}.} \bibinfo{year}{2013}\natexlab{}.
\newblock \showarticletitle{Quantum Algorithms for the Subset-Sum Problem}. In
  \bibinfo{booktitle}{\emph{Post-Quantum Cryptography}}.
\newblock
\urldef\tempurl%
\url{https://doi.org/10.1007/978-3-642-38616-9_2}
\showDOI{\tempurl}


\bibitem[Beverland et~al\mbox{.}(2020)]%
        {beverland2020}
\bibfield{author}{\bibinfo{person}{Michael Beverland}, \bibinfo{person}{Earl
  Campbell}, \bibinfo{person}{Mark Howard}, {and} \bibinfo{person}{Vadym
  Kliuchnikov}.} \bibinfo{year}{2020}\natexlab{}.
\newblock \showarticletitle{Lower bounds on the non-Clifford resources for
  quantum computations}.
\newblock \bibinfo{journal}{\emph{Quantum Science and Technology}}
  \bibinfo{volume}{5}, \bibinfo{number}{3} (\bibinfo{date}{May}
  \bibinfo{year}{2020}).
\newblock
\urldef\tempurl%
\url{https://doi.org/10.1088/2058-9565/ab8963}
\showDOI{\tempurl}


\bibitem[Biamonte et~al\mbox{.}(2017)]%
        {biamonte2017}
\bibfield{author}{\bibinfo{person}{Jacob Biamonte}, \bibinfo{person}{Peter
  Wittek}, \bibinfo{person}{Nicola Pancotti}, \bibinfo{person}{Patrick
  Rebentrost}, \bibinfo{person}{Nathan Wiebe}, {and} \bibinfo{person}{Seth
  Lloyd}.} \bibinfo{year}{2017}\natexlab{}.
\newblock \showarticletitle{Quantum Machine Learning}.
\newblock \bibinfo{journal}{\emph{Nature}} \bibinfo{volume}{549},
  \bibinfo{number}{7671} (\bibinfo{date}{Sep} \bibinfo{year}{2017}).
\newblock
\urldef\tempurl%
\url{https://doi.org/10.1038/nature23474}
\showDOI{\tempurl}


\bibitem[Bichsel et~al\mbox{.}(2020)]%
        {silq}
\bibfield{author}{\bibinfo{person}{Benjamin Bichsel},
  \bibinfo{person}{Maximilian Baader}, \bibinfo{person}{Timon Gehr}, {and}
  \bibinfo{person}{Martin Vechev}.} \bibinfo{year}{2020}\natexlab{}.
\newblock \showarticletitle{Silq: A High-Level Quantum Language with Safe
  Uncomputation and Intuitive Semantics}. In \bibinfo{booktitle}{\emph{ACM
  SIGPLAN Conference on Programming Language Design and Implementation}}.
\newblock
\urldef\tempurl%
\url{https://doi.org/10.1145/3385412.3386007}
\showDOI{\tempurl}


\bibitem[Brassard et~al\mbox{.}(2002)]%
        {brassard2002}
\bibfield{author}{\bibinfo{person}{Gilles Brassard}, \bibinfo{person}{Peter
  H{\o}yer}, \bibinfo{person}{Michele Mosca}, {and} \bibinfo{person}{Alain
  Tapp}.} \bibinfo{year}{2002}\natexlab{}.
\newblock \showarticletitle{Quantum Amplitude Amplification and Estimation}.
\newblock In \bibinfo{booktitle}{\emph{Quantum Computation and Information}}.
  Vol.~\bibinfo{volume}{305}.
\newblock
\urldef\tempurl%
\url{https://doi.org/10.1090/conm/305/05215}
\showDOI{\tempurl}


\bibitem[Bravyi and Kitaev(2005)]%
        {bravyi2005}
\bibfield{author}{\bibinfo{person}{Sergey Bravyi} {and} \bibinfo{person}{Alexei
  Kitaev}.} \bibinfo{year}{2005}\natexlab{}.
\newblock \showarticletitle{Universal quantum computation with ideal Clifford
  gates and noisy ancillas}.
\newblock \bibinfo{journal}{\emph{Phys. Rev. A}} \bibinfo{volume}{71},
  \bibinfo{number}{2} (\bibinfo{date}{Feb} \bibinfo{year}{2005}).
\newblock
\urldef\tempurl%
\url{https://doi.org/10.1103/physreva.71.022316}
\showDOI{\tempurl}


\bibitem[Bravyi et~al\mbox{.}(2021)]%
        {bravyi2021}
\bibfield{author}{\bibinfo{person}{Sergey Bravyi}, \bibinfo{person}{Ruslan
  Shaydulin}, \bibinfo{person}{Shaohan Hu}, {and} \bibinfo{person}{Dmitri
  Maslov}.} \bibinfo{year}{2021}\natexlab{}.
\newblock \showarticletitle{Clifford Circuit Optimization with Templates and
  Symbolic Pauli Gates}.
\newblock \bibinfo{journal}{\emph{{Quantum}}}  \bibinfo{volume}{5}
  (\bibinfo{date}{Nov.} \bibinfo{year}{2021}).
\newblock
\urldef\tempurl%
\url{https://doi.org/10.22331/q-2021-11-16-580}
\showDOI{\tempurl}


\bibitem[Childs et~al\mbox{.}(2018)]%
        {childs2018}
\bibfield{author}{\bibinfo{person}{Andrew~M. Childs}, \bibinfo{person}{Dmitri
  Maslov}, \bibinfo{person}{Yunseong Nam}, \bibinfo{person}{Neil~J. Ross},
  {and} \bibinfo{person}{Yuan Su}.} \bibinfo{year}{2018}\natexlab{}.
\newblock \showarticletitle{Toward the first quantum simulation with quantum
  speedup}.
\newblock \bibinfo{journal}{\emph{Proceedings of the National Academy of
  Sciences}} \bibinfo{volume}{115}, \bibinfo{number}{38} (\bibinfo{date}{Sep}
  \bibinfo{year}{2018}).
\newblock
\urldef\tempurl%
\url{https://doi.org/10.1073/pnas.1801723115}
\showDOI{\tempurl}


\bibitem[Childs et~al\mbox{.}(2007)]%
        {childs2007}
\bibfield{author}{\bibinfo{person}{Andrew~M. Childs}, \bibinfo{person}{Ben~W.
  Reichardt}, \bibinfo{person}{Robert Spalek}, {and} \bibinfo{person}{Shengyu
  Zhang}.} \bibinfo{year}{2007}\natexlab{}.
\newblock \bibinfo{title}{Every NAND formula of size $N$ can be evaluated in
  time $N^{1/2+o(1)}$ on a quantum computer}.
\newblock
\newblock
\urldef\tempurl%
\url{https://doi.org/10.48550/ARXIV.QUANT-PH/0703015}
\showDOI{\tempurl}
\showeprint[arxiv]{0703015}~[quant-ph]


\bibitem[Cross et~al\mbox{.}(2019)]%
        {cross2019}
\bibfield{author}{\bibinfo{person}{Andrew~W. Cross}, \bibinfo{person}{Lev~S.
  Bishop}, \bibinfo{person}{Sarah Sheldon}, \bibinfo{person}{Paul~D. Nation},
  {and} \bibinfo{person}{Jay~M. Gambetta}.} \bibinfo{year}{2019}\natexlab{}.
\newblock \showarticletitle{Validating quantum computers using randomized model
  circuits}.
\newblock \bibinfo{journal}{\emph{Phys. Rev. A}}  \bibinfo{volume}{100}
  (\bibinfo{date}{Sep} \bibinfo{year}{2019}).
\newblock
Issue 3.
\urldef\tempurl%
\url{https://doi.org/10.1103/PhysRevA.100.032328}
\showDOI{\tempurl}


\bibitem[Eastin and Knill(2009)]%
        {eastin2009}
\bibfield{author}{\bibinfo{person}{Bryan Eastin} {and} \bibinfo{person}{Emanuel
  Knill}.} \bibinfo{year}{2009}\natexlab{}.
\newblock \showarticletitle{Restrictions on Transversal Encoded Quantum Gate
  Sets}.
\newblock \bibinfo{journal}{\emph{Phys. Rev. Letters}} \bibinfo{volume}{102},
  \bibinfo{number}{11} (\bibinfo{date}{Mar} \bibinfo{year}{2009}).
\newblock
\urldef\tempurl%
\url{https://doi.org/10.1103/physrevlett.102.110502}
\showDOI{\tempurl}


\bibitem[Farhi et~al\mbox{.}(2014)]%
        {farhi2014}
\bibfield{author}{\bibinfo{person}{Edward Farhi}, \bibinfo{person}{Jeffrey
  Goldstone}, {and} \bibinfo{person}{Sam Gutmann}.}
  \bibinfo{year}{2014}\natexlab{}.
\newblock \bibinfo{title}{A Quantum Approximate Optimization Algorithm}.
\newblock
\newblock
\showeprint[arxiv]{1411.4028}~[quant-ph]


\bibitem[Fowler et~al\mbox{.}(2012)]%
        {fowler2012}
\bibfield{author}{\bibinfo{person}{Austin~G. Fowler}, \bibinfo{person}{Matteo
  Mariantoni}, \bibinfo{person}{John~M. Martinis}, {and}
  \bibinfo{person}{Andrew~N. Cleland}.} \bibinfo{year}{2012}\natexlab{}.
\newblock \showarticletitle{Surface codes: Towards practical large-scale
  quantum computation}.
\newblock \bibinfo{journal}{\emph{Phys. Rev. A}} \bibinfo{volume}{86},
  \bibinfo{number}{3} (\bibinfo{date}{Sep} \bibinfo{year}{2012}).
\newblock
\urldef\tempurl%
\url{https://doi.org/10.1103/PhysRevA.86.032324}
\showDOI{\tempurl}


\bibitem[Gidney and Eker{\aa}(2021)]%
        {gidney2021}
\bibfield{author}{\bibinfo{person}{Craig Gidney} {and} \bibinfo{person}{Martin
  Eker{\aa}}.} \bibinfo{year}{2021}\natexlab{}.
\newblock \showarticletitle{How to factor 2048 bit {RSA} integers in 8 hours
  using 20 million noisy qubits}.
\newblock \bibinfo{journal}{\emph{Quantum}}  \bibinfo{volume}{5}
  (\bibinfo{date}{Apr} \bibinfo{year}{2021}), \bibinfo{pages}{433}.
\newblock
\urldef\tempurl%
\url{https://doi.org/10.22331/q-2021-04-15-433}
\showDOI{\tempurl}


\bibitem[Gidney and Fowler(2019)]%
        {gidney2019}
\bibfield{author}{\bibinfo{person}{Craig Gidney} {and}
  \bibinfo{person}{Austin~G. Fowler}.} \bibinfo{year}{2019}\natexlab{}.
\newblock \showarticletitle{Efficient magic state factories with a catalyzed
  {$|CCZ\rangle$} to {$2|T\rangle$} transformation}.
\newblock \bibinfo{journal}{\emph{{Quantum}}}  \bibinfo{volume}{3}
  (\bibinfo{date}{April} \bibinfo{year}{2019}).
\newblock
\urldef\tempurl%
\url{https://doi.org/10.22331/q-2019-04-30-135}
\showDOI{\tempurl}


\bibitem[{Google Quantum AI}(2023)]%
        {google2023}
\bibfield{author}{\bibinfo{person}{{Google Quantum AI}}.}
  \bibinfo{year}{2023}\natexlab{}.
\newblock \showarticletitle{Suppressing quantum errors by scaling a surface
  code logical qubit}.
\newblock \bibinfo{journal}{\emph{Nature}}  \bibinfo{volume}{614}
  (\bibinfo{date}{02} \bibinfo{year}{2023}).
\newblock
\urldef\tempurl%
\url{https://doi.org/10.1038/s41586-022-05434-1}
\showDOI{\tempurl}


\bibitem[Gottesman(1998)]%
        {gottesman1998}
\bibfield{author}{\bibinfo{person}{Daniel Gottesman}.}
  \bibinfo{year}{1998}\natexlab{}.
\newblock \showarticletitle{Theory of fault-tolerant quantum computation}.
\newblock \bibinfo{journal}{\emph{Phys. Rev. A}} \bibinfo{volume}{57},
  \bibinfo{number}{1} (\bibinfo{date}{Jan} \bibinfo{year}{1998}).
\newblock
\urldef\tempurl%
\url{https://doi.org/10.1103/physreva.57.127}
\showDOI{\tempurl}


\bibitem[Green et~al\mbox{.}(2013)]%
        {quipper}
\bibfield{author}{\bibinfo{person}{Alexander~S. Green},
  \bibinfo{person}{Peter~LeFanu Lumsdaine}, \bibinfo{person}{Neil~J. Ross},
  \bibinfo{person}{Peter Selinger}, {and} \bibinfo{person}{Beno\^{\i}t
  Valiron}.} \bibinfo{year}{2013}\natexlab{}.
\newblock \showarticletitle{Quipper: A Scalable Quantum Programming Language}.
  In \bibinfo{booktitle}{\emph{ACM SIGPLAN Conference on Programming Language
  Design and Implementation}}.
\newblock
\urldef\tempurl%
\url{https://doi.org/10.1145/2491956.2462177}
\showDOI{\tempurl}


\bibitem[Grover(1996)]%
        {grover1996}
\bibfield{author}{\bibinfo{person}{Lov~K. Grover}.}
  \bibinfo{year}{1996}\natexlab{}.
\newblock \showarticletitle{A Fast Quantum Mechanical Algorithm for Database
  Search}. In \bibinfo{booktitle}{\emph{ACM Symposium on Theory of Computing}}.
\newblock
\urldef\tempurl%
\url{https://doi.org/10.1145/237814.237866}
\showDOI{\tempurl}


\bibitem[Hans and Groppe(2022)]%
        {hans2022}
\bibfield{author}{\bibinfo{person}{Julian Hans} {and} \bibinfo{person}{Sven
  Groppe}.} \bibinfo{year}{2022}\natexlab{}.
\newblock \showarticletitle{Silq2Qiskit - Developing a Quantum Language
  Source-to-Source Translator}. In \bibinfo{booktitle}{\emph{International
  Conference on Computer Science and Software Engineering}}.
\newblock
\urldef\tempurl%
\url{https://doi.org/10.1145/3569966.3570114}
\showDOI{\tempurl}


\bibitem[Hietala et~al\mbox{.}(2021)]%
        {hietala2021}
\bibfield{author}{\bibinfo{person}{Kesha Hietala}, \bibinfo{person}{Robert
  Rand}, \bibinfo{person}{Shih-Han Hung}, \bibinfo{person}{Xiaodi Wu}, {and}
  \bibinfo{person}{Michael Hicks}.} \bibinfo{year}{2021}\natexlab{}.
\newblock \showarticletitle{A Verified Optimizer for Quantum Circuits}. In
  \bibinfo{booktitle}{\emph{ACM SIGPLAN Symposium on Principles of Programming
  Languages}}.
\newblock
\urldef\tempurl%
\url{https://doi.org/10.1145/3434318}
\showDOI{\tempurl}


\bibitem[Hoefler et~al\mbox{.}(2023)]%
        {hoefler2023}
\bibfield{author}{\bibinfo{person}{Torsten Hoefler}, \bibinfo{person}{Thomas
  H\"{a}ner}, {and} \bibinfo{person}{Matthias Troyer}.}
  \bibinfo{year}{2023}\natexlab{}.
\newblock \showarticletitle{Disentangling Hype from Practicality: On
  Realistically Achieving Quantum Advantage}.
\newblock \bibinfo{journal}{\emph{Commun. ACM}} \bibinfo{volume}{66},
  \bibinfo{number}{5} (\bibinfo{date}{Apr} \bibinfo{year}{2023}).
\newblock
\urldef\tempurl%
\url{https://doi.org/10.1145/3571725}
\showDOI{\tempurl}


\bibitem[Ittah et~al\mbox{.}(2022)]%
        {ittah2022}
\bibfield{author}{\bibinfo{person}{David Ittah}, \bibinfo{person}{Thomas
  H\"{a}ner}, \bibinfo{person}{Vadym Kliuchnikov}, {and}
  \bibinfo{person}{Torsten Hoefler}.} \bibinfo{year}{2022}\natexlab{}.
\newblock \showarticletitle{QIRO: A Static Single Assignment-Based Quantum
  Program Representation for Optimization}.
\newblock \bibinfo{journal}{\emph{{ACM} Transactions on Quantum Computing}}
  \bibinfo{volume}{3}, \bibinfo{number}{3}, Article \bibinfo{articleno}{14}
  (\bibinfo{date}{Jun} \bibinfo{year}{2022}).
\newblock
\urldef\tempurl%
\url{https://doi.org/10.1145/3491247}
\showDOI{\tempurl}


\bibitem[JavadiAbhari et~al\mbox{.}(2014)]%
        {javadi2014}
\bibfield{author}{\bibinfo{person}{Ali JavadiAbhari}, \bibinfo{person}{Shruti
  Patil}, \bibinfo{person}{Daniel Kudrow}, \bibinfo{person}{Jeff Heckey},
  \bibinfo{person}{Alexey Lvov}, \bibinfo{person}{Frederic~T. Chong}, {and}
  \bibinfo{person}{Margaret Martonosi}.} \bibinfo{year}{2014}\natexlab{}.
\newblock \showarticletitle{ScaffCC: A Framework for Compilation and Analysis
  of Quantum Computing Programs}. In \bibinfo{booktitle}{\emph{ACM Conference
  on Computing Frontiers}}.
\newblock
\urldef\tempurl%
\url{https://doi.org/10.1145/2597917.2597939}
\showDOI{\tempurl}


\bibitem[Jochym-O'Connor et~al\mbox{.}(2018)]%
        {oconnor2018}
\bibfield{author}{\bibinfo{person}{Tomas Jochym-O'Connor},
  \bibinfo{person}{Aleksander Kubica}, {and} \bibinfo{person}{Theodore~J.
  Yoder}.} \bibinfo{year}{2018}\natexlab{}.
\newblock \showarticletitle{Disjointness of Stabilizer Codes and Limitations on
  Fault-Tolerant Logical Gates}.
\newblock \bibinfo{journal}{\emph{Phys. Rev. X}}  \bibinfo{volume}{8}
  (\bibinfo{date}{May} \bibinfo{year}{2018}).
\newblock
Issue 2.
\urldef\tempurl%
\url{https://doi.org/10.1103/PhysRevX.8.021047}
\showDOI{\tempurl}


\bibitem[Kissinger and van~de Wetering(2020)]%
        {kissinger2020}
\bibfield{author}{\bibinfo{person}{Aleks Kissinger} {and} \bibinfo{person}{John
  van~de Wetering}.} \bibinfo{year}{2020}\natexlab{}.
\newblock \showarticletitle{{PyZX}: Large Scale Automated Diagrammatic
  Reasoning}. In \bibinfo{booktitle}{\emph{International Conference on Quantum
  Physics and Logic}}.
\newblock
\urldef\tempurl%
\url{https://doi.org/10.4204/eptcs.318.14}
\showDOI{\tempurl}


\bibitem[Lee et~al\mbox{.}(2021)]%
        {lee2021}
\bibfield{author}{\bibinfo{person}{Joonho Lee}, \bibinfo{person}{Dominic~W.
  Berry}, \bibinfo{person}{Craig Gidney}, \bibinfo{person}{William~J. Huggins},
  \bibinfo{person}{Jarrod~R. McClean}, \bibinfo{person}{Nathan Wiebe}, {and}
  \bibinfo{person}{Ryan Babbush}.} \bibinfo{year}{2021}\natexlab{}.
\newblock \showarticletitle{Even More Efficient Quantum Computations of
  Chemistry Through Tensor Hypercontraction}.
\newblock \bibinfo{journal}{\emph{{PRX} Quantum}} \bibinfo{volume}{2},
  \bibinfo{number}{3} (\bibinfo{date}{Jul} \bibinfo{year}{2021}).
\newblock
\urldef\tempurl%
\url{https://doi.org/10.1103/prxquantum.2.030305}
\showDOI{\tempurl}


\bibitem[Leymann and Barzen(2020)]%
        {leymann2020}
\bibfield{author}{\bibinfo{person}{Frank Leymann} {and}
  \bibinfo{person}{Johanna Barzen}.} \bibinfo{year}{2020}\natexlab{}.
\newblock \showarticletitle{The bitter truth about gate-based quantum
  algorithms in the NISQ era}.
\newblock \bibinfo{journal}{\emph{Quantum Science and Technology}}
  \bibinfo{volume}{5}, \bibinfo{number}{4} (\bibinfo{date}{Sep}
  \bibinfo{year}{2020}).
\newblock
\urldef\tempurl%
\url{https://doi.org/10.1088/2058-9565/abae7d}
\showDOI{\tempurl}


\bibitem[Liu et~al\mbox{.}(2022)]%
        {liu2022}
\bibfield{author}{\bibinfo{person}{Junyi Liu}, \bibinfo{person}{Li Zhou},
  \bibinfo{person}{Gilles Barthe}, {and} \bibinfo{person}{Mingsheng Ying}.}
  \bibinfo{year}{2022}\natexlab{}.
\newblock \showarticletitle{Quantum Weakest Preconditions for Reasoning about
  Expected Runtimes of Quantum Programs}. In \bibinfo{booktitle}{\emph{ACM/IEEE
  Symposium on Logic in Computer Science}}.
\newblock
\urldef\tempurl%
\url{https://doi.org/10.1145/3531130.3533327}
\showDOI{\tempurl}


\bibitem[Lloyd et~al\mbox{.}(2014)]%
        {lloyd2015}
\bibfield{author}{\bibinfo{person}{Seth Lloyd}, \bibinfo{person}{Silvano
  Garnerone}, {and} \bibinfo{person}{Paolo Zanardi}.}
  \bibinfo{year}{2014}\natexlab{}.
\newblock \showarticletitle{Quantum algorithms for topological and geometric
  analysis of big data}.
\newblock \bibinfo{journal}{\emph{Nature Communications}} \bibinfo{volume}{7},
  \bibinfo{number}{10138} (\bibinfo{date}{Aug} \bibinfo{year}{2014}).
\newblock
\urldef\tempurl%
\url{https://doi.org/10.1038/ncomms10138}
\showDOI{\tempurl}


\bibitem[Low et~al\mbox{.}(2018)]%
        {low2018}
\bibfield{author}{\bibinfo{person}{Guang~Hao Low}, \bibinfo{person}{Vadym
  Kliuchnikov}, {and} \bibinfo{person}{Luke Schaeffer}.}
  \bibinfo{year}{2018}\natexlab{}.
\newblock \bibinfo{title}{Trading T-gates for dirty qubits in state preparation
  and unitary synthesis}.
\newblock
\newblock
\urldef\tempurl%
\url{https://doi.org/10.48550/arXiv.1812.00954}
\showDOI{\tempurl}
\showeprint[arxiv]{1812.00954}~[quant-ph]


\bibitem[Maslov(2016a)]%
        {maslov2016}
\bibfield{author}{\bibinfo{person}{Dmitri Maslov}.}
  \bibinfo{year}{2016}\natexlab{a}.
\newblock \showarticletitle{Advantages of using relative-phase Toffoli gates
  with an application to multiple control Toffoli optimization}.
\newblock \bibinfo{journal}{\emph{Phys. Rev. A}}  \bibinfo{volume}{93}
  (\bibinfo{date}{Feb} \bibinfo{year}{2016}).
\newblock
Issue 2.
\urldef\tempurl%
\url{https://doi.org/10.1103/PhysRevA.93.022311}
\showDOI{\tempurl}


\bibitem[Maslov(2016b)]%
        {maslov2016b}
\bibfield{author}{\bibinfo{person}{Dmitri Maslov}.}
  \bibinfo{year}{2016}\natexlab{b}.
\newblock \showarticletitle{Optimal and asymptotically optimal NCT reversible
  circuits by the gate types}.
\newblock \bibinfo{journal}{\emph{Quantum Information and Computation}}
  \bibinfo{volume}{16}, \bibinfo{number}{13-14} (\bibinfo{date}{Oct}
  \bibinfo{year}{2016}).
\newblock


\bibitem[Maslov et~al\mbox{.}(2005)]%
        {maslov2005}
\bibfield{author}{\bibinfo{person}{D. Maslov}, \bibinfo{person}{G.W. Dueck},
  {and} \bibinfo{person}{D.M. Miller}.} \bibinfo{year}{2005}\natexlab{}.
\newblock \showarticletitle{Toffoli network synthesis with templates}.
\newblock \bibinfo{journal}{\emph{IEEE Transactions on Computer-Aided Design of
  Integrated Circuits and Systems}} \bibinfo{volume}{24}, \bibinfo{number}{6}
  (\bibinfo{year}{2005}).
\newblock
\urldef\tempurl%
\url{https://doi.org/10.1109/TCAD.2005.847911}
\showDOI{\tempurl}


\bibitem[Mosca(2016)]%
        {dotQC}
\bibfield{author}{\bibinfo{person}{Michele Mosca}.}
  \bibinfo{year}{2016}\natexlab{}.
\newblock \bibinfo{title}{Specification of the .qc file format}.
\newblock
  \bibinfo{howpublished}{\url{https://circuits.qsoft.iqc.uwaterloo.ca/about/spec/}}.
\newblock


\bibitem[Nam et~al\mbox{.}(2018)]%
        {nam2018}
\bibfield{author}{\bibinfo{person}{Yunseong Nam}, \bibinfo{person}{Neil~J.
  Ross}, \bibinfo{person}{Yuan Su}, \bibinfo{person}{Andrew~M. Childs}, {and}
  \bibinfo{person}{Dmitri Maslov}.} \bibinfo{year}{2018}\natexlab{}.
\newblock \showarticletitle{Automated optimization of large quantum circuits
  with continuous parameters}.
\newblock \bibinfo{journal}{\emph{npj Quantum Information}}
  \bibinfo{volume}{4}, \bibinfo{number}{1} (\bibinfo{date}{May}
  \bibinfo{year}{2018}).
\newblock
\urldef\tempurl%
\url{https://doi.org/10.1038/s41534-018-0072-4}
\showDOI{\tempurl}


\bibitem[Newman and Shi(2018)]%
        {newman2018}
\bibfield{author}{\bibinfo{person}{Michael Newman} {and}
  \bibinfo{person}{Yaoyun Shi}.} \bibinfo{year}{2018}\natexlab{}.
\newblock \showarticletitle{Limitations on Transversal Computation through
  Quantum Homomorphic Encryption}.
\newblock \bibinfo{journal}{\emph{Quantum Information and Computation}}
  \bibinfo{volume}{18}, \bibinfo{number}{11-12} (\bibinfo{date}{Sep}
  \bibinfo{year}{2018}).
\newblock
\urldef\tempurl%
\url{https://doi.org/10.26421/QIC18.11-12-3}
\showDOI{\tempurl}


\bibitem[Nielsen and Chuang(2010)]%
        {nielsen_chuang_2010}
\bibfield{author}{\bibinfo{person}{Michael~A. Nielsen} {and}
  \bibinfo{person}{Isaac~L. Chuang}.} \bibinfo{year}{2010}\natexlab{}.
\newblock \bibinfo{booktitle}{\emph{Quantum Computation and Quantum
  Information: 10th Anniversary Edition}}.
\newblock
\urldef\tempurl%
\url{https://doi.org/10.1017/CBO9780511976667}
\showDOI{\tempurl}


\bibitem[Olmedo and Díaz-Caro(2019)]%
        {olmedo2019}
\bibfield{author}{\bibinfo{person}{Federico Olmedo} {and}
  \bibinfo{person}{Alejandro Díaz-Caro}.} \bibinfo{year}{2019}\natexlab{}.
\newblock \bibinfo{title}{Runtime Analysis of Quantum Programs: A Formal
  Approach}.
\newblock
\newblock
\urldef\tempurl%
\url{https://doi.org/10.48550/arXiv.1911.11247}
\showDOI{\tempurl}
\showeprint[arxiv]{1911.11247}~[cs.LO]


\bibitem[Paykin et~al\mbox{.}(2017)]%
        {qwire}
\bibfield{author}{\bibinfo{person}{Jennifer Paykin}, \bibinfo{person}{Robert
  Rand}, {and} \bibinfo{person}{Steve Zdancewic}.}
  \bibinfo{year}{2017}\natexlab{}.
\newblock \showarticletitle{QWIRE: A Core Language for Quantum Circuits}. In
  \bibinfo{booktitle}{\emph{ACM SIGPLAN Symposium on Principles of Programming
  Languages}}.
\newblock
\urldef\tempurl%
\url{https://doi.org/10.1145/3009837.3009894}
\showDOI{\tempurl}


\bibitem[Proos and Zalka(2003)]%
        {proos2003}
\bibfield{author}{\bibinfo{person}{John Proos} {and} \bibinfo{person}{Christof
  Zalka}.} \bibinfo{year}{2003}\natexlab{}.
\newblock \showarticletitle{Shor's Discrete Logarithm Quantum Algorithm for
  Elliptic Curves}.
\newblock \bibinfo{journal}{\emph{Quantum Information and Computation}}
  \bibinfo{volume}{3}, \bibinfo{number}{4} (\bibinfo{date}{Jul}
  \bibinfo{year}{2003}).
\newblock
\urldef\tempurl%
\url{https://doi.org/10.26421/QIC3.4-3}
\showDOI{\tempurl}


\bibitem[{Qiskit Developers}(2021)]%
        {qiskit}
\bibfield{author}{\bibinfo{person}{{Qiskit Developers}}.}
  \bibinfo{year}{2021}\natexlab{}.
\newblock \bibinfo{title}{Qiskit: An Open-source Framework for Quantum
  Computing}.
\newblock
\newblock
\urldef\tempurl%
\url{https://doi.org/10.5281/zenodo.2573505}
\showDOI{\tempurl}


\bibitem[{QuiZX Developers}(2022)]%
        {quizx}
\bibfield{author}{\bibinfo{person}{{QuiZX Developers}}.}
  \bibinfo{year}{2022}\natexlab{}.
\newblock \bibinfo{title}{QuiZX: a quick Rust port of PyZX}.
\newblock \bibinfo{howpublished}{\url{https://github.com/Quantomatic/quizx}}.
\newblock


\bibitem[Rebentrost et~al\mbox{.}(2018)]%
        {rebentrost2018}
\bibfield{author}{\bibinfo{person}{Patrick Rebentrost},
  \bibinfo{person}{Brajesh Gupt}, {and} \bibinfo{person}{Thomas~R. Bromley}.}
  \bibinfo{year}{2018}\natexlab{}.
\newblock \showarticletitle{Quantum computational finance: Monte Carlo pricing
  of financial derivatives}.
\newblock \bibinfo{journal}{\emph{Phys. Rev. A}} \bibinfo{volume}{98},
  \bibinfo{number}{2} (\bibinfo{date}{Aug} \bibinfo{year}{2018}).
\newblock
\urldef\tempurl%
\url{https://doi.org/10.1103/physreva.98.022321}
\showDOI{\tempurl}


\bibitem[Reiher et~al\mbox{.}(2017)]%
        {reiher2017}
\bibfield{author}{\bibinfo{person}{Markus Reiher}, \bibinfo{person}{Nathan
  Wiebe}, \bibinfo{person}{Krysta~M. Svore}, \bibinfo{person}{Dave Wecker},
  {and} \bibinfo{person}{Matthias Troyer}.} \bibinfo{year}{2017}\natexlab{}.
\newblock \showarticletitle{Elucidating reaction mechanisms on quantum
  computers}.
\newblock \bibinfo{journal}{\emph{Proceedings of the National Academy of
  Sciences}} \bibinfo{volume}{114}, \bibinfo{number}{29} (\bibinfo{date}{Jul}
  \bibinfo{year}{2017}).
\newblock
\urldef\tempurl%
\url{https://doi.org/10.1073/pnas.1619152114}
\showDOI{\tempurl}


\bibitem[Rines and Chuang(2018)]%
        {rines2018}
\bibfield{author}{\bibinfo{person}{Rich Rines} {and} \bibinfo{person}{Isaac
  Chuang}.} \bibinfo{year}{2018}\natexlab{}.
\newblock \bibinfo{title}{High Performance Quantum Modular Multipliers}.
\newblock
\newblock
\urldef\tempurl%
\url{https://doi.org/10.48550/arXiv.1801.01081}
\showDOI{\tempurl}
\showeprint[arxiv]{1801.01081}~[quant-ph]


\bibitem[Sanders et~al\mbox{.}(2020)]%
        {sanders2020}
\bibfield{author}{\bibinfo{person}{Yuval~R. Sanders},
  \bibinfo{person}{Dominic~W. Berry}, \bibinfo{person}{Pedro~C.S. Costa},
  \bibinfo{person}{Louis~W. Tessler}, \bibinfo{person}{Nathan Wiebe},
  \bibinfo{person}{Craig Gidney}, \bibinfo{person}{Hartmut Neven}, {and}
  \bibinfo{person}{Ryan Babbush}.} \bibinfo{year}{2020}\natexlab{}.
\newblock \showarticletitle{Compilation of Fault-Tolerant Quantum Heuristics
  for Combinatorial Optimization}.
\newblock \bibinfo{journal}{\emph{{PRX} Quantum}} \bibinfo{volume}{1},
  \bibinfo{number}{2} (\bibinfo{date}{Nov} \bibinfo{year}{2020}).
\newblock
\urldef\tempurl%
\url{https://doi.org/10.1103/prxquantum.1.020312}
\showDOI{\tempurl}


\bibitem[Seidel et~al\mbox{.}(2022)]%
        {seidel2022}
\bibfield{author}{\bibinfo{person}{Raphael Seidel}, \bibinfo{person}{Sebastian
  Bock}, \bibinfo{person}{Nikolay Tcholtchev}, {and} \bibinfo{person}{Manfred
  Hauswirth}.} \bibinfo{year}{2022}\natexlab{}.
\newblock \showarticletitle{Qrisp: A Framework for Compilable High-Level
  Programming of Gate-Based Quantum Computers}. In
  \bibinfo{booktitle}{\emph{International Workshop on Programming Languages for
  Quantum Computing}}.
\newblock


\bibitem[Selinger(2004)]%
        {selinger2004}
\bibfield{author}{\bibinfo{person}{Peter Selinger}.}
  \bibinfo{year}{2004}\natexlab{}.
\newblock \showarticletitle{Towards a quantum programming language}.
\newblock \bibinfo{journal}{\emph{Mathematical Structures in Computer Science}}
   \bibinfo{volume}{14} (\bibinfo{date}{Aug} \bibinfo{year}{2004}).
\newblock
\urldef\tempurl%
\url{https://doi.org/10.1017/S0960129504004256}
\showDOI{\tempurl}


\bibitem[Selinger(2013)]%
        {selinger2013}
\bibfield{author}{\bibinfo{person}{Peter Selinger}.}
  \bibinfo{year}{2013}\natexlab{}.
\newblock \showarticletitle{Quantum circuits of $T$-depth one}.
\newblock \bibinfo{journal}{\emph{Phys. Rev. A}}  \bibinfo{volume}{87}
  (\bibinfo{date}{Apr} \bibinfo{year}{2013}).
\newblock
Issue 4.
\urldef\tempurl%
\url{https://doi.org/10.1103/PhysRevA.87.042302}
\showDOI{\tempurl}


\bibitem[Shi(2003)]%
        {shi2002}
\bibfield{author}{\bibinfo{person}{Yaoyun Shi}.}
  \bibinfo{year}{2003}\natexlab{}.
\newblock \showarticletitle{Both Toffoli and Controlled-NOT Need Little Help to
  Do Universal Quantum Computing}.
\newblock \bibinfo{journal}{\emph{Quantum Information and Computation}}
  \bibinfo{volume}{3}, \bibinfo{number}{1} (\bibinfo{date}{Jan}
  \bibinfo{year}{2003}).
\newblock
\urldef\tempurl%
\url{https://doi.org/10.26421/QIC3.1-7}
\showDOI{\tempurl}


\bibitem[Shor(1997)]%
        {shor1997}
\bibfield{author}{\bibinfo{person}{Peter~W. Shor}.}
  \bibinfo{year}{1997}\natexlab{}.
\newblock \showarticletitle{Polynomial-Time Algorithms for Prime Factorization
  and Discrete Logarithms on a Quantum Computer}.
\newblock \bibinfo{journal}{\emph{SIAM J. Comput.}} \bibinfo{volume}{26},
  \bibinfo{number}{5} (\bibinfo{date}{Oct} \bibinfo{year}{1997}).
\newblock
\urldef\tempurl%
\url{https://doi.org/10.1137/S0097539795293172}
\showDOI{\tempurl}


\bibitem[Sivarajah et~al\mbox{.}(2020)]%
        {sivarajah2021}
\bibfield{author}{\bibinfo{person}{Seyon Sivarajah}, \bibinfo{person}{Silas
  Dilkes}, \bibinfo{person}{Alexander Cowtan}, \bibinfo{person}{Will Simmons},
  \bibinfo{person}{Alec Edgington}, {and} \bibinfo{person}{Ross Duncan}.}
  \bibinfo{year}{2020}\natexlab{}.
\newblock \showarticletitle{t|ket⟩: a retargetable compiler for NISQ
  devices}.
\newblock \bibinfo{journal}{\emph{Quantum Science and Technology}}
  \bibinfo{volume}{6}, \bibinfo{number}{1} (\bibinfo{year}{2020}).
\newblock
\urldef\tempurl%
\url{https://doi.org/10.1088/2058-9565/ab8e92}
\showDOI{\tempurl}


\bibitem[Steiger et~al\mbox{.}(2018)]%
        {steiger2018}
\bibfield{author}{\bibinfo{person}{Damian~S. Steiger}, \bibinfo{person}{Thomas
  Häner}, {and} \bibinfo{person}{Matthias Troyer}.}
  \bibinfo{year}{2018}\natexlab{}.
\newblock \showarticletitle{{ProjectQ}: an open source software framework for
  quantum computing}.
\newblock \bibinfo{journal}{\emph{Quantum}}  \bibinfo{volume}{2}
  (\bibinfo{date}{Jan} \bibinfo{year}{2018}).
\newblock
\urldef\tempurl%
\url{https://doi.org/10.22331/q-2018-01-31-49}
\showDOI{\tempurl}


\bibitem[Suchara et~al\mbox{.}(2013)]%
        {suchara2013}
\bibfield{author}{\bibinfo{person}{Martin Suchara}, \bibinfo{person}{John
  Kubiatowicz}, \bibinfo{person}{Arvin Faruque}, \bibinfo{person}{Frederic~T.
  Chong}, \bibinfo{person}{Ching-Yi Lai}, {and} \bibinfo{person}{Gerardo Paz}.}
  \bibinfo{year}{2013}\natexlab{}.
\newblock \showarticletitle{QuRE: The Quantum Resource Estimator toolbox}. In
  \bibinfo{booktitle}{\emph{IEEE International Conference on Computer Design}}.
\newblock
\urldef\tempurl%
\url{https://doi.org/10.1109/ICCD.2013.6657074}
\showDOI{\tempurl}


\bibitem[Svore et~al\mbox{.}(2018)]%
        {qsharp}
\bibfield{author}{\bibinfo{person}{Krysta Svore}, \bibinfo{person}{Martin
  Roetteler}, \bibinfo{person}{Alan Geller}, \bibinfo{person}{Matthias Troyer},
  \bibinfo{person}{John Azariah}, \bibinfo{person}{Christopher Granade},
  \bibinfo{person}{Bettina Heim}, \bibinfo{person}{Vadym Kliuchnikov},
  \bibinfo{person}{Mariia Mykhailova}, {and} \bibinfo{person}{Andres Paz}.}
  \bibinfo{year}{2018}\natexlab{}.
\newblock \showarticletitle{Q\#: Enabling Scalable Quantum Computing and
  Development with a High-level DSL}. In \bibinfo{booktitle}{\emph{Real World
  Domain Specific Languages Workshop}}.
\newblock
\urldef\tempurl%
\url{https://doi.org/10.1145/3183895.3183901}
\showDOI{\tempurl}


\bibitem[Takita et~al\mbox{.}(2016)]%
        {takita2016}
\bibfield{author}{\bibinfo{person}{Maika Takita}, \bibinfo{person}{A.~D.
  C\'orcoles}, \bibinfo{person}{Easwar Magesan}, \bibinfo{person}{Baleegh
  Abdo}, \bibinfo{person}{Markus Brink}, \bibinfo{person}{Andrew Cross},
  \bibinfo{person}{Jerry~M. Chow}, {and} \bibinfo{person}{Jay~M. Gambetta}.}
  \bibinfo{year}{2016}\natexlab{}.
\newblock \showarticletitle{Demonstration of Weight-Four Parity Measurements in
  the Surface Code Architecture}.
\newblock \bibinfo{journal}{\emph{Phys. Rev. Letters}}  \bibinfo{volume}{117}
  (\bibinfo{date}{Nov} \bibinfo{year}{2016}).
\newblock
Issue 21.
\urldef\tempurl%
\url{https://doi.org/10.1103/PhysRevLett.117.210505}
\showDOI{\tempurl}


\bibitem[Voichick et~al\mbox{.}(2023)]%
        {voichick2023}
\bibfield{author}{\bibinfo{person}{Finn Voichick}, \bibinfo{person}{Liyi Li},
  \bibinfo{person}{Robert Rand}, {and} \bibinfo{person}{Michael Hicks}.}
  \bibinfo{year}{2023}\natexlab{}.
\newblock \showarticletitle{Qunity: A Unified Language for Quantum and
  Classical Computing}. In \bibinfo{booktitle}{\emph{ACM SIGPLAN Symposium on
  Principles of Programming Languages}}.
\newblock
\urldef\tempurl%
\url{https://doi.org/10.1145/3571225}
\showDOI{\tempurl}


\bibitem[Xu et~al\mbox{.}(2023)]%
        {xu2023}
\bibfield{author}{\bibinfo{person}{Amanda Xu}, \bibinfo{person}{Abtin Molavi},
  \bibinfo{person}{Lauren Pick}, \bibinfo{person}{Swamit Tannu}, {and}
  \bibinfo{person}{Aws Albarghouthi}.} \bibinfo{year}{2023}\natexlab{}.
\newblock \showarticletitle{Synthesizing Quantum-Circuit Optimizers}. In
  \bibinfo{booktitle}{\emph{ACM SIGPLAN Conference on Programming Language
  Design and Implementation}}.
\newblock
\urldef\tempurl%
\url{https://doi.org/10.1145/3591254}
\showDOI{\tempurl}


\bibitem[Xu et~al\mbox{.}(2022)]%
        {xu2022}
\bibfield{author}{\bibinfo{person}{Mingkuan Xu}, \bibinfo{person}{Zikun Li},
  \bibinfo{person}{Oded Padon}, \bibinfo{person}{Sina Lin},
  \bibinfo{person}{Jessica Pointing}, \bibinfo{person}{Auguste Hirth},
  \bibinfo{person}{Henry Ma}, \bibinfo{person}{Jens Palsberg},
  \bibinfo{person}{Alex Aiken}, \bibinfo{person}{Umut~A. Acar}, {and}
  \bibinfo{person}{Zhihao Jia}.} \bibinfo{year}{2022}\natexlab{}.
\newblock \showarticletitle{Quartz: Superoptimization of Quantum Circuits}. In
  \bibinfo{booktitle}{\emph{ACM SIGPLAN Conference on Programming Language
  Design and Implementation}}.
\newblock
\urldef\tempurl%
\url{https://doi.org/10.1145/3519939.3523433}
\showDOI{\tempurl}


\bibitem[Ying et~al\mbox{.}(2012)]%
        {ying2012}
\bibfield{author}{\bibinfo{person}{Mingsheng Ying}, \bibinfo{person}{Nengkun
  Yu}, {and} \bibinfo{person}{Yuan Feng}.} \bibinfo{year}{2012}\natexlab{}.
\newblock \bibinfo{title}{Defining Quantum Control Flow}.
\newblock
\newblock
\urldef\tempurl%
\url{https://doi.org/10.48550/ARXIV.1209.4379}
\showDOI{\tempurl}
\showeprint[arxiv]{1209.4379}~[quant-ph]


\bibitem[Yuan and Carbin(2022)]%
        {tower}
\bibfield{author}{\bibinfo{person}{Charles Yuan} {and} \bibinfo{person}{Michael
  Carbin}.} \bibinfo{year}{2022}\natexlab{}.
\newblock \showarticletitle{Tower: Data Structures in Quantum Superposition}.
  In \bibinfo{booktitle}{\emph{ACM SIGPLAN Conference on Object-Oriented
  Programming, Systems, Languages, and Applications}}.
\newblock
\urldef\tempurl%
\url{https://doi.org/10.1145/3563297}
\showDOI{\tempurl}


\bibitem[Yuan and Carbin(2024)]%
        {artifact}
\bibfield{author}{\bibinfo{person}{Charles Yuan} {and} \bibinfo{person}{Michael
  Carbin}.} \bibinfo{year}{2024}\natexlab{}.
\newblock \bibinfo{booktitle}{\emph{{The T-Complexity Costs of Error Correction
  for Control Flow in Quantum Computation}}}.
\newblock
\urldef\tempurl%
\url{https://doi.org/10.5281/zenodo.10729070}
\showDOI{\tempurl}


\bibitem[Zeng et~al\mbox{.}(2011)]%
        {zeng2007}
\bibfield{author}{\bibinfo{person}{Bei Zeng}, \bibinfo{person}{Andrew Cross},
  {and} \bibinfo{person}{Isaac~L. Chuang}.} \bibinfo{year}{2011}\natexlab{}.
\newblock \showarticletitle{Transversality Versus Universality for Additive
  Quantum Codes}.
\newblock \bibinfo{journal}{\emph{IEEE Transactions on Information Theory}}
  \bibinfo{volume}{57}, \bibinfo{number}{9} (\bibinfo{year}{2011}).
\newblock
\urldef\tempurl%
\url{https://doi.org/10.1109/TIT.2011.2161917}
\showDOI{\tempurl}


\end{thebibliography}
\vfill
\clearpage

\appendix

\section{Effect of Variable Bit Width on \tgate{}-Complexity} \label{sec:bitwidth}

In practice, the bit width of variables in quantum programs can vary with the depth of control flow. In this work, we make the simplifying assumption that the bit width is constant for two reasons.

First, a large bit width (e.g. 32-bit integers) causes the \tgate{}-complexity of the programs we study to quickly be dominated by primitive arithmetic and memory operations. Though significant in practice and worthy of consideration, the costs of arithmetic and memory would obscure the asymptotic costs of control flow that we focus on. The costs of control flow would persist, but analyzing them would require us to use programs at significantly larger recursion depth.

Second, the bit width of condition variables and the presence of control flow contribute orthogonal, multiplicative costs to the \tgate{}-complexity of the program. For example, consider the program:

\begin{lstlisting}[numbers=none]
if x == 0 { s_1; s_2; ...; s_n; }
\end{lstlisting}

Here, the quantum \texttt{if}-statement incurs control bits for \lstinline{x == 0} for each statement \texttt{s} in the body. If \texttt{x} is $k$ bits wide, the overall complexity incurred would be $O(nk)$ across the $n$ statements in the \texttt{if}-body. Though significant, the factor of $k$ here is not unique to control flow and is simply the multiplicative factor incurred by operating over large data types in any quantum program.

\section{Full Language Semantics} \label{sec:full-semantics}

In this section, we present the full semantics of the Tower language as studied in this work.

\subsection{Type System} \label{sec:full-type-system}

In \Cref{fig:core-val-types-full,fig:core-exp-types-full,fig:core-stmt-ok}, we define the typing judgments for values, expressions, and statements respectively. Note that the context $\ctx$ is ordered and permits multiple distinct type bindings for the same variable $x$, the most recently inserted of which shadows previous instances.

The definitions of the typing judgments are identical to \citet[Section 4.2]{tower} except for two changes. The first change is to permit a variable to be re-declared in the same scope, which is a situation that arises due to the program-level optimizations we present. The second change is to define typing of $\sHadamard{x}$, in a manner consistent with existing constructs.

\begin{figure}[!htb]
\begin{mathpar}
\inferrule[TV-Var]{x \notin \ctx'}{\hastype{\ctx, x : \type, \ctx'}{x}{\type}}

\inferrule[TV-Unit]{\vphantom{\ctx}}{\hastype{\ctx}{\vUnit}{\tUnit}}

\inferrule[TV-Pair]{\hastype{\ctx}{x_1}{\type_1} \\ \hastype{\ctx}{x_2}{\type_2}}{\hastype{\ctx}{\ePair{x_1}{x_2}}{\tPair{\type_1}{\type_2}}}

\inferrule[TV-Num]{\vphantom{\ctx}}{\hastype{\ctx}{\vNum{n}}{\tUInt}}

\inferrule[TV-Bool]{b \in \{\vTrue, \vFalse\}}{\hastype{\ctx}{b}{\tBool}}

\inferrule[TV-Null]{\vphantom{\ctx}}{\hastype{\ctx}{\vNull{\type}}{\tPtr{\type}}}

\inferrule[TV-Ptr]{\vphantom{\ctx}}{\hastype{\ctx}{\vPtr{\type}{p}}{\tPtr{\type}}}
\end{mathpar}
\caption{Typing rules for values in Tower.} \label{fig:core-val-types-full}
\end{figure}

\begin{figure}[!htb]
\begin{mathpar}
\inferrule[TE-Val]{\hastype{\ctx}{\vValue}{\type}}{\hastype{\ctx}{\vValue}{\type}}

\inferrule[TE-Proj]{\hastype{\ctx}{x}{\tPair{\type_1}{\type_2}}}{\hastype{\ctx}{\eProj{i}{x}}{\type_i}}

\inferrule[TE-Not]{\hastype{\ctx}{x}{\tBool}}{\hastype{\ctx}{\eUnop{\oNot}{x}}{\tBool}}

\inferrule[TE-Test]{\hastype{\ctx}{x}{\type} \\ \type \in \{\tUInt, \tPtr{\type'}\}}{\hastype{\ctx}{\eUnop{\oTest}{x}}{\tBool}}

\inferrule[TE-Lop]{\hastype{\ctx}{x_1}{\tBool} \\ \hastype{\ctx}{x_2}{\tBool} \\ bop \in \{\oAnd, \oOr\}}{\hastype{\ctx}{\eBinop{x_1}{bop}{x_2}}{\tBool}}

\inferrule[TE-Aop]{\hastype{\ctx}{x_1}{\tUInt} \\ \hastype{\ctx}{x_2}{\tUInt} \\ bop \in \{\oAdd, \oSub, \oMul\}}{\hastype{\ctx}{\eBinop{x_1}{bop}{x_2}}{\tUInt}}
\end{mathpar}
\caption{Typing rules for expressions in Tower.} \label{fig:core-exp-types-full}
\end{figure}

\begin{figure}
\resizebox{\textwidth}{!}{
\parbox{1.1\textwidth}{
\begin{mathpar}
\inferrule[S-Skip]{\vphantom{\ctx}}{\stmtok{\ctx}{\sSkip}{\ctx}}

\inferrule[S-Seq]{\stmtok{\ctx}{\sStmt_1}{\ctx'} \\ \stmtok{\ctx'}{\sStmt_2}{\ctx''}}{\stmtok{\ctx}{\sSeq{\sStmt_1}{\sStmt_2}}{\ctx''}}

\inferrule[S-Assign]{\hastype{\ctx}{\eExp}{\type}}{\stmtok{\ctx}{\sBind{x}{e}}{\ctx, x : \type}}

\inferrule[S-UnAssign]{\hastype{\ctx}{\eExp}{\type} \\ x \notin \ctx'}{\stmtok{\ctx, x : \type, \ctx'}{\sUnbind{x}{e}}{\ctx, \ctx'}}

\inferrule[S-Hadamard]{\vphantom{\ctx}}{\stmtok{\ctx, x : \tBool}{\sHadamard{x}}{\ctx, x : \tBool}}

\inferrule[S-Swap]{\hastype{\ctx}{x_1}{\type} \\ \hastype{\ctx}{x_2}{\type}}{\stmtok{\ctx}{\sSwap{x_1}{x_2}}{\ctx}}

\inferrule[S-MemSwap]{\hastype{\ctx}{x_1}{\tPtr{\type}} \\ \hastype{\ctx}{x_2}{\type}}{\stmtok{\ctx}{\sMemSwap{x_1}{x_2}}{\ctx}}

\inferrule[S-If]{\stmtok{\ctx}{\sStmt}{\ctx'} \\ \hastype{\ctx}{x}{\tBool} \\ x \notin \modified{\sStmt} \\ \domain{\ctx} \subseteq \domain{\ctx'}}{\stmtok{\ctx}{\sIf{x}{s}}{\ctx'}}
\end{mathpar}
\begin{alignat*}{5}
\modified{\sSkip} &= \emptyset \quad & \modified{\sSwap{x_1}{x_2}} &= \{x_1, x_2\} \\
\modified{\sSeq{s_1}{s_2}} &= \modified{s_1} \cup \modified{s_2} \quad & \modified{\sMemSwap{x_1}{x_2}} &= \{x_2\} \\
\modified{\sBind{x}{e}} &= \modified{\sUnbind{x}{e}} = \modified{\sHadamard{x}} = \{x\} \quad & \modified{\sIf{x}{s}} &= \modified{s}
\end{alignat*}
}}
\setlength{\abovecaptionskip}{5pt}
\caption{Well-formation rules for statements in Tower.} \label{fig:core-stmt-ok}
\end{figure}

\subsection{Circuit Semantics} \label{sec:full-circuit-semantics}

In \Cref{fig:core-circuit}, we present the circuit semantics of programs.
The definition is identical to \citet[Section 4.6]{tower} except for two changes. The first is to allocate a re-declared variable to the same qubits as the original, and the second is to define $\sHadamard{x}$ as simply the Hadamard gate.

The semantics of Tower uses a quantum random-access memory (qRAM) gate that enables data to be addressed in superposition. It is defined as the unitary operation that maps~\citep{ambainis2004}:
\[
\ket{i, b, z_1, \ldots, z_m} \mapsto \ket{i, z_i, z_1, \ldots, z_{i-1}, b, z_{i+1}, \ldots, z_m}
\]
where $i$ is a $k$-bit integer, $b$ is a bit, and $\ket{z_i, \ldots, z_m}$ is the qRAM --- an array of $m$ qubits. The effect of this gate is to swap the data at position $i$ in the array $z$ with the data in the register $b$. The semantics generalizes this gate to multiple qubits, such that $b$ and each $z_i$ is a $k$-bit string.

The semantics is specified in terms of a \emph{register file} $\reg$ mapping variables to values and a \emph{memory} $\mem$ mapping addresses to values. The register file $\reg$ corresponds to the main quantum registers over which we may perform arbitrary gates, while the memory $\mem$ corresponds to the qRAM.
In each circuit, a wire depicted as \varwire{} denotes a $k$-bit register representing an individual program value. A wire depicted as \varswire{} denotes a collection of such values, such as $\reg$ and $\mem$. A circuit fragment may be expanded to operate on the entire program state by tensor product with the identity gate.

An expression $\eExp$ is lifted to a unitary gate $U_\eExp$ that XORs the value to which $\eExp$ evaluates into a register $x$. Each well-formed expression -- constant value, variable reference, projection from a pair, integer arithmetic, and Boolean logic -- can be implemented as a unitary gate in this manner. For details on the semantics of expressions, please see \citet[Section 4]{tower}.

\newsavebox{\skipcircuit}
\savebox{\skipcircuit}{
\begin{quantikz}[classical gap=0.075cm]
    \lstick{$\ket{\reg, \mem}$} \setwiretype{b} & \rstick{$\ket{\reg, \mem}$}
\end{quantikz}
}
\newsavebox{\assigncircuit}
\savebox{\assigncircuit}{
\begin{quantikz}[classical gap=0.075cm]
    \lstick{$\ket{\reg}$} \setwiretype{b} & \gate[wires=2]{U_{\eExp}} & \rstick{$\ket{\reg}$} \\
    \lstick{$\ket{x'}$} & \qwbundle[style={xshift=-0.25em}]{k} & \rstick{$\ket{x \oplus x'}$} \qwbundle{k}
\end{quantikz}
}
\newsavebox{\unassigncircuit}
\savebox{\unassigncircuit}{
\begin{quantikz}[classical gap=0.075cm]
    \lstick{$\ket{\reg}$} \setwiretype{b} & \gate[wires=2]{U_{\eExp}^\dagger} & \rstick{$\ket{\reg}$} \\
    \lstick{$\ket{x \oplus x'}$} & \qwbundle[style={xshift=-0.25em}]{k} & \rstick{$\ket{x'}$} \qwbundle{k}
\end{quantikz}
}
\newsavebox{\hadcircuit}
\savebox{\hadcircuit}{
\begin{quantikz}[classical gap=0.075cm]
    \lstick{$\ket{x}$} & \gate{H} & \rstick{$\ket{x'}$}
\end{quantikz}
}
\newsavebox{\seqcircuit}
\savebox{\seqcircuit}{
\begin{quantikz}[classical gap=0.075cm]
    \lstick{$\ket{\reg, \mem}$} \setwiretype{b} & \gate{\circuit{\sStmt_1}} & \gate{\circuit{\sStmt_2}} & \rstick{$\ket{\reg', \mem'}$}
\end{quantikz}
}
\newsavebox{\swapcircuit}
\savebox{\swapcircuit}{
\begin{quantikz}[classical gap=0.075cm]
    \lstick{$\ket{x_1}$} &[1em] \swap{1}\qwbundle[style={xshift=-0.25em}]{k} &[1em] \rstick{$\ket{x_2}$} \qwbundle{k} \\
    \lstick{$\ket{x_2}$} &[1em] \targX{}\qwbundle[style={xshift=-0.25em}]{k} &[1em] \rstick{$\ket{x_1}$} \qwbundle{k}
\end{quantikz}
}
\newsavebox{\memswapcircuit}
\savebox{\memswapcircuit}{
\begin{quantikz}[classical gap=0.075cm]
    \lstick{$\ket{x_1}$} & \gate[wires=3]{\textrm{qRAM}}\qwbundle[style={xshift=-0.35em}]{k_1} & \rstick{$\ket{x_1}$} \qwbundle{k_1} \\
    \lstick{$\ket{x_2}$} & \qwbundle[style={xshift=-0.35em}]{k_2} & \rstick{$\ket{x'_2}$} \qwbundle{k_2} \\
    \lstick{$\ket{\mem}$} \setwiretype{b} & & \rstick{$\ket{\mem'}$}
\end{quantikz}
}
\newsavebox{\ifcircuit}
\savebox{\ifcircuit}{
\begin{quantikz}[classical gap=0.075cm]
    \lstick{$\ket{x}$} & \ctrl{1} & \rstick{$\ket{x}$} \qw \\
    \lstick{$\ket{\reg, \mem}$} \setwiretype{b} & \gate{\circuit{s}} & \rstick{$\ket{\reg', \mem'}$}
\end{quantikz}
}
\begin{figure}
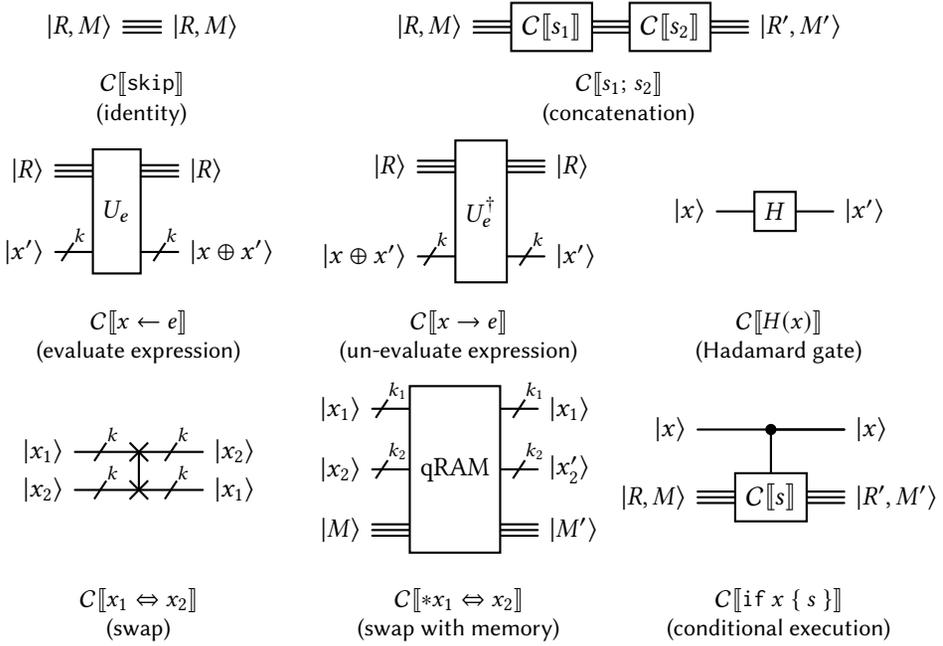

\captionsetup[subfigure]{labelformat=empty}
\captionsetup{justification=centering}
\centering
\begin{subfigure}[b]{.3\textwidth}
\centering
\makebox(100,20){\usebox{\skipcircuit}}
\caption{$\circuit{\sSkip}$ \\ (identity)}
\end{subfigure}
\begin{subfigure}[b]{.6\textwidth}
\centering
\makebox(100,20){\usebox{\seqcircuit}}
\caption{$\circuit{\sSeq{\sStmt_1}{\sStmt_2}}$ \\ (concatenation)}
\end{subfigure}
\begin{subfigure}[b]{.3\textwidth}
\centering
\makebox(100,60){\usebox{\assigncircuit}}
\caption{$\circuit{\sBind{x}{\eExp}}$ \\ (evaluate expression)}
\end{subfigure}
\begin{subfigure}[b]{.3\textwidth}
\centering
\makebox(100,60){\usebox{\unassigncircuit}}
\caption{$\circuit{\sUnbind{x}{\eExp}}$ \\ (un-evaluate expression)}
\end{subfigure}
\begin{subfigure}[b]{.3\textwidth}
\centering
\makebox(100,60){\usebox{\hadcircuit}}
\caption{$\circuit{\sHadamard{x}}$ \\ (Hadamard gate)}
\end{subfigure}
\begin{subfigure}[b]{.3\textwidth}
\centering
\makebox(100,75){\usebox{\swapcircuit}}
\caption{$\circuit{\sSwap{x_1}{x_2}}$ \\ (swap)}
\end{subfigure}
\begin{subfigure}[b]{.3\textwidth}
\centering
\makebox(100,75){\usebox{\memswapcircuit}}
\caption{$\circuit{\sMemSwap{x_1}{x_2}}$ \\ (swap with memory)}
\end{subfigure}
\begin{subfigure}[b]{.3\textwidth}
\centering
\makebox(100,75){\usebox{\ifcircuit}}
\caption{$\circuit{\sIf{x}{s}}$ \\ (conditional execution)}
\end{subfigure}
\caption{Definition of quantum circuit semantics of Tower.} \label{fig:core-circuit}
\end{figure}

\section{OCaml Implementation of Program-Level Optimizations} \label{sec:ocaml}

In \Cref{fig:ocaml}, we present the OCaml implementation of the program-level optimizations of \Cref{sec:optimizations}.

\begin{figure}
\begin{lstlisting}[language=caml]
let rec optimize_stmt s = match s with
| Sassign _ | Sunassign _ | Sswap _ | Smem_swap _ -> [ s ]
| Sif (x, ss) ->
  List.map ss ~f:(function
    | Swith (s1, s2) ->
      Swith (List.concat_map ~f:optimize_stmt s1, optimize_stmt @@ Sif (x, s2))
    | Sif (y, ss) ->
      let z = Symbol.new_symbol () in
      Swith ([ Sassign (Tbool, z, Ebop (Bland, x, y)) ], optimize_stmt @@ Sif (z, ss))
    | s -> Sif (x, optimize_stmt @@ s))
| Swith (s1, s2) ->
  [ Swith (List.concat_map ~f:optimize_stmt s1, List.concat_map ~f:optimize_stmt s2) ]
\end{lstlisting}
\caption{OCaml implementation of the program-level optimizations of \Cref{sec:optimizations}.} \label{fig:ocaml}
\end{figure}

\section{Register Allocation Case Study} \label{sec:implementation-case-study}

In this section, we illustrate the implementation challenges within the Spire register allocator that arise from performing the optimizations of \Cref{sec:optimizations} on a program. We present one challenge, and our solution to it, as a case study for quantum compiler developers.

\begin{figure}
\begin{subfigure}[t]{0.23\textwidth}
\begin{lstlisting}
if c {%\label{list:reg-alloc-before-if}%
  with {
    let x <- 1;%\label{list:reg-alloc-before-x}%
  } do {
    let x -> 1;
    let y <- 2;
    let x <- y-1;
  }
}
\end{lstlisting}
\caption{Original program.} \label{fig:reg-alloc-before}
\end{subfigure}%
\hspace*{\fill}%
\begin{subfigure}[t]{0.27\textwidth}
\begin{lstlisting}
if c {

  let x <- 1;   // r1%\label{list:reg-alloc-before-expand-1}%

  let x -> 1;   // r1%\label{list:reg-alloc-before-expand-2}%
  let y <- 2;   // r1%\label{list:reg-alloc-before-expand-3}%
  let x <- y-1; // r2%\label{list:reg-alloc-before-expand-4}%
  let x -> 1;   // r2%\label{list:reg-alloc-before-expand-5}%
}
\end{lstlisting}
\caption{Its expanded form.} \label{fig:reg-alloc-before-expand}
\end{subfigure}%
\hspace*{\fill}%
\begin{subfigure}[t]{0.23\textwidth}
\begin{lstlisting}
with {
  let x <- 1;
} do {
  if c {
    let x -> 1;
    let y <- 2;
    let x <- y-1;
  }
}
\end{lstlisting}
\caption{Optimized program.} \label{fig:reg-alloc-after}
\end{subfigure}%
\hspace*{\fill}%
\begin{subfigure}[t]{0.27\textwidth}
\begin{lstlisting}

let x <- 1;     // r1%\label{list:reg-alloc-after-expand-1}%

if c {
  let x -> 1;   // r1%\label{list:reg-alloc-after-expand-2}%
  let y <- 2;   // r1%\label{list:reg-alloc-after-expand-3}%
  let x <- y-1; // r2%\label{list:reg-alloc-after-expand-4}%
}
let x -> 1;     // ??%\label{list:reg-alloc-after-expand-5}%
\end{lstlisting}
\caption{Its expanded form.} \label{fig:reg-alloc-after-expand}
\end{subfigure}
\caption{An aggressive register allocation that becomes impermissible after conditional narrowing.} \label{fig:reg-alloc}
\end{figure}

\paragraph{Challenge}
The challenge is that while the conditional narrowing optimization always produces a program that compiles to a correct circuit under some appropriate allocation of variables to registers, the optimization causes certain aggressive allocations to become impermissible.

In \Cref{fig:reg-alloc-before}, we present a program that uses a \texttt{with}-\texttt{do} block inside an \texttt{if}-statement. This program is a candidate for the conditional narrowing optimization to move the \texttt{if} into the \texttt{do}-block and save control bits on line~\ref{list:reg-alloc-before-x}. Accordingly, in \Cref{fig:reg-alloc-after}, we depict the optimized program.

To compile a program to a quantum circuit, the Tower compiler first expands the construct $\sWith{s_1}{s_2}$ to the sequence $\sSeq{s_1}{\sSeq{s_2}{\reverse{s_1}}}$, where $\reverse{s_1}$ is the reverse of $s_1$ (\Cref{sec:tower}). In \Cref{fig:reg-alloc-before-expand} and \Cref{fig:reg-alloc-after-expand}, we depict the expansion of \Cref{fig:reg-alloc-before} and \Cref{fig:reg-alloc-after} respectively.

\paragraph{Register Allocation}
The next step in compilation is to assign program variables to registers. For simplicity, we assume there are two integer registers \texttt{r1} and \texttt{r2}. The comments in \Cref{fig:reg-alloc-before-expand} depict one possible allocation of the variables \texttt{x} and \texttt{y} to registers. On line~\ref{list:reg-alloc-before-expand-1}, \texttt{x} is assigned to \texttt{r1}.

On line~\ref{list:reg-alloc-before-expand-2}, \texttt{x} is uncomputed and \texttt{r1} is restored to zero. Because \texttt{r1} is now free, on line~\ref{list:reg-alloc-before-expand-3}, it is sound to assign \texttt{y} to \texttt{r1} and reuse it.
On line~\ref{list:reg-alloc-before-expand-4}, \texttt{x} is re-defined and assigned to the next free register, \texttt{r2}. Finally, on line~\ref{list:reg-alloc-before-expand-5}, \texttt{x} is uncomputed from the register \texttt{r2}, and register allocation is complete.

In general, it is essential for Tower to reuse registers, as occurred for \texttt{y} and \texttt{r1}, to ensure that the number of qubits used by a program does not blow up quickly. However, we will see next that such aggressive reuse can become impermissible after the conditional narrowing optimization.

\paragraph{Failed Allocation}
Moving to \Cref{fig:reg-alloc-after-expand}, the first four allocations on lines~\ref{list:reg-alloc-after-expand-1},~\ref{list:reg-alloc-after-expand-2},~\ref{list:reg-alloc-after-expand-3}, and~\ref{list:reg-alloc-after-expand-4} are identical to \Cref{fig:reg-alloc-before-expand}. In particular, \texttt{y} again reuses the register \texttt{r1}, and \texttt{x} is assigned to \texttt{r2}.

The problem comes on line~\ref{list:reg-alloc-after-expand-5}. It would be correct to uncompute \texttt{x} from \texttt{r2} in the case that \texttt{c} is true and the \texttt{if}-clause from lines~\ref{list:reg-alloc-after-expand-2} to~\ref{list:reg-alloc-after-expand-4} executed. However, if \texttt{c} is false, then they did not execute, meaning \texttt{x} still resides in \texttt{r1}. In general, uncomputing a variable from the wrong register corrupts the state of the program, meaning there is no correct way to complete this register allocation.

This example illustrates a general rule of thumb --- reusing registers in such a way that a variable is assigned different registers on different control flow paths can lead to a failed allocation.

\paragraph{Solution}
Our solution is simple and conservative. We define the set of \emph{affected} variables as those that are 1) used in a \texttt{with}-block, and 2) live at the beginning and end of the corresponding \texttt{do}-block. This definition covers \texttt{x} in \Cref{fig:reg-alloc-after}. Then, we add a condition that any affected variable must be assigned the same register at the beginning and end of the \texttt{do}-block, even if it is reallocated. This condition precludes the situation in \Cref{fig:reg-alloc-after} where \texttt{x} was assigned either \texttt{r1} or \texttt{r2}.

In principle, this fix increases register pressure and the risk of a spill, i.e.\ more qubits used by the compiled circuit. Any such increase is bounded by the number of variables re-defined in the same scope.
In practice, we have not observed significant increases in qubit usage --- typically, the same pool of registers is allocated in a different order, rather than spills occurring.

Other solutions are possible. For example, line~\ref{list:reg-alloc-before-expand-5} in \Cref{fig:reg-alloc-after} could be replaced with an \texttt{if} that uncomputes from either \texttt{r1} or \texttt{r2}, but doing so is counterproductive to conditional narrowing.

\section{Full Empirical MCX-Complexities and \tgate{}-Complexities} \label{sec:full-complexities}

In \Cref{tbl:benchmarks-full}, we report the full version of \Cref{tbl:benchmarks}, with all empirical figures reported.

\setlength\rotFPtop{0pt plus 1fil}
\begin{sidewaystable}
\centering
\caption{Full version of \Cref{tbl:benchmarks}, with all empirical figures reported.}
\resizebox{55em}{!}{%
\begin{tabular}{ l c c c c c c c c }
\toprule
                                    & \multicolumn{2}{c}{MCX-Complexity} & & \multicolumn{2}{c}{\tgate{}-Complexity Before Optimizations} & & \multicolumn{2}{c}{\tgate{}-Complexity After Optimizations} \\
\cmidrule{2-3} \cmidrule{5-6} \cmidrule{8-9}
Program                             & Predicted & Empirical & & Predicted & Empirical & & Predicted & Empirical \\
\midrule
List & & & & \\
\subname{}\texttt{length}          & $O(n)$ & $2246n + 32$ & & $O(n^2)$ & $15722n^2 + 19292n + 3934$ & & $O(n)$ & $12740n - 42$ \\
\subname{}\texttt{sum}             & $O(n)$ & $2642n + 32$ & & $O(n^2)$ & $18494n^2 + 19628n + 4298$ & & $O(n)$ & $13272n - 42$ \\
\subname{}\texttt{find\_pos}       & $O(n)$ & $2294n + 32$ & & $O(n^2)$ & $16058n^2 - 8820n + 6426$ & & $O(n)$ & $12740n - 42$ \\
\subname{}\texttt{remove}          & $O(n)$ & $4990n + 32$ & & $O(n^2)$ & $34930n^2 + 26376n + 10304$ & & $O(n)$ & $58912n - 12124$ \\
Queue & & & & \\
\subname{}\texttt{push\_back}      & $O(n)$ & $2864n + 32$ & & $O(n^2)$ & $20048n^2 + 11508n + 4634$ & & $O(n)$ & $46256n - 13006$ \\
\subname{}\texttt{pop\_front}      & $O(1)$ & $1452$ & & $O(1)$   & $8456$ & & $O(1)$ & $8456$ \\
String & & & & \\
\subname{}\texttt{is\_prefix}      & $O(n)$ & $4585n + 32$ & & $O(n^2)$ & $64190n^2 - 11529n + 6545$ & & $O(n)$ & $16758n - 42$ \\
\subname{}\texttt{num\_matching}   & $O(n)$ & $6052n + 5516$ & & $O(n^2)$ & $84728n^2 + 129360n + 59710$ & & $O(n)$ & $21826n + 18676$ \\
\subname{}\texttt{compare}         & $O(n)$ & $4633n + 32$ & & $O(n^2)$ & $97293n^2 + 10598n + 4781$ & & $O(n)$ & $17773n - 42$ \\
Set (radix tree) & & & & \\
\subname{}\texttt{insert}          & $O(d^2)$ & $70154d^2 + 299158d + 32$ & & $O(d^3)$ & $(3076192 / 3)d^3 + 5099374d^2 + (35136290 / 3)d + 155050$ & & $O(d^2)$ & $256914d^2 + 1413244d - 840$ \\
\subname{}\texttt{contains}        & $O(d^2)$ & $36680d^2 + 114553d + 32$ & & $O(d^3)$ & $1027040d^3 + 4142292d^2 + 3380461d + 26369$ & & $O(d^2)$ & $134064d^2 + 687008d - 42$ \\
\bottomrule
\end{tabular}%
}
\label{tbl:benchmarks-full}
\end{sidewaystable}

\section{Analysis of \tgate{} Gate and Qubit Costs of Optimizations} \label{sec:optimization-costs}

In \Cref{tbl:optimization-costs}, we report for each benchmark the number of \tgate{} gates in the circuit that are introduced by the additional uncomputation from conditional flattening. To compute the figures, we compiled each program to a circuit using a modified version of Spire that omits the additional uncomputation, and took the difference in \tgate{}-complexity between these modified circuits and those in \Cref{tbl:benchmarks}.

The results show that across all of the benchmarks in \Cref{tbl:optimization-costs} at recursion depth $n = 10$, 0 to 4.81\% (average 0.49\%) of the \tgate{} gates in the final compiled circuit correspond to the uncomputation that is introduced by conditional flattening. At depth $n = 2$, this figure is 0 to 2.85\% (average 0.30\%).

We also report the number of qubits used by each program, with and without Spire's optimizations, when compiled to a Clifford+Toffoli circuit by the process used in \Cref{sec:eval-circuit-optimizers}. To reiterate this process, we first used Spire to compile each program to an MCX circuit, and next used Feynman to convert MCX gates to Toffoli gates following the decomposition of \Cref{fig:mcx-decomposition}.

For the largest circuits \texttt{insert} and \texttt{contains} in \Cref{tbl:optimization-costs}, which are considerably larger than the \texttt{length} and \texttt{length-simplified} examples from \Cref{sec:eval-circuit-optimizers}, Feynman did not produce an output after one hour. In these cases, we report the results of our own consistent re-implementation of Feynman's MCX-decomposition following \Cref{fig:mcx-decomposition}. This re-implementation produces equal circuits to Feynman for all cases in \Cref{tbl:optimization-costs} where both tools produce outputs.

The results in \Cref{tbl:optimization-costs} show that at depths $n = 2$ and $n = 10$, Spire increases the qubit usage of the program by 1 qubit for the \texttt{push\_back} benchmark. It does not change qubit usage for \texttt{length}, \texttt{length-simplified}, \texttt{sum}, \texttt{remove}, and \texttt{pop\_front}. It decreases the qubit usage by 1 for \texttt{find\_pos}, \texttt{is\_prefix}, \texttt{num\_matching}, \texttt{compare}, and \texttt{contains}. For the \texttt{insert} benchmark, which composes the above functions as subroutines, Spire reduces qubit usage by 3 at depth $n = 2$ and 19 at $n = 10$.

\setlength\rotFPtop{0pt plus 1.1fil}
\begin{sidewaystable}
\centering
\newcommand{\timeout}{\textcolor{gray}{timeout}}
\caption{Cost incurred by Spire's optimizations, in terms of the number of \tgate{} gates introduced by the additional uncomputation from conditional flattening, and the number of qubits used in the compiled circuit after decomposing MCX to Toffoli gates using Feynman. We list each benchmark program from \Cref{tbl:benchmarks} at recursion depths $n = 10$ and $2$.\ An entry marked with $*$ denotes an instance where Feynman did not produce an output Clifford+Toffoli circuit after one hour, and where we report the output of our own consistent re-implementation of Feynman's behavior.}
\resizebox{45em}{!}{%
\begin{tabular}{ l c c c c c c c c }
\toprule
                                     & \multicolumn{3}{c}{\tgate{}-Complexity of Optimized Program} & & \multicolumn{3}{c}{Number of Qubits in Compiled Circuit} \\
\cmidrule{2-4} \cmidrule{6-8}
Program                              & Total & Uncomputation & \% Uncomputation & & Without Spire & With Spire & Difference \\
\midrule
Depth $n = 10$ \\
\subname{}\texttt{length}            & 127358   & 126   & 0.09\% & & 1128   & 1128   & 0 \\
\subname{}\texttt{length-simplified} & 2618     & 126   & 4.81\% & & 161    & 161    & 0 \\
\subname{}\texttt{sum}               & 132678   & 126   & 0.09\% & & 1436   & 1436   & 0 \\
\subname{}\texttt{find\_pos}         & 127358   & 126   & 0.09\% & & 1211   & 1210   & -1 \\
\subname{}\texttt{remove}            & 576996   & 252   & 0.04\% & & 1247   & 1247   & 0 \\
\subname{}\texttt{push\_back}        & 449554   & 126   & 0.02\% & & 843    & 844    & 1 \\
\subname{}\texttt{pop\_front}        & 8456     & 0     & 0.00\% & & 336    & 336    & 0 \\
\subname{}\texttt{is\_prefix}        & 167538   & 266   & 0.15\% & & 3627   & 3626   & -1 \\
\subname{}\texttt{num\_matching}     & 236936   & 294   & 0.12\% & & 4783   & 4782   & -1 \\
\subname{}\texttt{compare}           & 177688   & 406   & 0.22\% & & 4169   & 4168   & -1 \\
\subname{}\texttt{insert}            & 39823000 & 50498 & 0.12\% & & 647077* & 647058* & -19* \\
\subname{}\texttt{contains}          & 20276438 & 23226 & 0.11\% & & 313244* & 313243* & -1* \\
Average                              &          &       & 0.49\% \\
\midrule
Depth $n = 2$ \\
\subname{}\texttt{length}            & 25438   & 14   & 0.05\% & & 480   & 480   & 0 \\
\subname{}\texttt{length-simplified} & 490     & 14   & 2.85\% & & 41    & 41    & 0 \\
\subname{}\texttt{sum}               & 26502   & 14   & 0.05\% & & 548   & 548   & 0 \\
\subname{}\texttt{find\_pos}         & 25438   & 14   & 0.05\% & & 499   & 498   & -1 \\
\subname{}\texttt{remove}            & 105700  & 28   & 0.02\% & & 535   & 535   & 0 \\
\subname{}\texttt{push\_back}        & 79506   & 14   & 0.01\% & & 435   & 436   & 1 \\
\subname{}\texttt{pop\_front}        & 8456    & 0    & 0.00\% & & 336   & 336   & 0 \\
\subname{}\texttt{is\_prefix}        & 33474   & 42   & 0.12\% & & 819   & 818   & -1 \\
\subname{}\texttt{num\_matching}     & 62328   & 70   & 0.11\% & & 1351  & 1350  & -1 \\
\subname{}\texttt{compare}           & 35504   & 70   & 0.19\% & & 921   & 920   & -1 \\
\subname{}\texttt{insert}            & 3853304 & 3346 & 0.08\% & & 49341* & 49338* & -3* \\
\subname{}\texttt{contains}          & 1910230 & 1050 & 0.05\% & & 20588* & 20587* & -1* \\
Average                              &         &      & 0.30\% \\
\bottomrule
\end{tabular}%
}
\label{tbl:optimization-costs}
\end{sidewaystable}

\section{Additional Evaluation Results: Specific Optimizers} \label{sec:quartz-queso}

In this section, we present more evaluation details for Quartz~\citep{xu2022} and QUESO~\citep{xu2023}. We describe their running time behavior and present the partial results we obtained.

In \Cref{tbl:quartz-queso}, we present the results of Quartz and QUESO on the circuit for \texttt{length-simplified} at recursion depths 1 to 5. We report \tgate{}, $H$, and CNOT gate counts and the runtime of each optimizer.

\paragraph{Methodology Details}
Unlike the other circuit optimizers we evaluated, Quartz and QUESO do not by default accept input circuits that contain Toffoli gates. Instead, they accept controlled-controlled-$Z$ (CCZ) gates, to which Toffoli gates are closely related --- a Toffoli, or CCX, gate is simple to construct from one CCZ gate and two Hadamard gates. Thus, we translated each circuit to use CCZ gates, which does not change its \tgate{}-complexity, before invoking the optimizer.

Quartz optimizes circuits via a preprocessing phase that performs rotation merging and greedy CCZ decomposition followed by a search-based superoptimization phase.
We performed these steps following recommendations that the Quartz developers gave to us during correspondence. First, we invoked the preprocessing phase using the gate set of~\citet{nam2018} -- consisting of CNOT, $H$, $X$, and $R_z$ -- to ensure that Quartz's implementation of rotation merging operates as intended.
Second, before invoking the search-based phase, we converted each $R_z$ gate back to Clifford+$T$ gates following the identities $R_z(\pi/4) = T$, $R_z(-\pi/4) = T^\dagger$, $R_z(\pi/2) = S$, $R_z(-\pi/2) = S^\dagger$.

Both Quartz and QUESO require a rule file and associated command-line arguments to be provided to the optimizer, and we obtained the necessary rule files and configurations from their respective artifacts.
For Quartz, we used the provided \lstinline{gen_ecc_set} tool to generate the rule file \lstinline{3_2_5_complete_ECC_set.json},\footnote{The name of this file has since been changed by the Quartz developers to \lstinline{Clifford_T_5_3_complete_ECC_set.json}.} which we then provided to the optimizer using the \lstinline{--eqset} flag.
For QUESO, we invoked the \lstinline{EnumeratorPrune -g nam -q 3 -s 6} flags to generate the rule files \lstinline{rules_q3_s6_nam.txt} and \lstinline{rules_q3_s6_nam_symb.txt}, which we then provided to the optimizer using the \texttt{-r} and \texttt{-sr} flags respectively, along with the \texttt{-j "nam"} flag.

\paragraph{Running Time Behavior}
Unlike other circuit optimizers whose runtime is bounded a priori by the length of the input circuit, Quartz and QUESO use open-ended search strategies to discover possible circuit rewrites. The runtime of this process is bounded in practice only by an explicit timeout.
In our evaluation, we used a timeout of \SI{3600}{\second}, or one hour. Quartz uses one hour as its timeout value by default. For QUESO, we specified the timeout using the \texttt{-t 3600} flag.

For Quartz, the user-specifiable timeout affects only the search-based phase and not preprocessing (rotation merging and greedy CCZ decomposition). As shown in \Cref{tbl:quartz-queso}, preprocessing takes about 8 hours on the $n = 5$ circuit, and the search-based phase takes about 1 additional hour.

For possibly similar reasons, QUESO also takes longer than 1 hour on our benchmark circuits despite the specified 1 hour timeout. As shown in \Cref{tbl:quartz-queso}, QUESO takes about 1.5 hours to terminate at depth $n = 1$ and about 26 hours at depth $2$ when the timeout is set to 1 hour.

\paragraph{Results}
As shown in \Cref{tbl:quartz-queso}, the preprocessing phase of Quartz improves the \tgate{}-complexity of the circuit by 32\% at depth $n = 1$ and 37\% at depth $n = 5$. The scaling of \tgate{}-complexity is quadratic. Afterwards, running the search phase for one hour yields circuits that have the same number of $T$ gates and fewer $H$ and CNOT gates compared to those output by the preprocessing phase.

The Quartz developers informed us of possible explanations for why the impact of preprocessing on \tgate{} gates dominates over that of search for our benchmarks: ``The Toffoli decomposition (into 13 gates) is known to be optimal, so inside each CCZ gate, Quartz does not have any chance to optimize it further'', and ``Although T gates across multiple CCZ gates can be optimized together, it will need [equivalent circuit class] sets with more than 13 gates, which is beyond Quartz's scalability''.

After we contacted the Quartz developers, they provided to us a new version of Quartz that, according to the developers, was created after consideration of the Clifford+\tgate{} gate set. The new version adds to Quartz a native operation that translates an $R_z$ gate into \tgate{} and $S$ gates, and one that translates a CCZ gate into a fixed sequence of CNOT and $R_z$ gates. These new operations enable an option to bypass greedy CCZ decomposition and only run rotation merging in preprocessing.

This new version (v0.1.1) of Quartz uses less runtime to achieve \tgate{}-complexities that are equal to or slightly greater than those in \Cref{tbl:quartz-queso}, which were achieved by the version (v0.1.0) of Quartz available at the start of our experimentation.
We depict the results of Quartz v0.1.1 in \Cref{tbl:quartz-new}.

Meanwhile, QUESO improves the \tgate{}-complexity of the circuit by 20\% at depth $n = 1$ and 13\% at depth 2. Though we were unable to run QUESO for longer than 26 hours to collect additional data, these results suggest that the \tgate{}-complexity of the output is not asymptotically linear.

\begin{sidewaystable}
\caption{Optimization performance of Quartz~\citep{xu2022} and QUESO~\citep{xu2023} on the \texttt{length-simplified} program for recursion depths 1 to 5. Reported columns are the \tgate{}, $H$, and CNOT gate counts for the original and optimized circuits, and the time taken by each optimizer. Times reported are for a single run. For Quartz, the column ``preprocess only'' reports the gate counts and running time of the preprocessing phase of rotation merging and CCZ decomposition. The column ``preprocess + search'' reports the gate counts and total running time of preprocessing followed by search-based optimization.}
\centering
\resizebox{55em}{!}{%
\begin{tabular}{ c c c c c c c c S[table-format=5.2] c c c c S[table-format=5.2] c c c c S[table-format=5.2] }
\toprule
          & \multicolumn{3}{c}{Original Circuit} & & \multicolumn{4}{c}{Quartz (preprocessing only)} & & \multicolumn{4}{c}{Quartz (preprocessing + search)} & & \multicolumn{4}{c}{QUESO} \\
\cmidrule{2-4} \cmidrule{6-9} \cmidrule{11-14} \cmidrule{16-19}
Depth $n$ & \tgate{} & $H$ & CNOT & & \tgate{} & $H$ & CNOT & {Time (s)} & & \tgate{} & $H$ & CNOT & {Time (s)} & & \tgate{} & $H$ & CNOT & {Time (s)} \\
\midrule
1 & 392 & 112 & 346 & & 266 & 112 & 346 & 1.58 & & 266 & 68 & 272 & 3602.06 & & 312 & 68 & 332 & 5524.79 \\
2 & 1638 & 468 & 1448 & & 1114 & 468 & 1448 & 99.10 & & 1114 & 262 & 1130 & 3701.61 & & 1424 & 262 & 1420 & 93725.16 \\
3 & 4046 & 1156 & 3512 & & 2642 & 1156 & 3512 & 797.96 & & 2642 & 640 & 2638 & 4415.72       & & & & & \\
4 & 7854 & 2244 & 6776 & & 5010 & 2244 & 6776 & 5723.98 & & 5010 & 1210 & 5006 & 9411.60     & & & & & \\
5 & 13062 & 3732 & 11240 & & 8240 & 3732 & 11240 & 30064.91 & & 8240 & 1988 & 8190 & 33959.49 & & & & & \\
\bottomrule
\end{tabular}%
}
\label{tbl:quartz-queso}
\end{sidewaystable}

\begin{sidewaystable}
\caption{Optimization performance of the new version of Quartz (v0.1.1) the developers provided to us. The column ``RM only'' reports the gate counts and running time of rotation merging alone. The column ``RM + search'' reports gate counts and running time of rotation merging followed by search-based optimization. The column ``RM + CD + search'' reports those of rotation merging and greedy CCZ decomposition followed by search-based optimization.}
\centering
\resizebox{55em}{!}{%
\begin{tabular}{ c c c c c c c c S[table-format=2.2] c c c c S[table-format=4.2] c c c c S[table-format=5.2] }
\toprule
          & \multicolumn{3}{c}{Original Circuit} & & \multicolumn{4}{c}{Quartz (RM only)} & & \multicolumn{4}{c}{Quartz (RM + search)} & & \multicolumn{4}{c}{Quartz (RM + CD + search)} \\
\cmidrule{2-4} \cmidrule{6-9} \cmidrule{11-14} \cmidrule{16-19}
$n$ & \tgate{} & $H$ & CNOT & & \tgate{} & $H$ & CNOT & {Time (s)} & & \tgate{} & $H$ & CNOT & {Time (s)} & & \tgate{} & $H$ & CNOT & {Time (s)} \\
\midrule
1 & 392 & 112 & 346 & &
266 & 112 & 346 & 0.05 & &
266 & 68 & 292 & 3600.44 & &
266 & 68 & 272 & 3602.13 \\
2 & 1638 & 468 & 1448 & &
1114 & 468 & 1448 & 0.74 & &
1114 & 262 & 1160 & 3603.37 & &
1114 & 262 & 1128 & 3697.53 \\
3 & 4046 & 1156 & 3512 & &
2650 & 1156 & 3512 & 2.74 & &
2650 & 640 & 2780 & 3623.43 & &
2642 & 640 & 2634 & 4403.51 \\
4 & 7854 & 2244 & 6776 & &
5023 & 2244 & 6776 & 14.58 & &
5023 & 1210 & 5390 & 3677.66 & &
5010 & 1210 & 5004 & 9204.48 \\
5 & 13062 & 3732 & 11240 & &
8268 & 3732 & 11240 & 72.90 & &
8268 & 1988 & 9022 & 3834.36 & &
8240 & 1988 & 8188 & 33238.18 \\
\bottomrule
\end{tabular}%
}
\label{tbl:quartz-new}
\end{sidewaystable}

\section{Additional Evaluation Results: Synergy Effect} \label{sec:more-evaluation}

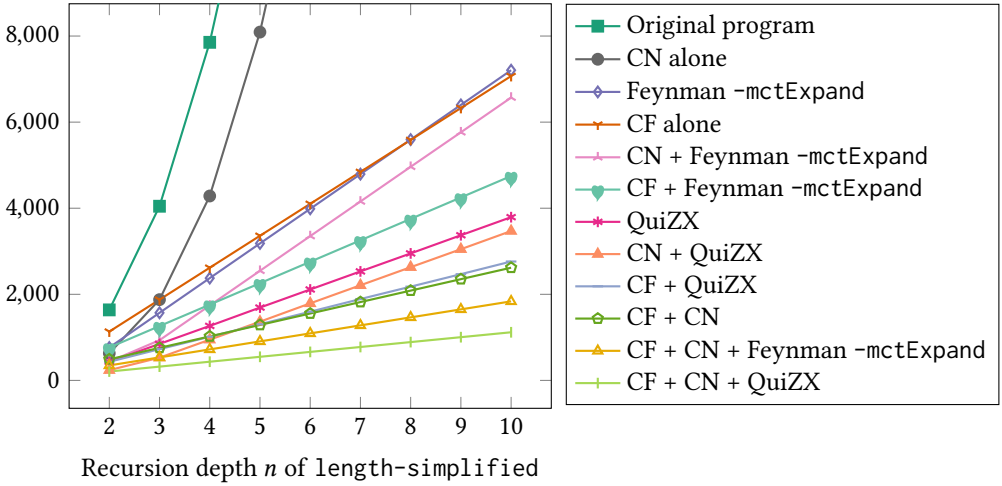
\begin{figure}
\begin{tikzpicture}
\begin{axis}[
    xlabel=Recursion depth $n$ of \texttt{length-simplified},
    scaled y ticks=false,
    ytick distance=2000,
    xtick distance=1,
    ymax=8750,
    width=8cm,
    legend pos=outer north east,
    legend cell align=left,
    domain=2:10,
    samples=100,
    every axis plot/.append style={thick}
]
\addplot[Dark2-A,-,mark=square*] plot coordinates {
    (2,1638)
    (3,4046)
    (4,7854)
    (5,13062)
    (6,19670)
    (7,27678)
    (8,37086)
    (9,47894)
    (10,60102)
};
\addlegendentry{Original program}
\addplot[Dark2-H,-,mark=*] plot coordinates {
    (2,630)
    (3,1876)
    (4,4284)
    (5,8092)
    (6,13300)
    (7,19908)
    (8,27916)
    (9,37324)
    (10,48132)
};
\addlegendentry{CN alone}
\addplot[Dark2-C,-,mark=diamond] plot coordinates {
    (2,764)
    (3,1570)
    (4,2376)
    (5,3182)
    (6,3986)
    (7,4792)
    (8,5594)
    (9,6400)
    (10,7206)
};
\addlegendentry{Feynman \texttt{-mctExpand}}
\addplot[Dark2-B,-,mark=Mercedes star flipped] plot coordinates {
    (2,1134)
    (3,1876)
    (4,2618)
    (5,3360)
    (6,4102)
    (7,4844)
    (8,5586)
    (9,6328)
    (10,7070)
};
\addlegendentry{CF alone}
\addplot[Set2-D,-,mark=Mercedes star] plot coordinates {
    (2,436)
    (3,934)
    (4,1740)
    (5,2546)
    (6,3352)
    (7,4156)
    (8,4962)
    (9,5764)
    (10,6570)
};
\addlegendentry{CN + Feynman \texttt{-mctExpand}}
\addplot[Set2-A,-,mark=heart] plot coordinates {
    (2,764)
    (3,1262)
    (4,1760)
    (5,2258)
    (6,2756)
    (7,3254)
    (8,3752)
    (9,4250)
    (10,4748)
};
\addlegendentry{CF + Feynman \texttt{-mctExpand}}
\addplot[Dark2-D,-,mark=asterisk] plot coordinates {
    (2,430)
    (3,848)
    (4,1266)
    (5,1692)
    (6,2112)
    (7,2532)
    (8,2950)
    (9,3372)
    (10,3792)
};
\addlegendentry{QuiZX}
\addplot[Set2-B,-,mark=triangle*] plot coordinates {
    (2,238)
    (3,528)
    (4,946)
    (5,1364)
    (6,1790)
    (7,2210)
    (8,2630)
    (9,3048)
    (10,3470)
};
\addlegendentry{CN + QuiZX}
\addplot[Set2-C,-,mark=-] plot coordinates {
    (2,432)
    (3,726)
    (4,1014)
    (5,1304)
    (6,1596)
    (7,1890)
    (8,2178)
    (9,2468)
    (10,2762)
};
\addlegendentry{CF + QuiZX}
\addplot[Dark2-E,-,mark=pentagon] plot coordinates {
    (2,490)
    (3,756)
    (4,1022)
    (5,1288)
    (6,1554)
    (7,1820)
    (8,2086)
    (9,2352)
    (10,2618)
};
\addlegendentry{CF + CN}
\addplot[Dark2-F,-,mark=triangle] plot coordinates {
    (2,348)
    (3,534)
    (4,720)
    (5,906)
    (6,1092)
    (7,1278)
    (8,1464)
    (9,1650)
    (10,1836)
};
\addlegendentry{CF + CN + Feynman \texttt{-mctExpand}}
\addplot[Set2-E,-,mark=|] plot coordinates {
    (2,206)
    (3,320)
    (4,434)
    (5,548)
    (6,662)
    (7,776)
    (8,890)
    (9,1004)
    (10,1118)
};
\addlegendentry{CF + CN + QuiZX}
\end{axis}
\end{tikzpicture}%
\caption{Synergy of individual program-level optimizations with Feynman and QuiZX.} \label{fig:ablation}
\end{figure}

In this section, we present experimental results demonstrating that the synergy effect from \Cref{sec:eval-program-level} occurs even when program-level optimizations are used individually rather than together.

In \Cref{fig:ablation}, we plot the \tgate{}-complexity of \texttt{length-simplified} after applying conditional flattening and conditional narrowing alone, together, and alongside Feynman. We observe:
\begin{itemize}
    \item Applying conditional narrowing and then Feynman achieves better results than either alone.
    \item The same pattern emerges with the conditional flattening optimization.
    \item Applying both program-level optimizations followed by Feynman achieves better results than applying each one individually followed by Feynman.
\end{itemize}

We also plot the \tgate{}-complexity achieved by QuiZX, with and without first applying the program-level optimizations individually or together, and observe the same patterns as for Feynman.

\section{Alternative Toffoli Optimization Strategies} \label{sec:optimization-alternatives}

Apart from \Cref{fig:toffoli-non-cancellability} itself, \citet{bravyi2021} give examples of rewrite rules over sequences longer than two adjacent Clifford+\tgate{} gates that can be used to optimize the decomposed form of adjacent Toffoli gates. Another approach is to use an alternative decomposition of the Toffoli gate, such as proposed by \citet{maslov2016,selinger2013}, whose structure is simpler to optimize than \Cref{fig:toffoli-non-cancellability}.

\end{document}
\endinput